\documentclass[12pt,latin9,utf8,reqno]{amsart}
\usepackage{geometry}
\usepackage{amsbsy}
\usepackage{amstext}
\usepackage{amsthm}
\usepackage{amssymb}
\usepackage{graphicx}
\usepackage[hyperfootnotes=false,unicode=true,
 bookmarks=false,
 breaklinks=false,pdfborder={0 0 1},backref=section,colorlinks=false]
 {hyperref}
\usepackage{caption}
\usepackage{subcaption}

\makeatletter
\numberwithin{equation}{section}
\numberwithin{figure}{section}
\theoremstyle{plain}
\newtheorem{thm}{\protect\theoremname}[section]
\newtheorem{ass}{Assumption}
\newtheorem{conjecture}[thm]{\protect\conjecturename}
\newtheorem{lem}[thm]{\protect\lemmaname}
\newtheorem{prop}[thm]{\protect\propositionname}
\newtheorem{cor}[thm]{\protect\corollaryname}
\theoremstyle{remark}
\newtheorem{rem}[thm]{\protect\remarkname}
\theoremstyle{definition}
\newtheorem{defn}[thm]{\protect\definitionname}

\@ifundefined{date}{}{\date{}}

\usepackage{amscd}
\usepackage{comment}
\usepackage{mathrsfs}
\usepackage{setspace}
\usepackage{color}

\DeclareMathOperator{\vol}{vol}
\DeclareMathOperator{\var}{Var}
\DeclareMathOperator{\spec}{spec}
\DeclareMathOperator{\tr}{tr}
\DeclareMathOperator{\Span}{span}
\DeclareMathOperator{\Ran}{Ran}
\DeclareMathOperator{\Bin}{Bin}
\DeclareMathOperator{\hess}{Hess}

\DeclareMathOperator{\conv}{conv}
\DeclareMathOperator{\w}{\textbf{w}}

\DeclareMathOperator{\M}{ind}

\newcommand\p{\mathrm{p}}%
\newcommand\bsigma{\boldsymbol{\sigma}}
\newcommand\ba{\boldsymbol{a}}
\newcommand\bH{\boldsymbol{H}}
\newcommand\EL{\E_{\textrm{loops}}}%
\newcommand\LL{L_{\textrm{loops}}}%
\newcommand{\vv}{{\vec{v}\,}}
\newcommand\Vas{\boldsymbol{V}_{\!\!\mathrm{as}}}
\newcommand\Vs{\boldsymbol{V}_{\!\!\mathrm{sym}}}

\providecommand{\conjecturename}{Conjecture}
\providecommand{\corollaryname}{Corollary}
\providecommand{\definitionname}{Definition}
\providecommand{\lemmaname}{Lemma}
\providecommand{\propositionname}{Proposition}
\providecommand{\remarkname}{Remark}
\providecommand{\theoremname}{Theorem}

\makeatother

\providecommand{\conjecturename}{Conjecture}
\providecommand{\corollaryname}{Corollary}
\providecommand{\definitionname}{Definition}
\providecommand{\lemmaname}{Lemma}
\providecommand{\propositionname}{Proposition}
\providecommand{\remarkname}{Remark}
\providecommand{\theoremname}{Theorem}

\begin{document}
\makeatletter
\providecommand*{\dd}{\@ifnextchar^{\DIfF}{\DIfF^{}}}
\def\DIfF^#1{
\mathop{\mathrm{\mathstrut d}}%
\nolimits^{#1}\gobblespace}
\def\gobblespace{
\futurelet\diffarg\opspace}
\def\opspace{
\let\DiffSpace\!
\ifx\diffarg(
\let\DiffSpace\relax
\else
\ifx\diffarg[
\let\DiffSpace\relax
\else
\ifx\diffarg\{
\let\DiffSpace\relax
\fi\fi\fi\DiffSpace}

\global\long\def\at{\left.\right|_{\Omega}}%
\global\long\def\set#1#2{\left\{  #1\,\,:\,\,#2\right\}  }%

\global\long\def\id{\textbf{I}}%

\global\long\def\sym{G_{\Gamma}}%

\global\long\def\A{\mathcal{A}}%

\global\long\def\G{\mathcal{G}}%


\global\long\def\L{\mathcal{L}}%

\global\long\def\D{\mathcal{D}}%

\global\long\def\U{\mathcal{U}}%

\global\long\def\Rv{\mathcal{R}^{\left(v\right)}}%

\global\long\def\R{\mathcal{R}}%

\global\long\def\d{\partial}%

\global\long\def\dg{\partial\Gamma}%

\global\long\def\do{\partial\Omega}%

\global\long\def\E{\mathcal{E}}%

\global\long\def\V{\mathcal{V}}%

\global\long\def\Vint{\mathcal{V}\setminus\partial\Gamma}%

\global\long\def\Vin{V_{\textrm{in}}}%

\global\long\def\la{\lambda}%



\global\long\def\Z{\mathbb{Z}}%

\global\long\def\R{\mathbb{R}}%

\global\long\def\C{\mathbb{C}}%

\global\long\def\N{\mathbb{N}}%

\global\long\def\Q{\mathbb{Q}}%

\global\long\def\msing{\Sigma^{\mathrm{sing}}}%

\global\long\def\mreg{\Sigma^{\mathrm{reg}}}%

\global\long\def\mgen{\Sigma^{\G}}%

\global\long\def\mg{\Sigma}%

\global\long\def\lap{\Delta}%

\global\long\def\na{\nabla}%

\global\long\def\opcl#1{\left[#1\right)}%

\global\long\def\clop#1{\left[#1\right)}%

\global\long\def\bs#1{\boldsymbol{#1}}%

\global\long\def\deg#1{\mathrm{deg}(#1)}%

\global\long\def\T{\mathbb{T}}%

\global\long\def\TE{\mathbb{T^{\left|\E\right|}}}%

\global\long\def\BGm{\mu_{\vec{l}}}%

\global\long\def\lv{{\boldsymbol{\ell}}}%
\global\long\def\lve{\ell_{e}}%
\global\long\def\lvj{\ell_{j}}%
\global\long\def\l{\ell}%

\global\long\def\ts{t_{S}}%

\global\long\def\Lv{\vec{L}}%

\global\long\def\av{\vec{\alpha}}%

\global\long\def\kv{{\vec{\kappa}}}%

\global\long\def\xv{{\vec{x}}}%

\global\long\def\Tv{\vec{\kappa}}%

\global\long\def\dL{d_{\vec{L}}}%

\global\long\def\sgn{\mathrm{sgn}}%

\global\long\def\undercom#1#2{\underset{_{#2}}{\underbrace{#1}}}%
\global\long\def\fr#1{\{#1\}_{2\pi}}%

\global\long\def\diag{\textrm{diag}}%

\global\long\def\as{\boldsymbol{v}}%
\global\long\def\mL{\Sigma_{\textrm{loops}}}
\global\long\def\mregL{\mreg_{\textrm{loops}}}

\global\long\def\ones{{\boldsymbol{1}}}%
\global\long\def\uone{\textrm{u}(1)}%
\global\long\def\Tw{\widetilde{\T}^{E}}
\global\long\def\bx{\boldsymbol{\xi}}%
\title{Universality of nodal count distribution in large metric graphs}
\author{Lior Alon, Ram Band and Gregory Berkolaiko}
\address{$^{1}${\small{}Lior Alon, School of Mathematics, Institute for Advanced
Study, Princeton, NJ 08540, USA. e-mail: lalon@ias.edu.il}}
\address{$^{2}${\small{}Ram Band, Department of Mathematics, Technion--Israel
Institute of Technology, Haifa 32000, Israel. e-mail: ramband@technion.ac.il}}
\address{$^{2}${\small{}Gregory Berkolaiko, Department of Mathematics, Texas A$\&$M university,College Station,
TX 77843-3368, USA. email: berko@math.tamu.edul}}
\begin{abstract}
  An eigenfunction of the Laplacian on a metric (quantum) graph has an
  excess number of zeros due to the graph's non-trivial topology.
  This number, called the nodal surplus, is an integer between 0 and
  the graph's first Betti number $\beta$.  We study the
  distribution of the nodal surplus values in the countably infinite
  set of the graph's eigenfunctions.  We conjecture that this
  distribution converges to Gaussian for any sequence of graphs of
  growing $\beta$.  We prove this conjecture for several special graph
  sequences and test it numerically for a variety of well-known graph
  families.  Accurate computation of the distribution is made possible
  by a formula expressing the nodal surplus distribution as an
  integral over a high-dimensional torus.
\end{abstract}

\maketitle

\tableofcontents

\section{Introduction}

Denoting by $\nu_n$ the number of nodal domains of the $n$-th
Laplacian eigenfunction, Courant's theorem
\cite{Cou_ngwg23,CourantHilbert_volume1} establishes the bound
$\nu_n \leq n$.  Here a ``nodal domain'' is a maximal connected
component of the underlying physical space where the eigenfunction does not vanish.  The theorem, originally stated for planar domains, is extremely robust and remains
valid for Laplacians on manifolds, with or without boundary, as
well as in numerous other settings, see \cite{KelSch_jst20} and
references therein.  It is also valid on discrete and metric graphs
\cite{DavGlaLeySta_laa01,GnuSmiWeb_wrm04}.

Pleijel \cite{Ple_cpam56} strengthened this bound for Dirichlet
Laplacians,\footnote{Only relatively recently it was extended to
  Neumann Laplacians \cite{Lena_afourier19,Pol_pams09}.} proving that
the ratio $\nu_n/n$ is asymptotically bounded away from 1.  More
recently, it was suggested \cite{BluGnuSmi_prl02} that the
distribution of this ratio carries important domain-specific
information.  Bogomolny and Schmit \cite{BS02} made an appealing
quantitative conjecture about the mean and variance of $\nu_n/n$ based
on an analogy with percolation.  Remarkably, extensive numerical
calculations \cite{BeliaevKereta13,Kon_thesis12,Nas_thesis11} revealed
statistically significant deviations from the Bogomolny-Schmit
prediction but only of relative size of less than $5\%$.  Mathematical
progress is being made within the Random Wave Model, proving that the
mean of $\nu_n/n$ is non-zero \cite{BelMcaMui_ptrf20,NazSob_ajm09,NazSod_JMPAG16}
and providing bounds on the variance \cite{NazSod_jmp20}.

In this paper, we focus on the setting of metric graphs.  Here, one
can equivalently study the \emph{number of zeros} of the $n$-th
eigenfunction, which we denote by $\phi_n$.  Assuming the eigefunction does not vanish on vertices (which happens generically \cite{BerLiu_jmaa17}) and for large enough $ n $, the two quantities are directly related by
\begin{equation}
  \label{eq:domain_count_relation}
  \nu_n = \phi_n - \beta + 1,
\end{equation}
where $\beta$ is the number of cycles in the graph (the first Betti
number).  Due to the bounds
\begin{equation}
  \label{eq:nodal_bounds}
  n-\beta \leq \nu_n \leq n,
  \qquad
  n-1 \leq \phi_n \leq n-1 + \beta,
\end{equation}
valid for \emph{generic} eigenfunctions
\cite{Ber_cmp08,BanBerSmi_ahp12}, the ratio $\nu_n/n$ converges
trivially\footnote{In the non-generic case, the behavior of $\nu_n/n$
  becomes highly non-trivial.  Recent progress was made in
  \cite{HofKenMugPlu_21}, showing that any sub-sequential limit of
  $\nu_n/n$ is given as a ratio between the length of a sub-graph and
  the length of the entire graph.  This provides lower and upper
  bounds on the possible limits.} to 1.  In this paper we therefore
focus on the \emph{finer properties} of the number of zeros, namely
the distribution of the difference
$\sigma(n) := \phi_n - (n-1) = \nu_n - n +\beta$ which we call the
``nodal surplus''.  We conjecture that, after a suitable rescaling,
this \emph{distribution converges to a universal limit}, namely
Gaussian, for \emph{any} sequence of graphs with increasing
number of cycles, $\beta$.

The contribution of this paper is three-fold.  First of all, in
Section~\ref{sec:main_results} we formulate a precise version of the
above conjecture (which was already present, less explicitly, in
\cite{GnuSmiWeb_wrm04}) and present supporting numerical evidence.
Second, in Theorem~\ref{thm: sampling simplified} we present a
convenient formula which allowed us to numerically compute the
distribution of the nodal surplus for a given graph to high precision,
\emph{without calculating the spectrum and eigenfunctions}.  In essence, we reduce the problem to integration over a high dimensional torus $ \T^{E}:=\R^{E}/2\pi\Z^{E} $ (where $ E $
denotes the number of edges of the graph).  Third, we
prove the conjecture for several sequences of graphs, in particular,
mandarin graphs (also known as ``pumpkin'' or ``watermelon'') and
flower graphs.  This rigorous evidence for the veracity of the
conjecture comes in addition to the results of our previous work
\cite{AloBanBer_cmp18} where the nodal surplus distribution of graphs
with disjoint cycles was calculated explicitly.

The conjecture we formulate in this paper belongs, in spirit, to the
family of ``quantum chaos'' conjectures and results, such as the
conjecture of universality of spectral statistics
(Bohigas--Giannoni--Schmit conjecture), quantum ergodicity or,
indeed, the universality of nodal fluctuations conjecture of Smilansky
and co-workers
(\cite{BanHarSep_jmaa19,BanOreSmi_incoll08,BerKeaWin_cmp04,GnuAlt_pre05,GnuSmiWeb_wrm04,HarrHudg20,KotSmi_prl97,Tan_jpa01}
is a partial reference list \emph{in the context of metric graphs}).
A distinguishing feature of our conjecture, however, is the absence of
any restrictions on the type of graphs where convergence is expected.
While this may be viewed as overly bold, our wide search has not
revealed any counter-examples.  We choose to view it as fortuitous: the common situation in other ``quantum chaos
on graphs'' conjectures is having a list of
known exceptions. If a conjecture is true without a list of exceptions, then proving it might be easier. 

Allowing ourselves to speculate, the nodal surplus appears to behave
as a sum of weakly correlated small contributions ``localized'' on
individual cycles of the graphs.  When all cycles are spatially
separated \cite{AloBanBer_cmp18}, this is a rigorous statement and,
moreover, the cycle contributions are independent.  On the opposite
side of the spectrum, the graphs such as mandarins (which are shown to
satisfy the conjecture in Theorem~\ref{thm: stowers and mandarins})
have cycles that are all ``bunched together''.  And yet, we manage to
find a sum of uncorrelated variables that approximates the nodal
surplus up to a small correction.  This observation strengthens our belief that the conjecture is true in its full generality, and our
hope that it can be proven.

\section{Definitions and preliminaries}

Let $\Gamma(\E,\V)$ be a (discrete) graph. Here, $\E$ and $\V$ denote
the sets of edges and vertices of cardinalities $E:=|\E|$ and
$V:=|\V|$. We allow multiple edges between a given pair of
vertices. An edge may connect a vertex to itself. Such an edge is
called a $\emph{loop}$ (not to be confused with a $\emph{cycle}$ ---
a closed simple path). The multi-set of edges incident to a given vertex
$v\in\V$ is denoted by $\E_{v}$; it is a multi-set because every loop
appears twice. The degree of a vertex is defined as
$\deg{v}=|\E_{v}|$. We call an edge incident to a vertex of degree
one, a \emph{tail}. Throughout the paper, we always assume that:

\begin{ass}
  \label{ass: graph assumptions}
  $\Gamma$ is connected, there are no vertices of degree two, and both
  $\E$ and $\V$ are finite and non-empty.
\end{ass}

The first Betti number of $\Gamma$, i.e.\ its number of ``independent"
cycles, will be denoted by $\beta$. Formally, $\beta$ is the rank of
the first homology group of $\Gamma$ and is given by
\begin{equation*}
  \beta=E-V+1.
\end{equation*}
Given a graph $\Gamma$ and a positive vector $\lv\in\R_{+}^{E}$, the
metric graph $\Gamma_{\lv}$ is obtained by equipping each edge
$e\in\E$ with a uniform metric $\dd x_e$ such that the total length of
the edge is $\ell_e$.  This makes $\Gamma_{\lv}$ a compact metric
space which can be viewed as a one-dimensional Riemannian manifold
with singularities at the vertices.  A function
$f:\Gamma_\lv\rightarrow\C$ may be specified by its restrictions to
every edge $f|_{e}\colon e\to\C$. Introducing the Sobolev
space
$H^{2}(\Gamma_{\lv}):=\bigoplus_{e\in\E}H^{2}\left(e\right)$,
the Laplacian acts on functions in $H^2(\Gamma_{\lv})$ edge-wise:
\[
(\Delta f)|_{e}=-\frac{\dd^{2}}{\dd^{2}{x_{e}}}f|_{e}.
\]
The Laplacian is self-adjoint when restricted to the family of
functions in $H^{2}(\Gamma_{\lv})$ which satisfy Neumann--Kirchhoff\footnote{Also
  called natural or standard vertex conditions.} vertex conditions
at every $v\in\V$:
\begin{enumerate}
\item $f$ is continuous at $v$. Namely, for every
  $e,e'\in\E_{v}$, \[f|_{e}(v)=f|_{e'}(v).\]
\item The outgoing derivatives of $f|_e$ at $v$, denoted by
  $\partial_{e}f(v)$, satisfy
  \[
    \sum_{e\in\E_{v}}\partial_{e}f(v)=0.
  \]
\end{enumerate}

\begin{rem}
  Under these conditions, two edges connected by a vertex of degree 2
  can be replaced by a single edge. Hence, the only finite connected
  graphs excluded in Assumption \ref{ass: graph assumptions} are the
  ``loop graph'' (when $ V=0 $) or a single point (when $ E=0 $).
\end{rem}

The definition of a $\emph{quantum graph}$ may include different
choices of vertex conditions and the addition of scalar and magnetic
potentials to the Laplacian. A thorough review of quantum graphs may
be found in \cite{ KurasovBook,BerKuc_graphs,GnuSmi_ap06,Mugnolo_book} among other sources. In this paper we
only consider pure Laplacian on graphs with Neumann--Kirchhoff
conditions.

\begin{ass}
  \label{assum:standard}
  The graph $\Gamma_{\lv}$ is a finite metric graph equipped with
  Neumann--Kirchhoff vertex conditions.  We will call such graphs
  \emph{standard graphs}.  When referring to the
  spectrum/eigenvalues/eigenfunctions of $\Gamma_{\lv}$ they should be
  understood as those of the Laplacian.
\end{ass}

A standard graph $\Gamma_{\lv}$ has a discrete non-negative
spectrum. It has a complete set of eigenfunctions
$\left\{ f_{n}\right\} _{n\in\N}$ corresponding to eigenvalues
$k_{n}^{2}$ such that
\[
0=k_{1}< k_{2}\le k_{3}\ldots \nearrow\infty.
\]
For a multiple eigenvalue, there is a freedom in choosing an
orthonormal basis of its eigenspace and the number of zeros of the
$n$-th eigenfunction may depend on this choice.  To avoid any
ambiguity, we will focus on ``generic'' eigenfunctions.

\begin{defn}
  \label{def: generic_index}
  An eigenfunction is called $\emph{generic}$ if it corresponds to a
  simple eigenvalue, and it does not vanish on any vertex. Given
  $\Gamma_{\lv}$, we denote the set of labels of generic eigenfunctions
  by
  \[
    \G:=\set{n\in\N}{f_{n}\,\text{is generic}}.
  \]
\end{defn}

\begin{defn}
  Given $\Gamma_{\lv}$, its \emph{nodal surplus},
  $\sigma:\G\rightarrow\left\{ 0,\ldots \beta\right\} $, is defined by
  \begin{equation}
    \label{eq:nodal_surplus_def}
    \sigma\left(n\right) := \big|\set{x\in\Gamma_{\lv}}{f_{n}\left(x\right)=0}\big|
    - (n-1).
  \end{equation}
\end{defn}

\begin{defn} A vector $\lv\in\R_{+}^{E}$ is said to be
  $\emph{rationally independent}$, if the only solution to
  $\lv\cdot\vec{q}=0$ with $\vec{q}\in\Q^{E}$ is $\vec{q}=0$.
\end{defn}

Large graphs with rationally independent lengths are regarded as a
good paradigm of quantum chaos \cite{GnuSmi_ap06,KotSmi_prl97}. It was
shown in \cite[Theorem 2.1]{AloBanBer_cmp18} that for any $\Gamma$ and any
rationally independent $\lv$, the nodal surplus of $\Gamma_{\lv}$ has a
well-defined distribution.  Namely, the limits
\begin{equation}
  \label{eq:prob_defn}
  P(\sigma=s) := \lim_{N\rightarrow\infty}
  \frac{\big|\sigma^{-1}(s)\cap[N]\big|}{\big|\G\cap[N]\big|}, \qquad
  [N] := \left\{1,\ldots ,N\right\}
\end{equation}
exist for all $s \in \{0, \ldots, \beta\}$, are non-negative, and
their sum is 1.

\section{Conjecture and main results}
\label{sec:main_results}

It was suggested by Gnutzmann, Smilansky and Weber in
\cite{GnuSmiWeb_wrm04} that the nodal statistics of large
graphs\footnote{The large graphs limit in \cite{GnuSmiWeb_wrm04} is
  $V\rightarrow\infty$ for ``well connected" graphs.} with rationally
independent edge lengths should exhibit a universal Gaussian
behavior. Recent progress was made in \cite{AloBanBer_cmp18}, showing
that $\sigma$ has a well-defined distribution which is supported
symmetrically on $\{0,1,\ldots \beta\}$ with
$\mathbb{E}\left(\sigma\right)=\frac{\beta}{2}$. To have a continuous
limit, as suggested by \cite{GnuSmiWeb_wrm04}, a sequence of discrete
distributions must be supported on sets of growing cardinality.  We
conjecture that this basic necessary condition is in fact sufficient.

\begin{conjecture}
  \label{conj: Universality} Let
  $\left\{ \Gamma^{\left(\beta\right)}\right\} _{\beta\nearrow\infty}$
  be any sequence of standard graphs, labeled by their first Betti
  numbers. Choosing arbitrary rationally independent edge lengths for
  each $\Gamma^{\left(\beta\right)}$, let
  $\sigma^{\left(\beta\right)}$ denote its nodal surplus random
  variable. Then, in the limit of $\beta\rightarrow\infty$,
  \begin{equation}\label{eq: universality}
    \frac{\sigma^{\left(\beta\right)}-\frac{\beta}{2}}
    {\sqrt{\var\left(\sigma^{(\beta)}\right)}}
    \xrightarrow{\ \mathcal{D}\ }
    N(0,1),
  \end{equation}
  where the convergence is in distribution and $N(0,1)$ is the
  standard normal distribution. Moreover,
  $\var\left(\sigma^{(\beta)}\right)$ has linear growth,
  \begin{equation}
    \frac{\beta}{C}\le\var\left(\sigma^{(\beta)}\right)\le C \beta,
  \end{equation}
  for some constant $C>1$ and for large enough $\beta$.
\end{conjecture}

Let us state more explicitly the claim of the Conjecture.  Given a real random variable $X$, let $N(X)$ denote a normal random variable with mean $\mathbb{E}(X)$ and variance $\var(X)$. The Kolmogorov–Smirnov distance between real random variables $X$ and $Y$ is defined by,
\begin{equation}\label{eq: DS dist defined}
d_{KS}(X,Y):=\sup_{t\in\R}|P(X\le t)-P(Y\le t)|.
\end{equation}

Let $G(\beta)$ denote the (infinite) family of discrete graphs with first Betti number $\beta$, and given a graph $\Gamma$ let $\L(\Gamma)$ denote all possible rationally independent lengths. Then, Conjecture \ref{conj: Universality} says:
\begin{enumerate}
    \item In the limit of $\beta\rightarrow \infty$, \begin{equation}\label{eq: conj dks convergence}
        \sup_{\Gamma\in G(\beta)}\sup_{\lv\in\L(\Gamma)}d_{KS}(\sigma,N(\sigma))\rightarrow 0,
    \end{equation}
    where $\sigma$ indicates the nodal surplus random variable of $\Gamma_{\lv}$.
\item There exists $C>0$ such that for large enough $\beta$,
\begin{equation}\label{eq: conj variance growth}
     \frac{\beta}{C}\le\inf_{\Gamma\in G(\beta)}\inf_{\lv\in\L(\Gamma)}\var(\sigma)\le\sup_{\Gamma\in G(\beta)}\sup_{\lv\in\L(\Gamma)}\var(\sigma)\le C \beta.
\end{equation}
\end{enumerate}

The aim of this paper is to provide both numerical and analytical
evidence in support of this conjecture.

\subsection{Analytical evidence in support of Conjecture \ref{conj:
    Universality}}

We prove the conjecture for three families of graphs:
\begin{enumerate}
    \item \emph{\textbf{Graphs with disjoint cycles}} - We say that a metric graph $\Gamma_{\lv}$ has \emph{disjoint cycles} if every two of its cycles (i.e, simple closed paths on $\Gamma_{\lv}$) are disjoint. See figure \ref{fig:TOC} for an example.
    \item \emph{\textbf{Stower graphs}} - A graph all of whose edges are loops is called a \emph{flower}. A graph all of whose edges are tails is called a \emph{star}. We call a graph \emph{stower}, if each of its edges is either a loop or a tail. See Figure \ref{fig: Stowers and mandarins} for an example.
    \item \emph{\textbf{Mandarin graphs}} - We call a graph \emph{mandarin}\footnote{also known as pumpkin or watermelon.} if it has only two vertices, and every edge of the graph is connected to both vertices. In particular, it has no loops. See Figure \ref{fig: Stowers and mandarins} for an example.
\end{enumerate}
\begin{figure}
\includegraphics[height=0.2\paperheight]{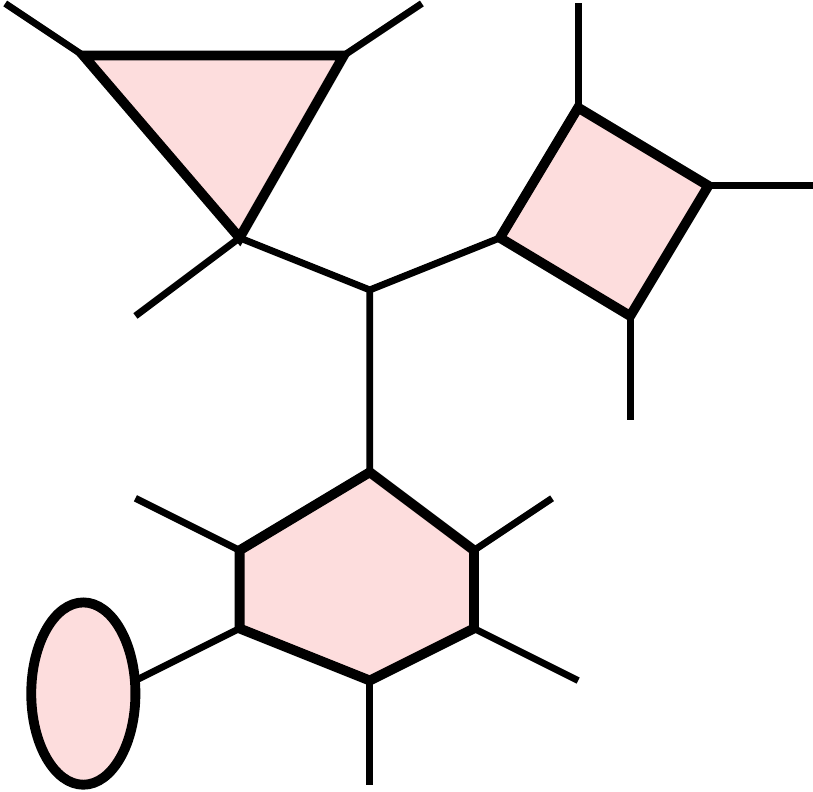}

\caption{\label{fig:TOC} A graph with disjoint cycles. The cycles are emphasized by coloring the faces they bound.}
\end{figure}
\begin{figure}
\includegraphics[width=0.4\paperwidth]{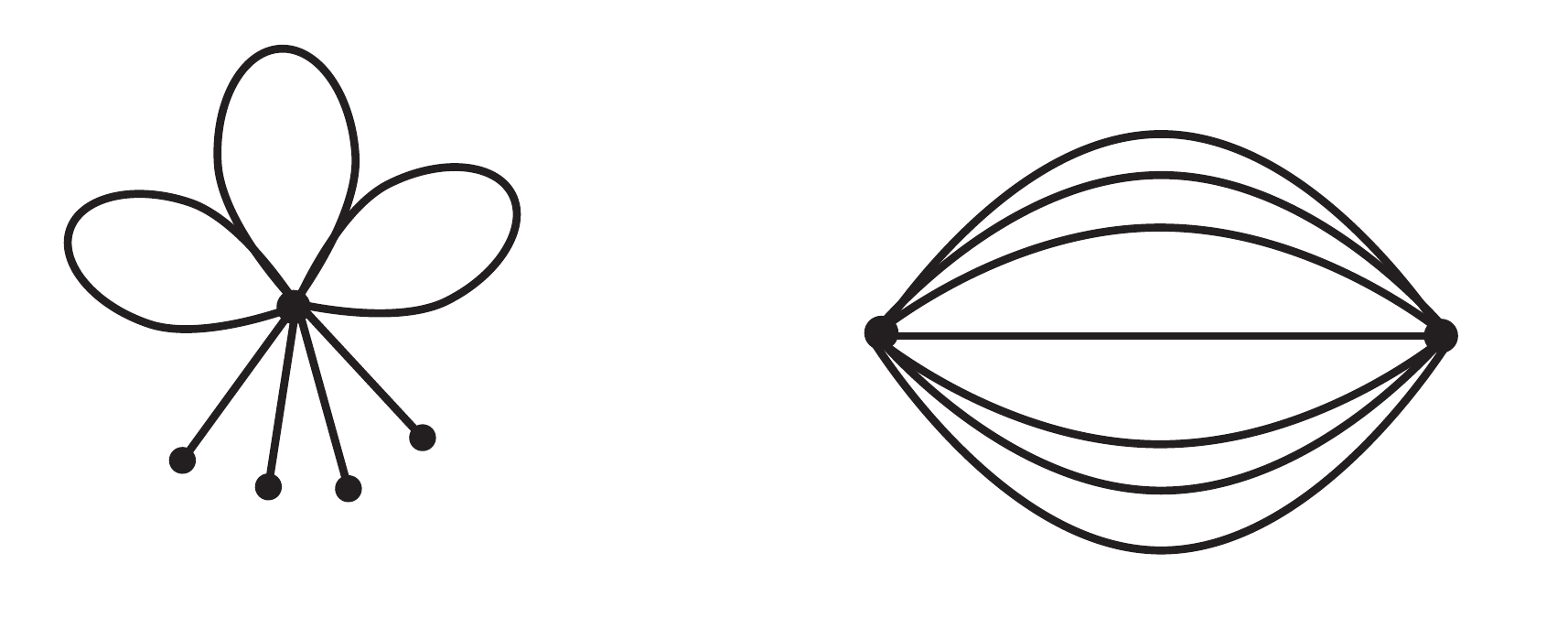}
\caption{\label{fig: Stowers and mandarins} On the left, a stower graph with 3 loops and 4 tails. On the right, a mandarin graph with 7 edges.}
\end{figure}
The nodal surplus distribution for graphs with disjoint cycles was calculated in \cite{AloBanBer_cmp18}:
\begin{thm}\cite{AloBanBer_cmp18}\label{thm: TOC}
If $\Gamma_{\lv}$ has disjoint cycles and rationally independent edge lengths, then its nodal surplus distribution is binomial with parameters $n=\beta$ and $p=\frac{1}{2}$. That is,
\begin{equation}\label{eq: binomial TOC}
    P(\sigma=s) = \binom{\beta}{s}2^{-\beta}.
\end{equation}
\end{thm}
\begin{cor}
The family of graphs with disjoint cycles satisfies Conjecture \ref{conj: Universality}.
\end{cor}

In contrast to this family of graphs, the cycles of stower graphs and mandarin graphs are clustered together such that every pair of cycles share an edge or a vertex. These are, in a sense, opposite to the case of disjoint cycles.

\begin{thm}\label{thm: stowers and mandarins}
Both graph families of stowers and mandarins satisfy Conjecture \ref{conj: Universality}.
\end{thm}
We prove Theorem \ref{thm: stowers and mandarins} in Section \ref{subsec: proof for stowers and mandarins}.

\begin{rem}
  Theorem \ref{thm: stowers and mandarins} can be extended to include
  mandarin graphs with an added tail, but the proof is cumbersome and
  we do not include it here. Additionally, the methods used in the
  proof of Theorem \ref{thm: TOC} can be extended to graphs obtained
  by concatenation of small graphs.
\end{rem}

\subsection{Efficient computation of the nodal surplus distribution}
The next theorem provides an efficient way to evaluate the nodal surplus distribution for large graphs. Given a graph $\Gamma$ with $E$ edges, its so-called bond scattering matrix $S$, explicitly defined in \eqref{eq: S matrix},  is a $2E\times2E$ constant real orthogonal matrix.

\begin{defn}
\label{def:The-unitary-evolution}The \emph{unitary evolution matrix} associated to  $\kv\in\T^{E}:=\R^{E}/2\pi\Z^{E}$ is defined by
\begin{equation}
    \label{eq:Udef}
    U_{\kv}:=e^{i\hat{\kappa}}S,
\end{equation}
where $e^{i\hat{\kappa}}$ is the diagonal matrix
\[
e^{i\hat{\kappa}}:=\diag\left(e^{i\kappa_{1}},e^{i\kappa_{2}}\ldots e^{i\kappa_{E}},e^{i\kappa_{1}},e^{i\kappa_{2}},\ldots e^{i\kappa_{E}}\right).
\]
The eigenvalues and (orthonormal) eigenvectors of $U_{\kv}$ will be denoted by $e^{i\theta_{n}}(\kv)$ and $\ba_{n}(\kv)$, for $1\le n\le2E$. When $\kv$ is understood from the context, we simply write $e^{i\theta_{n}}$ and $\ba_{n}$.
\end{defn}

Given $\kv$ and an eigenpair $(e^{i\theta_{n}},\ba_{n})$ of
$U_{\kv}$, we construct $\bH_n(\kv)$, an $E\times E$ real symmetric matrix
whose signature is related to the nodal
surplus, see Sections \ref{subsec: nodal magnetic} and \ref{subsec:hessian} for details. To be able to define $\bH_n(\kv)$, we first introduce the self-adjoint matrix $ g(e^{-i\theta_{n}}U_{\kv}) $ given by 
\begin{equation*}
  g(e^{-i\theta_{n}}U_{\kv})
  = \sum_{m\le 2E~;~e^{i\theta_{m}}\ne e^{i\theta_{n}}}
  \cot\left(\frac{\theta_{n}-\theta_{m}}{2}\right)\ba_{m}\ba_{m}^{*}.  
\end{equation*}
The map $ g $ which sends unitary matrices to self-adjoint matrices is
defined in Section \ref{subsec:hessian}. We will remark that $g(M)$
agrees with the inverse Cayley transform of $M$ whenever defined. Next we define the matrices $Z_{j}$  for $ j=1,2,\ldots,E $. Each $Z_{j}$ is a $2E$-dimensional matrix, with
\begin{equation}
    \label{eq:Zmatrix_def}
    (Z_{j})_{j,j}=-(Z_{j})_{j+E,j+E}=1,
\end{equation}
and zero in all other entries. We define $\bH_n(\kv)$ by
\begin{equation}
  \label{eq:Hmatrix_def}
  \left(\bH_n(\kv)\right)_{j,j'}
  :=
  \ba_{n}^{*}Z_j g\left(e^{-i\theta_{n}}U_{\kv}\right)
  Z_{j'}\ba_{n},
  \qquad j,j' \in \{1, \ldots, E\}.
\end{equation}

\newcommand{\idxH}{I^+_s\big(\bH_n(\kv)\big)}

\begin{thm}
  \label{thm: sampling simplified}
  Let $\Gamma$ be a graph with $E$ edges and first Betti number
  $\beta$. Assume $\Gamma$ has no loops. Then for any rationally
  independent $\lv\in \R_{+}^{E}$, the nodal surplus distribution of
  $\Gamma_{\lv}$ is given by
  \begin{equation}
    P\left(\sigma=s\right)
    = \frac{1}{(2\pi)^{E}} \int_{\T^{E}}
    \sum_{n=1}^{2E} \idxH
    \frac{\ba_{n}^{*}\boldsymbol{L}\ba_{n}}{\tr\boldsymbol{L}}
    \dd \kv,
    \label{eq: nodal stat in the simplified thm}
  \end{equation}
  where
  $\boldsymbol{L}:=\diag\left(l_{1},l_{2}\ldots
    l_{E},l_{1},l_{2},\ldots l_{E}\right)$ and
  the indicator function $I^+_s(H)$ is equal to 1 if the matrix $H$ has exactly $s$ positive eigenvalues and 0 otherwise.
\end{thm}

The ``no loops" assumption was introduced for readability only; the
complete version of this result appears as Theorem \ref{thm:
  generalized theorem} in section \ref{subsec: proof of main
  theorem}. We will also show that choice of the eigenvalue basis in
the definition of $\bH_n$ does not alter the integral in \eqref{eq:
  nodal stat in the simplified thm}.

\begin{rem}
  \label{rem:index_connections}
  The Cayley transform of a scattering matrix is related to the
  Dirichlet-to-Neumann map (see \cite[eq.~(5.4.8)]{BerKuc_graphs} in
  the context of metric graphs).  The appearance of an index such as
  $I^+_s\big(\bH_n(\kv)\big)$ suggests links to the
  Cox--Jones--Marzuola formula for the nodal surplus in terms of the
  Morse index of the DtN map \cite{CoxJonMar_iumj17}.
\end{rem}

Theorem~\ref{thm: sampling simplified} is built upon two major ingredients:
\begin{enumerate}
    \item  A method, originated in \cite{BarGas_jsp00} by Barra and Gaspard, that replaces certain spectral averages with integration over a hypersurface $\Sigma$ in the torus $\T^{E}$. In Theorem \ref{thm: BGmeasure as spectral measures} we show that the integration over the implicitly defined $\Sigma$ can be replaced by an explicit integration over the torus $\T^{E}$, in the spirit of \cite{BerWin_tams10,GnuSmi_ap06}.
    \item A connection between the nodal surplus and eigenvalue stability with respect to certain perturbations, observed in \cite{Ber_apde13,ColindeVerdiere_graphes,BerWey_ptrsa14}, and its formulation as a function on the hypersurface $\Sigma$, \cite{Ban_ptrsa14,AloBanBer_cmp18}. In sections \ref{subsec: nodal magnetic} and \ref{subsec:hessian} we compute this function explicitly in terms of the data used in the integral of Theorem \ref{thm: BGmeasure as spectral measures}.
\end{enumerate}

\subsection{The polytope of nodal surplus distributions of a given
  graph}\label{subsec: We and convex hull}
Theorem \ref{thm: sampling simplified} gives a computationally efficient way to calculate the nodal surplus for a given graph and a given choice of edge lengths.  On the other hand, Conjecture \ref{conj: Universality} claims convergence for any sequence of graphs and almost any choice of lengths.  We will now explain how to cover all choices of lengths for a given graph by testing only a few selected choices.

We note that the expression in the right-hand-side of (\ref{eq: nodal stat in the simplified thm}) is well-defined for every non-zero $\lv\in\R_{\ge0}^{E}$. In particular, given $e$ we define $W_{e,s}$ by taking $\lv$ with $\lve=1$ and zero elsewhere,
\begin{equation}
  \label{eq: Wej def}
  W_{e,s} := \frac{1}{\left(2\pi\right)^{E}} \int_{\T^{E}}
  \sum_{n=1}^{2E} \idxH
  \frac{|(\ba_{n})_{e}|^2+|(\ba_{n})_{e+E}|^2}{2}\dd \kv.
\end{equation}
This way, (\ref{eq: nodal stat in the simplified thm}) may be written as a convex combination of $W_{e,s}$'s with weights given by the normalized edge lengths.  Namely,
\begin{align}
  \label{eq: P convex in W}
  P(\sigma = s) = \frac{1}{L}\sum_{e\in\E}\lve W_{e,s},
  \qquad
  L := \sum_{e\in \E} \ell_e,
\end{align}
for rationally independent $\lv$.  The nodal surplus distribution of
$\Gamma_\lv$ is characterized by the vector
$\vec{P_\lv}\in \R_{\ge0}^{\beta+1}$ of probabilities $P(\sigma=s)$,
$s=0,\ldots ,\beta$.  The set of all such vectors for a given graph
$\Gamma$ (and varying $\lv$) satisfies
\begin{equation}\label{eq: convex hull}
  \overline{\set{\vec{P_\lv}}{\lv\in\L(\Gamma)}}
  = \conv\set{\vec{W}_e}{e \in \E},
\end{equation}
where $\conv$ is the convex hull and $\vec{W}_{e}=(W_{e,0},W_{e,1},\ldots, W_{e,\beta})$ and the closure is with respect to the standard topology on $ \R^{\beta+1} $. Thus to understand the set of all nodal surplus distributions of a given graph we only need to understand the extreme points $\vec{W}_e$.  Furthermore,
this task is often reduced by the presence of symmetries in the underlying (discrete) graph $\Gamma$.  It will be shown in Theorem \ref{thm: symmetry theorem} that if $e_1$ and $e_2$ are related by a symmetry, then $\vec{W}_{e_1}=\vec{W}_{e_2}$.  All of the above remains valid for graphs with loops after a small modification of equation \eqref{eq: P convex in W}, see equation \eqref{eq: Pvec convex in We}.

We bound the $\lv$ dependence in
Conjecture~\ref{conj: Universality} and equations \eqref{eq: conj dks
  convergence}-\eqref{eq: conj variance growth} in terms of the
extreme points of the convex hull. Introducing new random variables
$\omega_{e}$ supported on $\{0,1,\ldots \beta\}$ with
probabilities given by the vector $\vec{W}_e$, namely
$P(\omega_{e}=s):=W_{e,s}$, we have the following.

\begin{lem}\label{lem: upper bound We}
Given a graph $\Gamma$, let $\varepsilon:=\sqrt{\frac{\max_{e\in\E}\var(\omega_{e})}{\min_{e\in\E}\var(\omega_{e})}}-1$. Then, for any $\lv\in\L(\Gamma)$,
\begin{align}
    \label{eq: variance bounds} \min_{e\in\E}\var(\omega_{e})\le & \var(\sigma)\le\max_{e\in\E}\var(\omega_{e}),\quad\text{ and }\\
   \label{eq:upper bound}\sup_{\lv\in\L(\Gamma)}d_{KS}(\sigma,N(\sigma))\le & \max_{e\in\E}d_{KS}(\omega_{e},N(\omega_{e}))+\varepsilon.
\end{align}
\end{lem}
This lemma is proved in Appendix \ref{sec: appendix proof of upper bound}. It follows that Conjecture \ref{conj: Universality} holds if the following two sufficient conditions hold:
\begin{enumerate}
	\item the upper bound in \eqref{eq:upper bound} converges to zero in $\beta$ uniformly on $G(\beta)$, and
	\item for every graph in $G(\beta)$ and any edge $ e $ of that graph,  $\var(\omega_{e})$ is of order $\beta$.
\end{enumerate}
We remark that these are actually necessary and sufficient conditions, however, we will not show that here. To show it, one needs to provide a lower bound on the supremum in equation \eqref{eq:upper bound}, which can be done but the proof is cumbersome and we only need the upper bound in this work.


\subsection{Numerical evidence in support of Conjecture \ref{conj: Universality}}

\subsubsection{Setting}\label{subsec: exp setting}

To provide supportive evidence for Conjecture \ref{conj:
  Universality}, we compute, numerically, the bounds in
\eqref{eq:upper bound} and \eqref{eq: variance bounds} for $ 26 $ different discrete graphs with first Betti number ranging between $ \beta=6 $ to $ \beta=55 $. These $ 26 $ graphs that were chosen are
\begin{enumerate}
\item Complete graphs on $n$ vertices for $n\in\{5,6\ldots 12\}$.
\item Periodic ladder graphs, see Figure~\ref{fig:ladder and Z2}, of
  $n$ steps for $n\in\{6,10,14,18,22\}$.
\item Periodic square lattices of $n^{2}$ vertices (Cayley graphs of
  $\Z^{2}/n\Z^{2}$, see Figure~\ref{fig:ladder and Z2}) for
  $n\in\{4,5,6,7\}$.
\item Random $5$-regular graphs. We choose one graph at random
  (uniformly) out of all possible $5$-regular graphs of $n$ vertices
  for each $n\in\{12,14,16,18,20\}$.
\item Random Erd\H{o}s-R\'enyi graphs with $n$ vertices: each pair of
  vertices is connected by an edge with probability $0.75$
  independently of all others. We choose one random graph for each
  $n\in\{9,10,11,12\}$.
\end{enumerate}

\begin{rem}\label{rem: random graphs}
  We emphasize that in our experiments we sampled one graph from each
  class of random graphs of a given size.  The distribution of a
  nodal-like quantity over an ensemble of random graphs of fixed size
  is a related but distinct problem, and we do not address it here.
  We also remark that the first Betti number of an Erd\H{o}s-R\'enyi
  graph of size $n$ is random (but, in an appropriate sense, growing
  with $n$). Thus, the first Betti number assigned to each such graph,
  both in Figure \ref{fig:KS distance} and Figure \ref{fig:Variances},
  was calculated after the random graph was generated.
\end{rem}
\begin{rem}
Computations of another family of graphs, which we do not present here, was suggested by the authors of \cite{KolSar20} due to their constant spectral gap property. It includes the graphs in \cite[fig.~13,e,f]{KolSar20} and their $\mathcal{T}$ iteration (defined in \cite{KolSar20}). These graphs were investigated and their data (variance and KS-distance) agreed with the data presented in Figure \ref{fig:Variances} and Figure \ref{fig:KS distance}.
\end{rem}

\begin{figure}
    \centering
    \includegraphics[width=0.6\paperwidth]{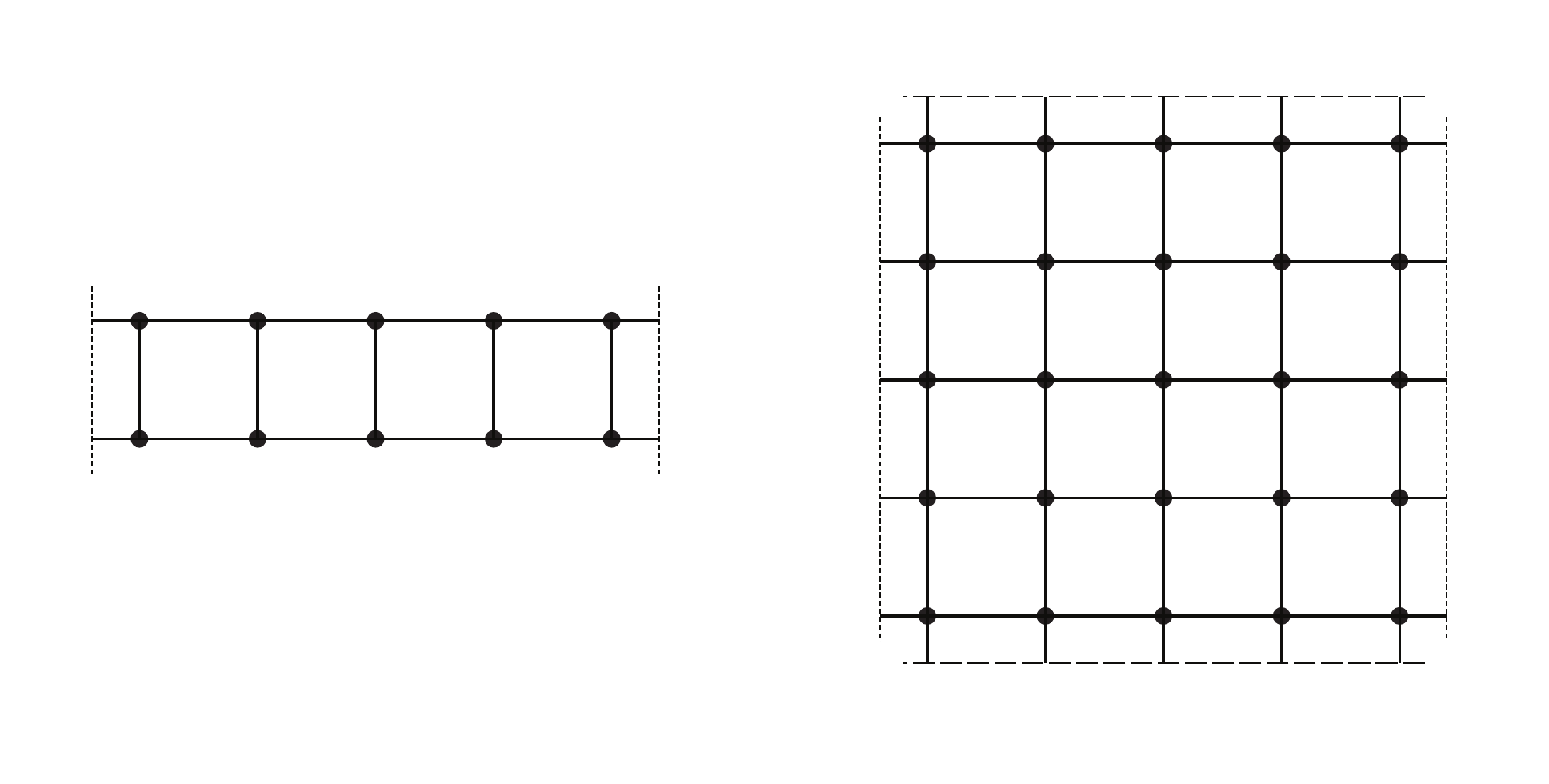}
    \caption{(left) Periodic ladder graph with $n=5$. (right) periodic
      square lattices of $n^{2}=5^2$ vertices. Dotted and dashed lines
      indicates gluing.}
    \label{fig:ladder and Z2}
\end{figure}

Given each of these graphs we compute its $\omega_{e}$'s as follows. We approximate the integral in \eqref{eq: Wej def} as
\[P(\omega_{e}=s)\approx\frac{\sum_{i=1}^{N}h(\kv_{i})}{N},\]
with
\begin{equation*}
  h(\kv) := \sum_{n=1}^{2E}
  \idxH
  \frac{|(\ba_{n})_{e}|^2+|(\ba_{n})_{e+E}|^2}{2},
\end{equation*}
by randomly sampling $N=\frac{10^6}{2E}$ points
$\kv_{i}\in[0,2\pi]^{E}$ with uniform distribution. Heuristically,
each sample provides the data of $2E$ eigenfunctions (see, for
example, Theorem \ref{thm: BGmeasure as spectral measures} and Remark \ref{rem: average over orbits}), and so sampling $N=\frac{10^6}{2E}$ points provides the data of $10^{6}$ eigenfunctions. For a rigorous estimate of accuracy, we may apply Bernstein inequalities\footnote{We apply Bernstein inequality for bounded zero-mean independent  random variables. We consider the random variables  $X_{i}:=f(\kv_{i})-P(\omega_{e}=s)$, which have zero-mean and are bounded by $|X_{i}|\le 1$.}, by which the approximation error can be quantified as follows:
\[
  \forall\delta>0,\qquad
  P\left(\left|P(\omega_{e}=s)-\frac{\sum_{i=1}^{N}h(\kv_{i})}{N}\right|>\delta\right)
  \le 2 e^{-N\frac{\delta^{2}}{2(1+\frac{\delta}{3})}}.
\]
Another simplification is possible due to Theorem \ref{thm: symmetry theorem}, which states that $\omega_{e}=\omega_{e'}$ if $e$ and $e'$ are related by a symmetry of the underlying discrete graph. In particular, in complete graphs and square lattices, every pair of edges can be related by a symmetry and so there is only one $\omega_e$ to compute.

We may estimate the running time of our algorithm as $O(N E^{3})$. To
see that, note that we have $N$ iterations. In every iteration we
first compute the eigenvalues and eigenvectors of $U_{\kv}$ at the
cost $O(E^{3})$. Given the eigenvalues and eigenvectors of $U_{\kv}$, generating the $2E$ matrices $\bH_n(\kv)$ is $O(E^{3})$ and computing their signature is $O(E^{3})$.

\subsubsection{Experimental results}\label{subsec: exp result}
Our numerical results strongly support the conjecture. Figure
\ref{fig:KS distance} shows the upper bound \eqref{eq:upper bound} on
$d_{KS}(\sigma,N(\sigma))$. The upper bound is presented as a function
of $\beta$, and a uniform decrease can be observed as $\beta$
grows. Figure \ref{fig:Variances} shows $\var(\omega_{e})$ for
all graphs and all edges. The variance is presented as a function of
$\beta$, and its linear growth is evident. All values obtained in the
experiment satisfy
\[
  \frac{\beta}{10}\le\var(\omega_{e})\le\frac{\beta}{5}.
\]
In particular, all graphs in the experiment have variance smaller than $\beta/4$ which is the variance for graphs with disjoint
cycles, by Theorem \ref{thm: TOC}. We believe that $ \frac{\beta}{4} $ should be the upper bound for all graphs.

\begin{figure}
  \includegraphics[width=0.5\paperwidth]{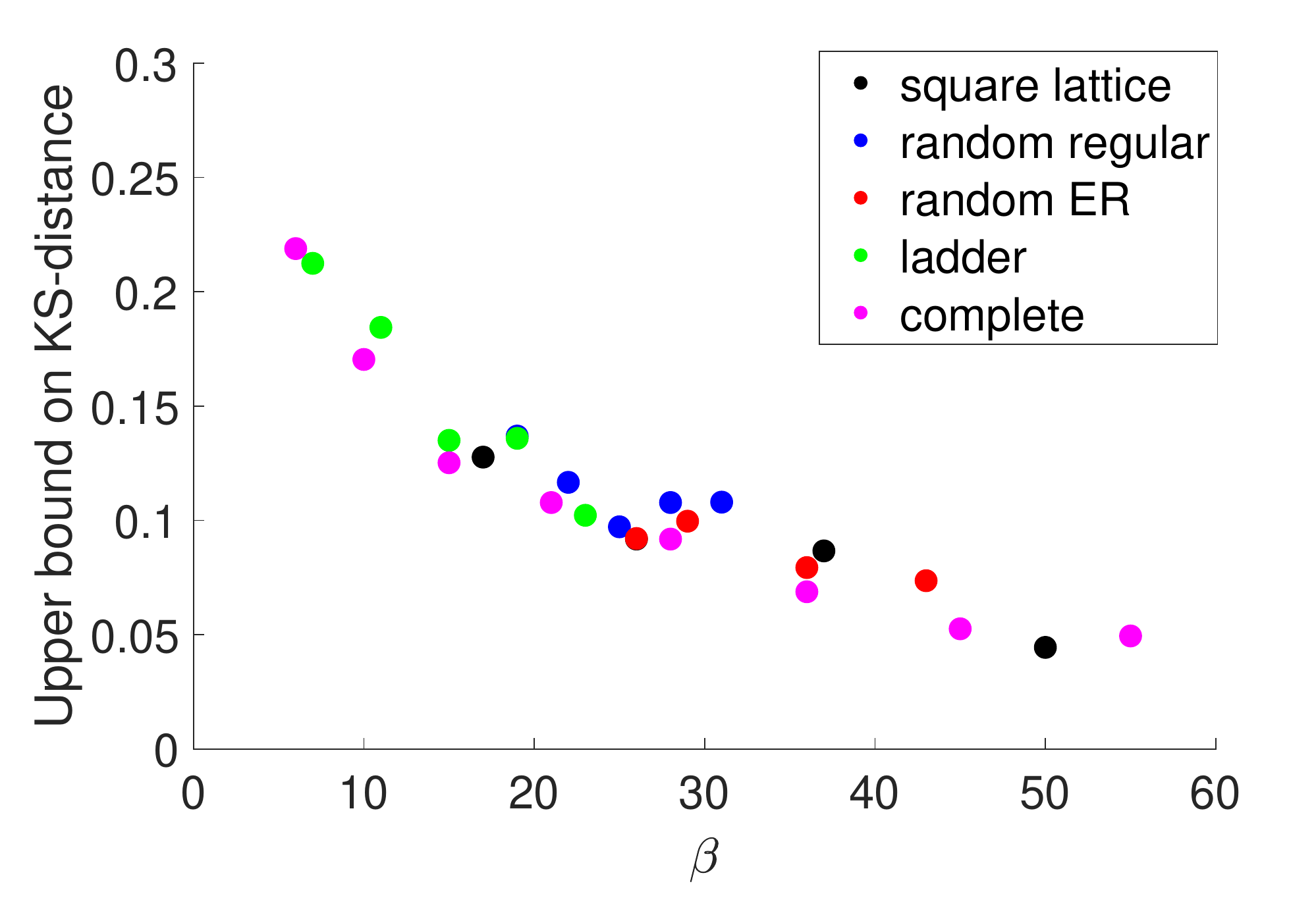}
  \caption{\label{fig:KS distance}The upper bound, by \eqref{eq:upper
      bound}, on $d_{KS}(\sigma,N(\sigma))$, as a function of the
    first Betti number $\beta$.  Each point corresponds to a single
    graph. The colors of points indicate their family: (black) square
    lattices, (blue) 5-regular graphs, (red) Erd\H{o}s-R\'enyi graphs,
    (green) ladder graphs and (magenta) complete graphs.}
\end{figure}

\begin{figure}
  \includegraphics[width=0.5\paperwidth]{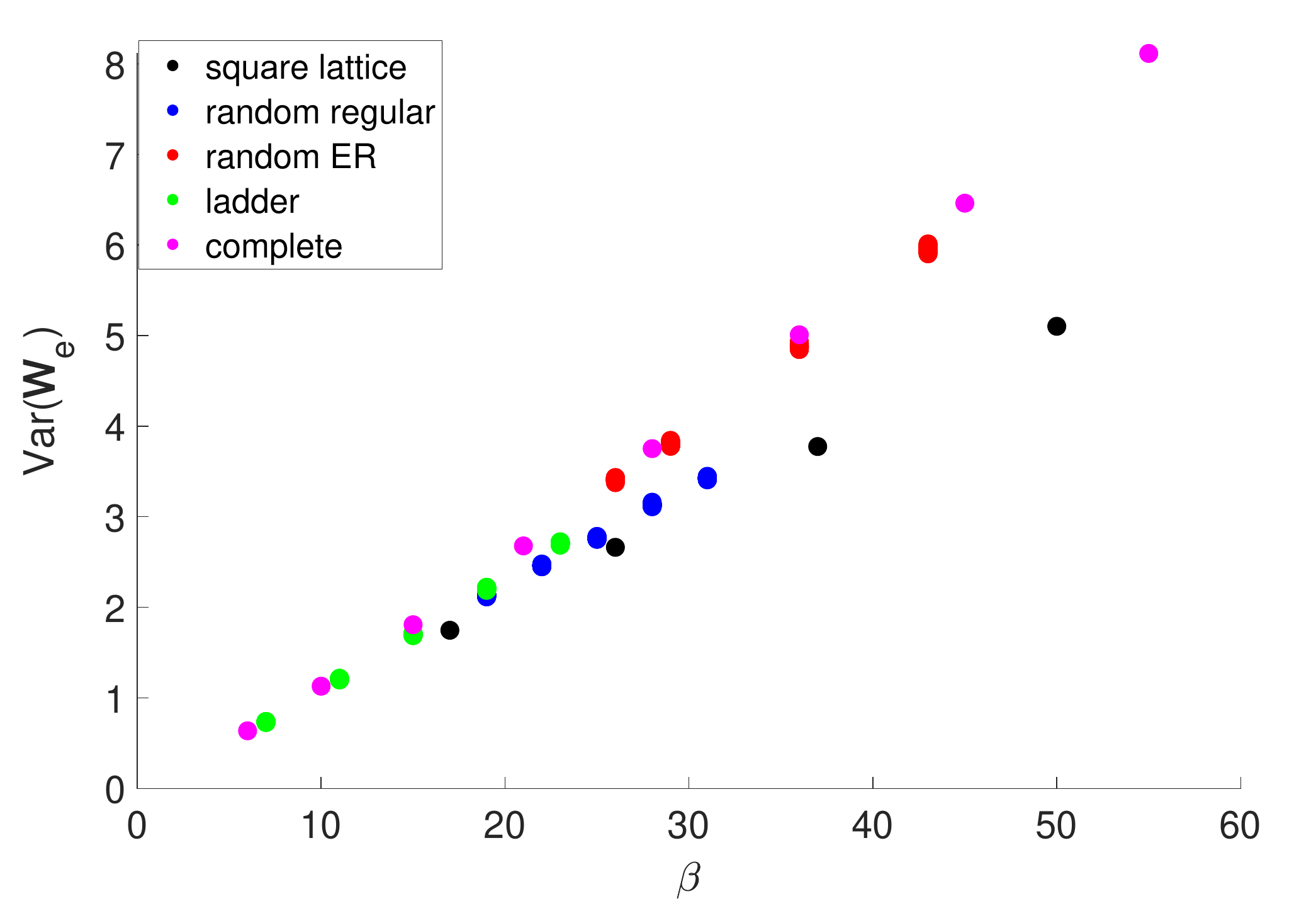}
  \caption{\label{fig:Variances}The variance,
    $\var{(\omega_{e})}$, as a function of the first Betti
    number $\beta$. The points representing $\var{(\omega_{e})}$,
    are plotted for each edge $e$ of every graph.  In many cases, the differences
    between d$\var{(\omega_{e})}$ for different edges are barely visible. The
    colors of points indicate their family: (black) square lattices,
    (blue) 5-regular graphs, (red) Erd\H{o}s-R\'enyi graphs, (green)
    ladder graphs and (magenta) complete graphs.}
\end{figure}

The convergence of the $\omega_{e}$'s to Gaussian and the
linear growth of their variance can also be seen in appendix \ref{sec:
  appendix experiment}.  There we show the full probability
distribution of the ``worst'' $\omega_{e}$ for each graph
considered.  More precisely, we plot the $\omega_{e}$ which
maximizes
$d_{KS}\left(\omega_{e},N(\omega_{e})\right)$ among
all edges of the graph.

\begin{rem}
  In this experiment we cover the nodal statistics of $26$ graphs,
  with $\beta$ ranging from $6$ to $55$ and $E$ ranging from $10$ to
  $98$.  The efficiency of our algorithm using Theorem \ref{thm:
    sampling simplified} is a major improvement over the direct
  computation of eigenfunctions and their zeros. For comparison,
  computing the first $10^{6}$ eigenfunctions of a graph with $E=10$
  can take more than a day (for a given choice of edge lengths), while
  our algorithm will compute the nodal surplus statistics (for
  all rationally independent edge lengths) in a few minutes.
\end{rem}

\section{Secular averages}
\label{sec:secular}

In this section we establish a new analytic expression for computing
certain spectral averages of standard graphs with rationally
independent lengths.  In section~\ref{sec:secular_manifold} we recall
that the spectrum of a standard graph is given as the roots of the
characteristic (secular) equation or, equivalently, the intersection
times of the \emph{secular manifold} by a flow.  Barra and Gaspard
observed in \cite{BarGas_jsp00} that since the flow is ergodic, some
spectral averages may be computed by integrating over the secular
manifold with a suitable measure, which we review in
section~\ref{sec:BGidea}.  In section~\ref{sec:torus_average} we adapt
this approach for more efficient numerical computation by replacing
the integral over the (implicitly defined) secular manifold with an
integral over the embedding torus. The main result of this
section is Theorem~\ref{thm: BGmeasure as spectral measures}.

\subsection{The secular manifold}
\label{sec:secular_manifold}

The discrete graphs $\Gamma$ we consider are undirected. An undirected
graph $\Gamma$ can be made into a bi-directed graph by replacing each
edge $e$ with two directed edges (sometimes called bonds) of opposite
orientation.  Given a directed edge $b$ we denote its reversed version
by $\hat{b}$. We number the directed edges by
$\{1,2,\ldots E,\hat{1},\hat{2}\ldots \hat{E}\}$. Let $i$ and $j$ be
two directed edges such that $i$ terminates at a vertex $v$ and $j$
originates from a vertex $u$.  If $u=v$ we write
$i\xrightarrow[v]{}j$, otherwise we write
$i\not\rightarrow{}j$.  The
\emph{bond scattering matrix}, $S$, is a $2E\times2E$ real orthogonal
matrix defined as follows:
\begin{equation}\label{eq: S matrix}
S_{j,i}=\begin{cases}
\frac{2}{\deg v} & i\xrightarrow[v]{}j\,\,\,\text{and}\,\,i\ne\hat{j}\\
\frac{2}{\deg v}-1 & i\xrightarrow[v]{}j\,\,\,\text{and}\,\,i=\hat{j}\\
0 & i\not\rightarrow{}j
\end{cases}.
\end{equation}
Recall our definition of the unitary evolution matrix, $U_{\kv}:=e^{i\hat{\kappa}}S$, associated to a point $\kv\in\T^{E}$ (see Definition \ref{def:The-unitary-evolution}).
\begin{defn}\label{def: Sigma}
The \emph{secular manifold }is the torus subset $\mg\subset\T^{E}$
for which $U_{\kv}$ has eigenvalue $1$. Namely,
\[
\mg:=\set{\kv\in\T^{E}}{\det\left(1-U_{\kv}\right)=0}.
\]
Its set of regular points is defined as,
\[
\mreg:=\set{\kv\in\mg}{\nabla\det\left(1-U_{\kv}\right)\ne0},
\]
and the set of singular points is $\msing:=\mg\setminus\mreg$.
\end{defn}

\begin{rem}
  \label{rem:Sigma_sets_properties}
  The sets $\mg, \mreg$ and $\msing$ can be
  characterized as the sets for which $\dim(\ker(1-U_{\kv}))\ge 1$,
  $\dim(\ker(1-U_{\kv}))= 1$ and $\dim(\ker(1-U_{\kv})) > 1$
  correspondingly, \cite{BerWin_tams10, CdV_ahp15}.  The regular part
  $\mreg$ is an $E-1$ dimensional oriented smooth\footnote{Even
    real analytic \cite{AloBan19}.} Riemannian manifold, \cite{CdV_ahp15}.  
\end{rem}

We use the notation $\fr{\xv}$ to denote the
remainder, modulo $2\pi$, in each coordinate.  The secular manifold
$\mg$ can be used to compute the spectrum of the standard graphs
$\Gamma_{\lv}$ for any choice of $\lv\in\R_{+}^{E}$.

\begin{lem}
  \label{lem: secular equation and isomorphism}
  \cite{Bel_laa85,KotSmi_ap99,BarGas_jsp00,CdV_ahp15} Let $\Gamma$ be a graph
  and let $\mg$ be its secular manifold. Given a fixed choice of
  $\lv\in\R_{+}^{E}$, the spectrum of $\Gamma_{\lv}$ is characterized
  by
  \begin{enumerate}
  	\item $ k^{2} $  is an eigenvalue of $ \Gamma_{\lv} $ if and only if $ \fr{k\lv} $ lies in $ \mg $.
  	\item A non-zero eigenvalue $k^2\ne0$ is simple if $\fr{k\lv}$ lies in $ \mreg $, and has multiplicity if $\fr{k\lv}$ lies in $ \msing$.
  	\item The singular set $ \msing $ has positive co-dimension in $ \mg $, $ \dim(\msing)\le E-2 $, and so the regular set $ \mreg $ has full measure in $ \mg $.
  \end{enumerate}
\end{lem}
Part (1) of the lemma, served as the motivation for constructing the secular manifold in \cite{BarGas_jsp00}, and was already deduced implicitly in \cite{Bel_laa85,KotSmi_ap99}. Parts (2) and (3) can be found in \cite[Theorem 1.1]{CdV_ahp15}.
\begin{figure}[h!]
  \begin{subfigure}{0.2\textwidth}
\includegraphics[width=1\textwidth]{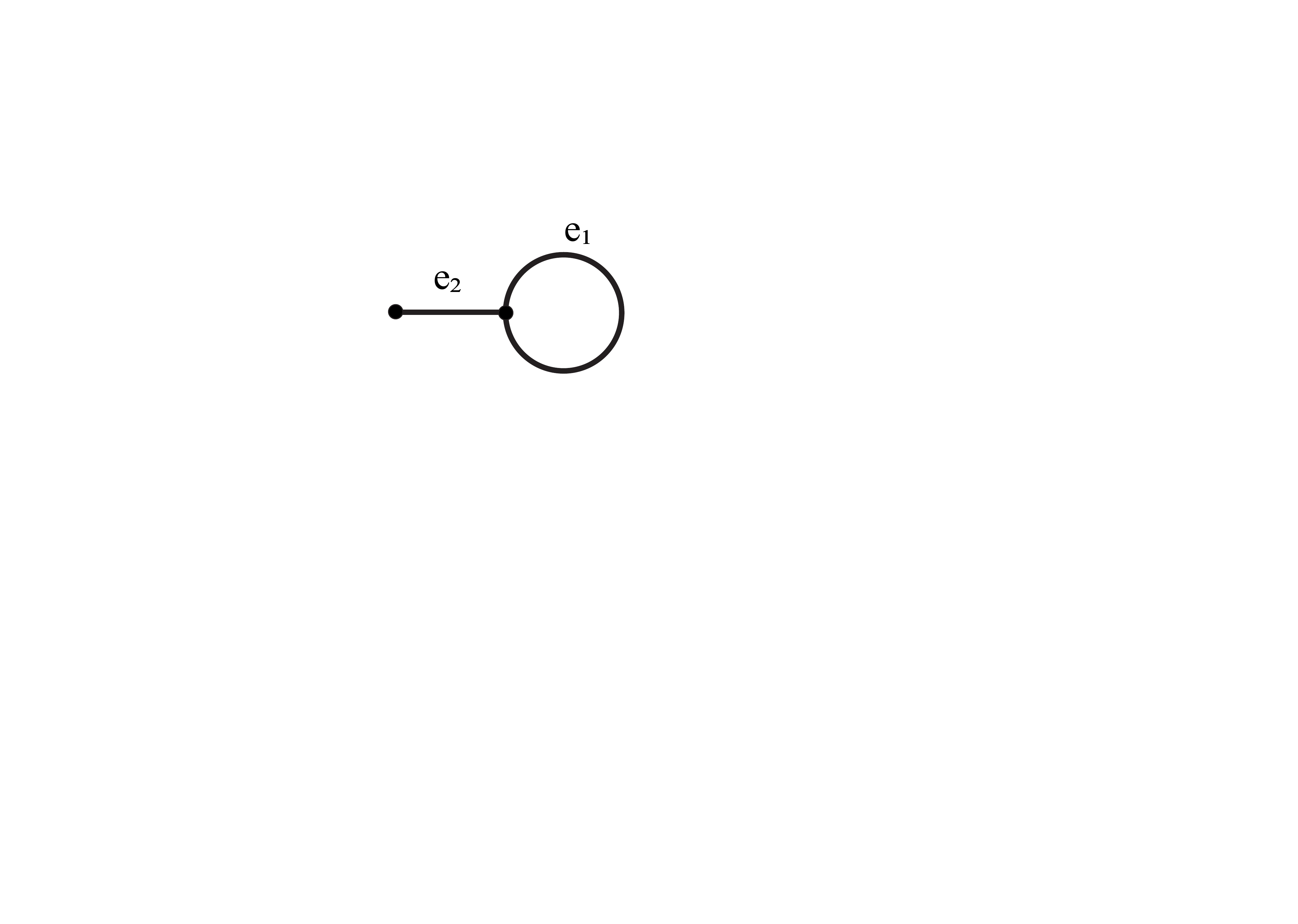}
\end{subfigure}
\begin{subfigure}{0.75\textwidth}
\includegraphics[width=1\textwidth]{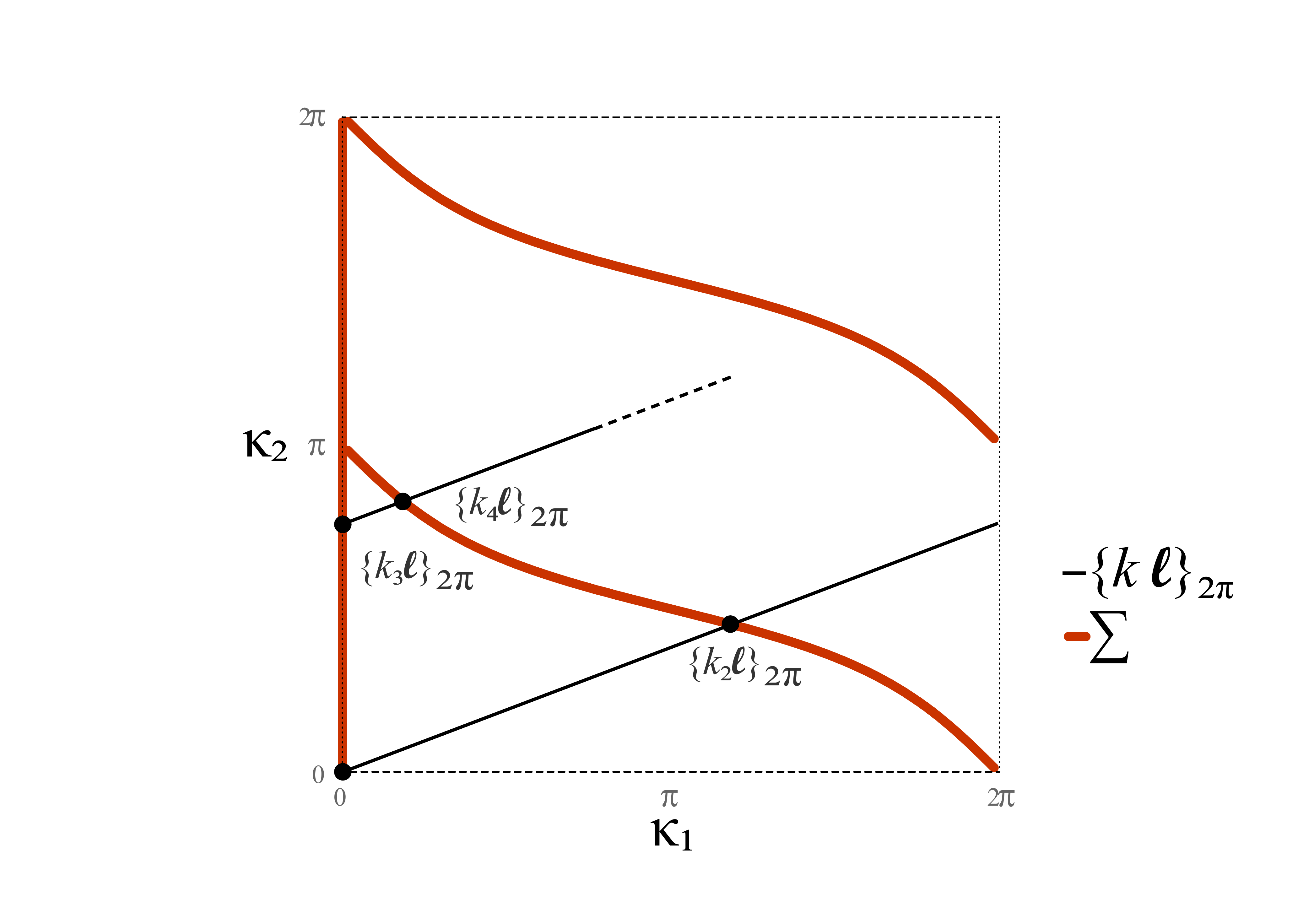}
\end{subfigure}

\caption{(left) a graph $\Gamma$ with two edges. (right) the secular manifold $\mg$ and the flow $k\mapsto\fr{k\lv}$ intersecting it. There are two singular points $\msing=\{(0,\pi),(0,2\pi)\}$ (and their identifications).}

\label{fig: flow on secman}
\end{figure}
The lemma is illustrated in Figure \ref{fig: flow on secman}.
The (discrete) graph $\Gamma$ determines the secular manifold
$\mg$, while the lengths $\lv$ determine the linear flow on
the torus, $k\mapsto \fr{k\lv}\in\T^{E}$. The (square-root)
eigenvalues of $\Gamma_{\lv}$, $\{k_{n}\}_{n\in\N}$, are the ``hitting
times" for which $\fr{k\lv}\in\mg$.  This simple but fruitful viewpoint was first introduced, in the
context of quantum graphs, by Barra and Gaspard \cite{BarGas_jsp00}.

\subsection{The Barra-Gaspard method}
\label{sec:BGidea}

A \emph{spectral observable} is a function on the eigenpairs $ (k_{n},f_{n}) $
of $\Gamma_{\lv}$, which we denote by
$\tilde{h}(\Gamma_{\lv},k_n,f_n)$.
\begin{defn}
	We say that the spectral observable $\tilde{h}$ has an \emph{oracle function} $ h_{\lv}:\mg\to\R $, if for
	every eigenpair $(k_n,f_{n})$ of $\Gamma_{\lv}$,
	\[
	\tilde{h}(\Gamma_{\lv},k_n,f_n) = h_\lv\big(\fr{k_{n}\lv}\big).
	\]
\end{defn}
 
Examples of observables having an oracle function include the eigenvalues spacing $ n\mapsto (k_{n+1}-k_{n}) $ discussed in \cite{BarGas_jsp00}, eigenfunction
statistics \cite{BerWin_tams10,CdV_ahp15}, band widths in the
continuous spectrum of periodic
graphs \cite{BanBer_prl13,ExnTur_jpa17} as well as
the nodal surplus of $f_n$ \cite{Ban_ptrsa14,AloBanBer_cmp18} and Neumann surplus \cite{AloBan19}.
Surprisingly,\footnote{An oracle depending only on $\kv$ has no access
  to the label $n$ of the eigenfunction which enters the definition of
  the nodal surplus.  The label is highly sensitive to the changes in
  the edge lengths $\lv$.  Remarkably, taking the difference in
  \eqref{eq:nodal_surplus_def} erases this dependence on $\lv$.}
the oracles for the latter two cases (nodal surplus and Neumann surplus) do not depend on $\lv$.

If $\tilde{h}$ has an oracle function $h_\lv$, then its spectral average
\[
  \left\langle \tilde{h} \right\rangle
  :=
  \lim_{N\rightarrow\infty}\frac{1}{N}\sum_{n=1}^{N}
  \tilde{h}(\Gamma_{\lv},k_n,f_n)
  = \lim_{N\rightarrow\infty}\frac{1}{N}\sum_{n=1}^{N}
  h_{\lv}\big(\fr{k_{n}\lv}\big),
\]
can be replaced by an integral of $h_{\lv}$ over $\mg$ with the
appropriate measure called the Barra-Gaspard (BG) measure.

\begin{defn}
  \label{def: BG}
  \cite{BarGas_jsp00,BerWin_tams10,CdV_ahp15} Let $\Gamma$ be a graph
  with secular manifold $\mg$ and let $\mreg$ be its regular
  part. Denote the volume form\footnote{$ \mreg $ is an orientable Riemannian manifold (with metric inherited by the flat metric on $ \T^{E} $) and as such has a standard volume form.} (or volume measure) of $\mreg$ by $\dd s$ and let
  $\hat{n}$ be the normal vector field of $\mreg$, chosen to have
  non-negative entries\footnote{The normal can be chosen to have
    non-negative entries by \cite[Theorem 1.1]{CdV_ahp15}.}. Then, for
  any $\lv\in\R_{\ge 0}^{E}\setminus\{0\}$ the BG-measure $\mu_{\lv}$
  on $\mg$ is defined by $\mu_{\lv}(\mg\setminus\mreg)=0$ and
  its density on $\mreg$
  \begin{equation}
    \label{eq: BG ds}
    \dd\mu_{\lv}:=\frac{\pi}{L}\frac{1}{\left(2\pi\right)^{E}}(\hat{n}\cdot\lv) \dd s,
  \end{equation}
  where $L=\sum_{e\in\E}\lve$ is the total length. In terms of
  differential forms, the density can be written as
  \begin{equation}
    \label{eq: BG differetial forms}
    \dd \mu_{\lv}=\frac{\pi}{L}\frac{1}{\left(2\pi\right)^{E}}\sum_{j=1}^{E}(-1)^{j-1}\lvj\,d\kappa_{1}\wedge d\kappa_{2}\ldots \wedge\widehat{d\kappa_{j}}\ldots \wedge d\kappa_{E},
  \end{equation}
  where $\widehat{d\kappa_{j}}$ indicates that $d\kappa_{j}$ is omitted.
\end{defn}
\begin{figure}[h!]
  \centering

  \includegraphics[width=0.7\textwidth]{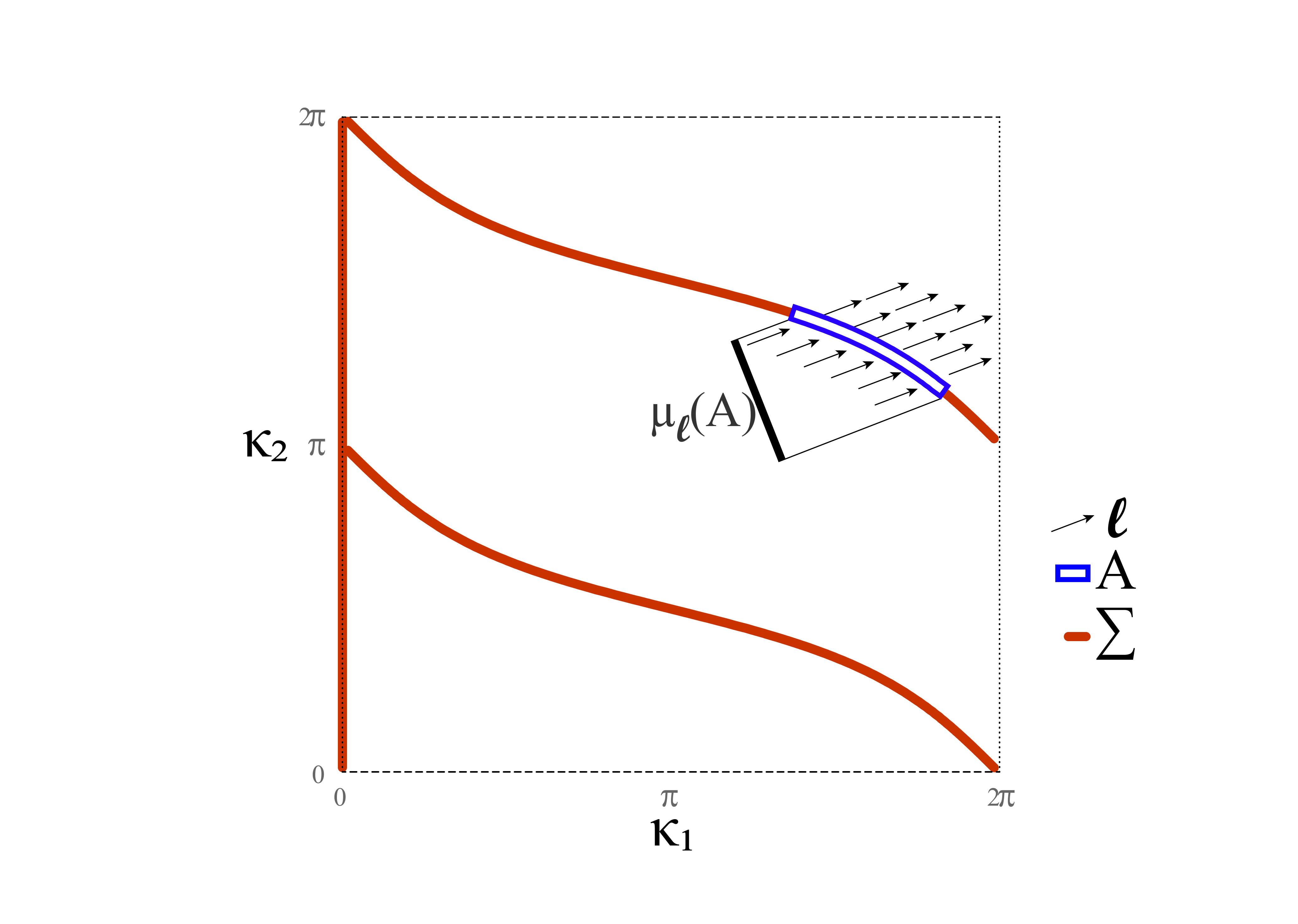}

\caption{Illustration of a set $A\subset\mg$ and its BG-measure $\mu_{\lv}(A)$. The arrows indicate the constant vector field $\lv$.}

\label{fig: BG-measure}
\end{figure}
\begin{rem}
  \label{rem:meaning_BGmeasure}
  Up to a normalizing factor, the measure of $A \subset \mg$ is the
  flux across $A$ of the constant vector field in the direction $\lv$.
  Equivalently, it is the ``cross section'' (in physics terminology) of $A$ as seen from the direction
  $\lv$, see Figure \ref{fig: BG-measure}.  The
  normalization is chosen so that the BG-measures are Borel
  probability measures on $\mg$.
\end{rem}

The next theorem states that for rationally
independent $\lv$, the discrete measures that average over the first $ N $ points of $ \left\{ \fr{k_{n}\lv} \right\}_{n\in\N} \subset \mg$ converge to $\mu_\lv$ when $ N\to\infty $.

\begin{thm}
  \label{thm: equidistirbution}
  \cite{BarGas_jsp00,BerWin_tams10,CdV_ahp15} Let $\Gamma$ be a graph
  with secular manifold $\mg$. Let $\Gamma_{\lv}$ be a standard
  graph with (square-root) eigenvalues $k_{n}$ for $n\in\N$. If $\lv$
  is rationally independent, then for any Riemann
  integrable\footnote{By Riemann integrable we mean that its
    discontinuity set has measure zero.} $h:\mg\rightarrow\C$,
  \begin{equation}\label{eq: equidistribution}
    \lim_{N\rightarrow\infty}\frac{1}{N}\sum_{n=1}^{N}
    h\big(\fr{k_{n}\lv}\big)
    =\int_{\mg}h \dd \mu_{\lv}.
\end{equation}
\end{thm}

Although we are motivated by Theorem \ref{thm: equidistirbution} which only applies to rationally independent $\lv$'s, a prominent role will be played by the following measures with rational $\lv$.

\begin{defn}\label{def: mue}
    \label{defn:specialBGmeasures}
    Given $e\in\E$, let $\mu_{e}$ denote the BG-measure corresponding to $\lv=(0,\ldots,1,\ldots,0)$, with 1 in the $e$-th position.  We also use $\mu_\ones$ to denote the BG-measure corresponding to $\lv=\ones:=(1,\ldots,1)$.
\end{defn}
As can be seen immediately from Definition \ref{def: BG},
for any non-zero lengths $\lv$,
\begin{equation}\label{eq: BG mue}
    \mu_{\lv}=\frac{1}{L}\sum_{e\in\E}\lve \mu_{e}.
\end{equation}
In other words, the set of all BG-measures is the convex hull of the
$\mu_e$'s, and so has the structure of a finite convex polytope (with
$ n\le E $ vertices). We may therefore interpret
$\mu_\ones=\frac{1}{E}\sum_{e\in\E}\mu_{e}$ as 
the average BG-measure (or the middle point of the polytope).
We can also say that the BG-measures
with rationally independent $\lv$ are dense in the set of BG-measures (in the same way that the rationally independent points of a polytope in $ \R^{n} $ are dense).

\subsection{BG method as a torus integral}
\label{sec:torus_average}

Theorem \ref{thm: equidistirbution} provides an analytic tool to
investigate spectral averages. However, an integral over an implicit
high-dimensional hypersurface, such as $\mg$, may be harder to
compute than the spectral average itself. For this reason we introduce
Theorem \ref{thm: BGmeasure as spectral measures} which describes the
BG-measure in terms of the unitary evolution matrices, $ U_{\kv} $ for $ \kv\in{\T^{E}} $. Using this theorem, integrals against a BG-measure can be evaluated efficiently by sampling random points uniformly on
$\T^{E}$. This is a new result, generalizing, in particular,
\cite[Thm.~3.4]{BerWin_tams10}.

In what comes next, rather than talking about eigenvalues of unitary matrices 
$z = e^{it}$, where $ e^{it}\in S^{1}:=\set{z\in\C}{|z|=1} $, it will be more
convenient to use their \emph{eigenphases} $t \in \T :=
\R/2\pi\Z$.  We will also use the following notational convention.
\begin{defn}
  \label{def: measures and phi}
  Consider the \emph{diagonal action} of the abelian group $\T$ on
  $\T^{E}$ defined by
  \begin{equation}
    \label{eq:diag_action_def}
    t \colon \kv \mapsto \fr{\kv+t\ones}, \qquad t\in \T,
  \end{equation}
  where $\ones:=(1,1,\ldots )$.  With a slight abuse of notation, we will
  denote this action simply as $\kv + t$.  We denote the orbit of
  $\kv$ under this action by $\kv+\T$.
\end{defn}

Combining the above notation with Definition
\ref{def:The-unitary-evolution} gives
\begin{equation}
  \label{eq:Ukv-t}
  U_{\kv-t} = e^{-i t}U_{\kv}.
\end{equation}
In particular, $e^{it}$ is an eigenvalue of $U_{\kv}$ if and only if 1
is an eigenvalue of $U_{\kv-t}$, yielding
\begin{equation}
  \label{eq:ukv}
  e^{i t} \text{ is an eigenvalue of } U_{\kv} \iff \kv-t \in\mg.
\end{equation}
\begin{thm}
  \label{thm: BGmeasure as spectral measures} Let $\Gamma$ be a (discrete) graph
  and let $\lv\in\R_{\ge 0}^{E}\setminus\{0\}$. At every $\kv\in\T^{E}$, denote the eigenphases and (orthonormal) eigenvectors of $ U_{\kv} $ by $\theta_{n}(\kv)$ and $\ba_{n}(\kv)$ for $n=1,2,\ldots,2E$. Then for any measurable $h:\mg\to\C$
  we have
  \begin{equation}
    \label{eq: integral BG measure}
    \int_{\mg}h \dd \mu_{\lv}
    = \int_{\T^{E}}\sum_{n=1}^{2E} h\big(\kv-\theta_{n}\big)
      \frac{\ba_{n}^{*}\boldsymbol{L}\ba_{n}}{\tr(\boldsymbol{L})}
      \frac{\dd \kv}{(2\pi)^E},
  \end{equation}
  where the $\kv$ dependence of $\theta_{n}$ and $\ba_{n}$ is omitted for brevity and $\boldsymbol{L}:=\emph{\diag}(\lv,\lv)$.
\end{thm}

\begin{rem}\label{rem: average over orbits}
  The integrand in the right-hand-side of \eqref{eq: integral BG measure} is a weighted average of $h$ over the intersection points $\{\kv+\T\}\cap\mg=\{\kv-\theta_{n}\}_{n=1}^{2E}$, see Figure \ref{fig: T orbit} for example. Implicit in \eqref{eq: integral BG measure} is the claim that the
  integrand is measurable on $\T^E$.
\end{rem}

\begin{rem}
  Consider the spectral decomposition of $U_\kv$,
  \[
    U_{\kv} = \sum_{e^{i\theta} \in \spec(U_{\kv})} e^{i\theta} P_\theta,
  \]
  where the sum is over \emph{distinct} eigenvalues $e^{i\theta}$ of $U_\kv$ and
  $P_\theta$ is the spectral projection, i.e. the orthogonal projection onto the
  eigenspaces $ \ker(e^{i\theta}-U_{\kv}) $. In this form, Theorem \ref{thm: BGmeasure as spectral measures} can be stated as
  \begin{equation}
    \int_{\mg}h \dd \mu_{\lv}
    \label{eq: integral BG measure alt}
    = \int_{\T^{E}} \sum_{e^{i\theta} \in \spec(U_{\kv})}
    h\big(\kv-\theta(\kv)\big)
    \frac{\tr\big(P_\theta(\kv)\boldsymbol{L}\big)}{\tr(\boldsymbol{L})}
    \frac{\dd \kv}{(2\pi)^E}
  \end{equation}
  which highlights its independence of the numbering of eigenphases and
  the choice of the eigenvectors.
\end{rem}

The proof of Theorem \ref{thm: BGmeasure as spectral measures} will be
partitioned into two steps. The first step would be to establish the
Theorem for one particular BG-measure, $\mu_{\textbf{1}}$ with
$\lv=\textbf{1}:=(1,1,1\ldots )$. The second step would be to extend
this result to every $ \mu_{\lv} $ by calculating the Radon-Nikodym
derivative of $ \mu_{\lv} $ with respect to $ \mu_{\textbf{1}}$.

To study the integral over $\Sigma$, we will partition $\Sigma$ into
``layers'', with the $n$-th layer defined as the set of $\kv \in \T^E$
where the $n$-th eigenvalue of $U_\kv$ is equal to 1.  However, as can
be seen from Proposition~\ref{prop: continuous eigenvalues}, the eigenvalues cannot be ordered counterclockwise continuously throughout $\T^E$.  This fact necessitates the following
construction.

Let $X = \opcl{0,2\pi}^{E-1}$ which we will identify with the subset
$\{\kv\in\T^E: \kappa_E=0\}$ of $\T^E$ via the mapping
$\xv\in X \mapsto (\xv,0)\in\T^E$.  Note that we do not identify the
sides of the cube $X$.  In a slight abuse of notation we will write
$U_\xv$ with $\xv \in X$ instead of $U_{(\xv,0)}$ and will use
a similar shorthand for the eigenphases and eigenvectors of
$U_{(\xv,0)}$.

We define the
\emph{diagonal projection} $\p: \T^E \to X$ by
\begin{equation}
  \label{eq:diag_projection}
  \p : (\kappa_1, \kappa_2, \ldots , \kappa_E) \mapsto
  (\kappa_1-\kappa_E, \kappa_2-\kappa_E, \ldots ,
  \kappa_{E-1}-\kappa_E),
\end{equation}
see Figure \ref{fig: T orbit} for illustration.  By definition of the
matrix $U_\kv$, see~\eqref{eq:Udef}, we have
\begin{equation}
  \label{eq:diagonal_proj_U}
  U_\kv = e^{i \kappa_E} U_{\p(\kv)},
\end{equation}
Therefore, $U_\kv$ and $U_{\p(\kv)}$ share the same eigenspaces and
their eigenphases are related by a shift by $\kappa_E$.  More
precisely, $\ba$ is an eigenvector of $U_{\p(\kv)}$ with eigenvalue
$e^{i\theta}$ if and only if $\ba$ is an eigenvector of $U_\kv$ 
with eigenvalue $e^{i(\theta + \kappa_E)}$.

\begin{lem}\label{lem: structure}
  Consider a graph $\Gamma$ of $E$ edges. There exist continuous
  functions $\theta_n \colon X \to \T$, $n=1,\ldots ,E$, such that the
  spectrum of $U_\xv$ for any $\xv\in X$ is given by
  $\spec(U_\xv) = \left\{ e^{i\theta_n(\xv)} \right\}_{n=1}^E$.
  When $e^{i\theta_n(\xv)}$ is a simple eigenvalue, $\theta_n(\xv)$
  and its corresponding eigenvector $\ba_n(\xv)$ are analytic in some
  neighborhood of $\xv$.
\end{lem}


\begin{proof}
  The family $U_{\xv}$ is analytic in $\xv$, and the domain $X$ is
  simply connected, therefore the lemma is a
  direct consequence of Propostion~\ref{prop: continuous eigenvalues}
  of Appendix~\ref{sec: appendix numbering eigenvalues}, as well as
  classical results of perturbation theory \cite{Rellich1969perturbation}.
\end{proof}

Using relation \eqref{eq:diagonal_proj_U}, we extend the functions
$\theta_n$ from $X$ to the whole of $\T^E$,
\begin{equation}
  \label{eq:theta_def}
  \theta_n: \T^E\to\T
  \qquad
  \theta_n(\kv) := \theta_n\big(\p(\kv)\big) + \kappa_E,
  \qquad
  \kv = (\kappa_1,\ldots ,\kappa_E).
\end{equation}
Consequently, the spectrum of $U_\kv$ for any $\kv\in\T^E$ is given by
  \begin{equation}
    \label{eq:spec_Uk}
    \spec(U_\kv) = \left\{ e^{i\theta_n(\kv)} \right\}_{n=1}^E.
  \end{equation}
Note that $\theta_n$ are continuous only on $\T^E \setminus
\partial\Omega$, where $ \Omega $ is the embedding of $ X\times\T $ into $ \T^{E} $ (see Figure \ref{fig: XxT} for example) and 
\begin{equation*}
  \partial\Omega := \{ \kv\in \T^E \colon \kappa_j = \kappa_E \text{
    for some } j, \ 1\leq j \leq E-1\}.
\end{equation*}

\begin{figure}[h!]
  \centering
  \includegraphics[width=0.7\textwidth]{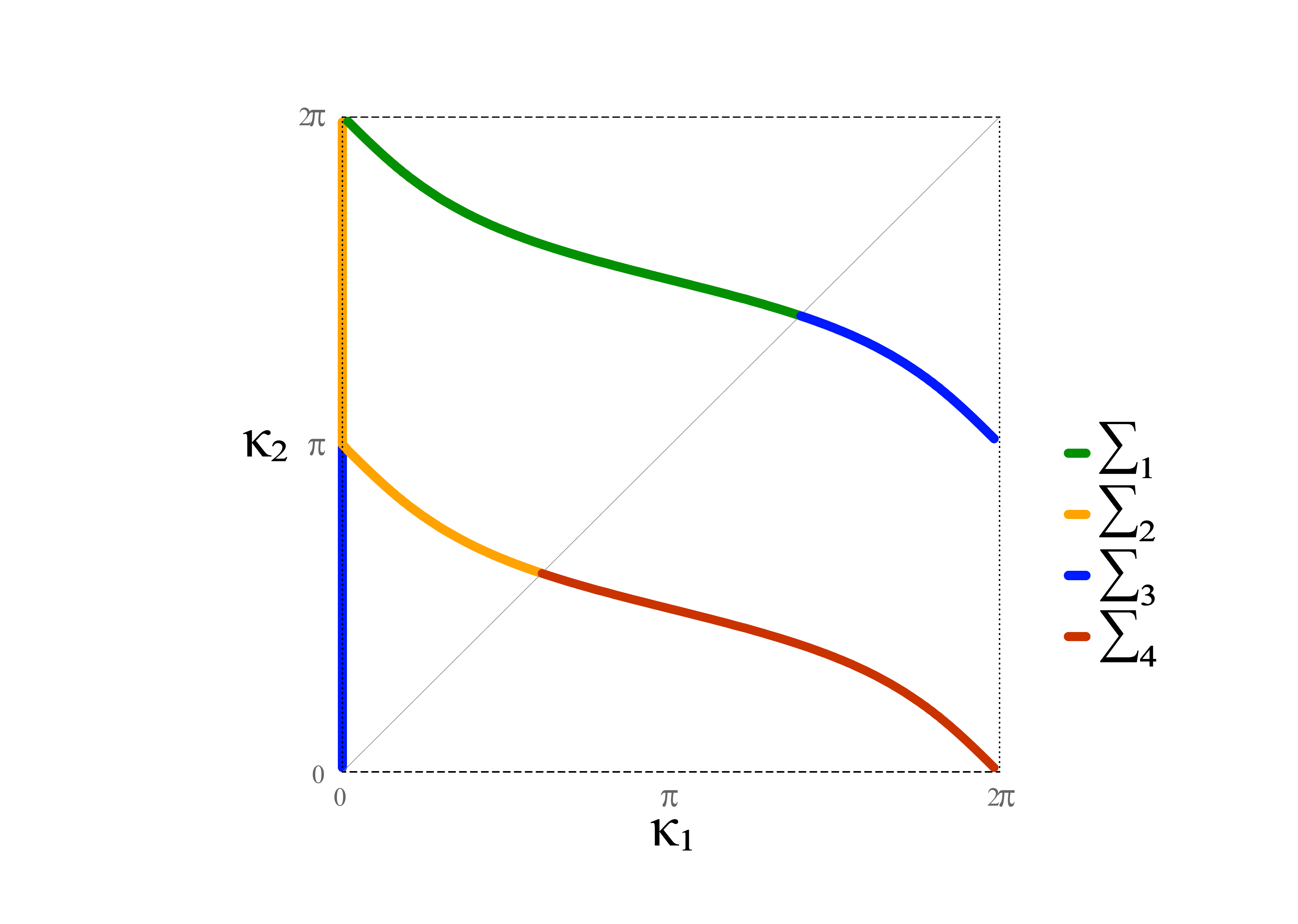}
\caption{A secular manifold colored according to its layers $\mg_{n}$. Here $ \partial\Omega $ is the diagonal line $ \kappa_{1}=\kappa_{2} $. The coloring is discontinuous on $ \partial\Omega $ and the singular points $ \msing=\{(0,\pi),(0,2\pi)\} $ (and their identifications).}
\label{fig: layers structure}
\end{figure}

\begin{figure}[h!]
  \centering
  \includegraphics[width=1\textwidth]{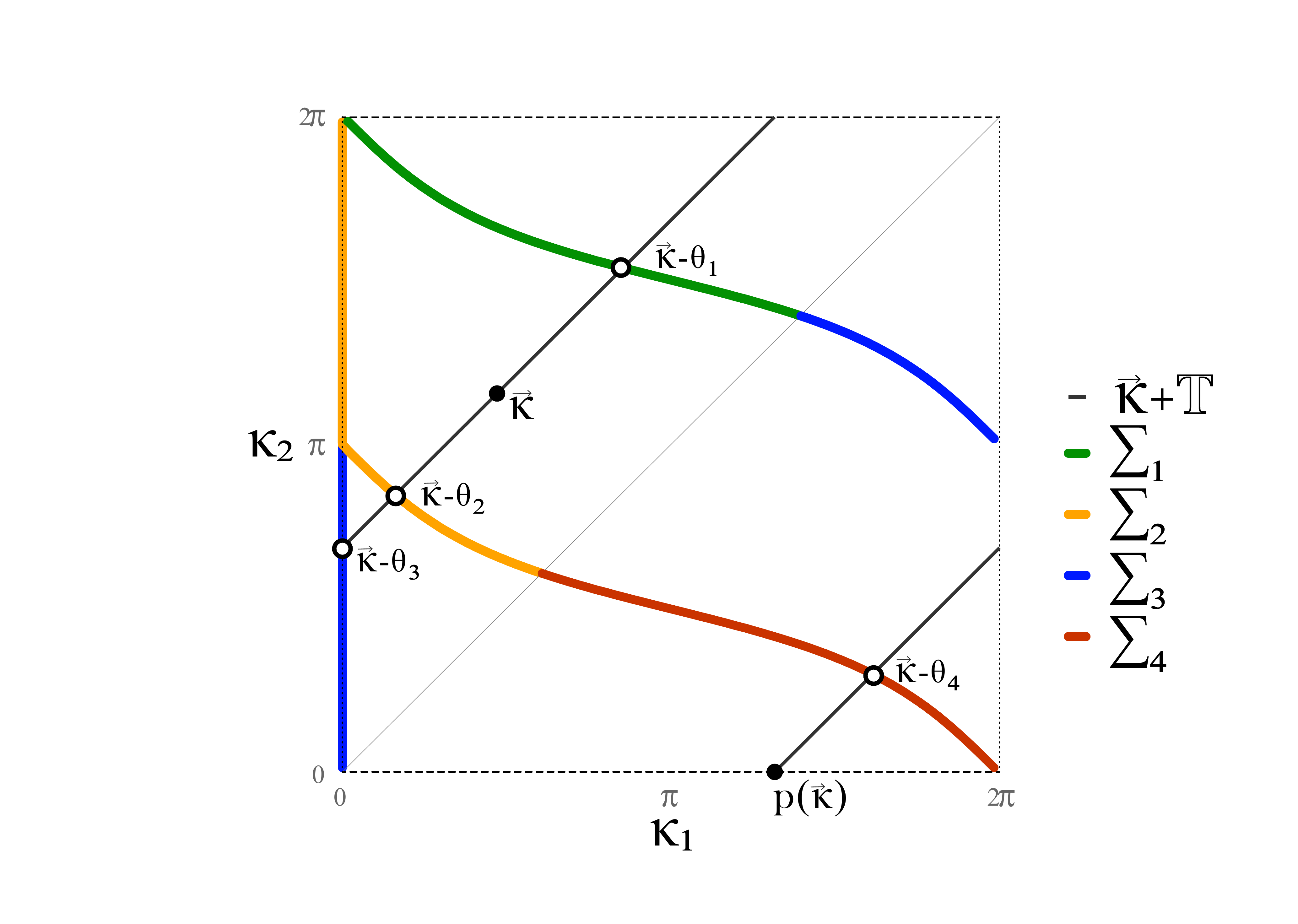}
\caption{A point $\kv\in\T^{2}$ and its diagonal orbit $\kv+\T$. The
  (embedding of the) point $\p(\kv)$ is presented and the intersection
  points $\{\kv+\T\}\cap\mg=\{\kv-\theta_{1},\ldots,\kv-\theta_{4}\}$
  are shown to match the layer structure numbering. }
\label{fig: T orbit}
\end{figure}

\begin{figure}[h!]
  \centering
  \includegraphics[width=1\textwidth]{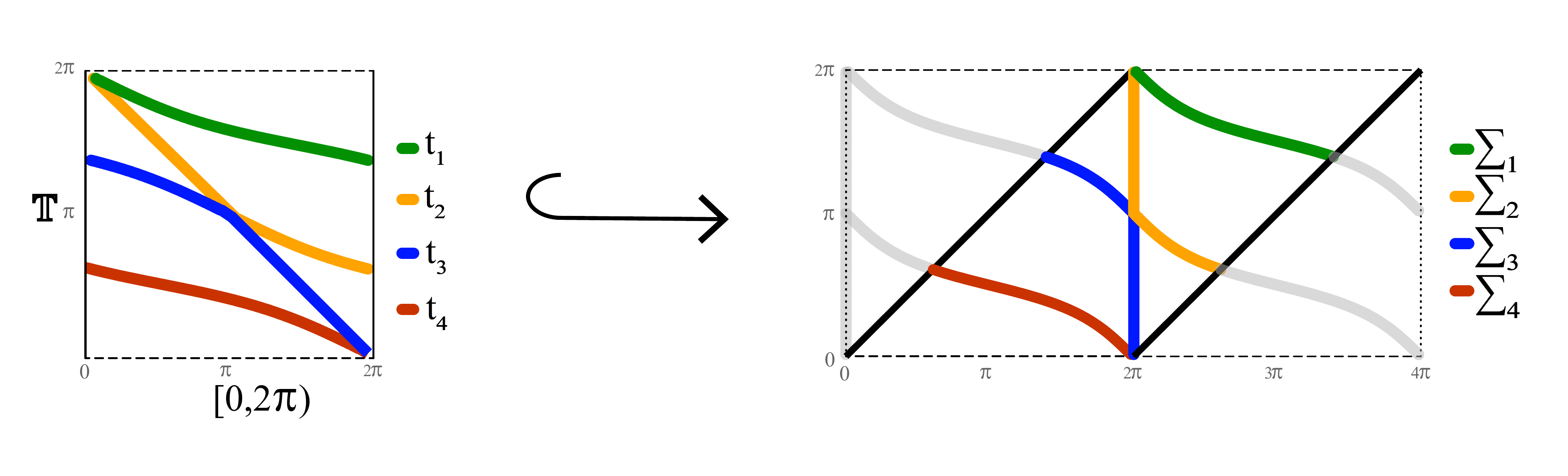}
\caption{The embedding $[0,2\pi)\times\T\hookrightarrow \T^2$. (left) $[0,2\pi)\times\T$. The graphs of the functions $t_{n}(\xv)=-\theta_{n}(\xv)$ are displayed in different colors. (right) Illustration of $ \Omega $, the embedding of $[0,2\pi)\times\T$ into $\T^{2}$, with black margins indicating the boundary $ \partial\Omega $. The layers $\mg_{n}$, coming from the graphs of $t_{n}$, are displayed in different colors. }
\label{fig: XxT}
\end{figure}

\begin{lem}\label{lem: layers are cubes}
  With $\theta_n: \T^E\to \T$ defined by \eqref{eq:theta_def}, let
  \begin{equation}
    \label{eq:Sigman_def}
    \mg_{n}:=\set{\kv\in\T^{E}}{e^{i\theta_{n}(\kv)}=1},\qquad n=1,2,\ldots,2E.    
  \end{equation}
  Then $\mg_{n}$ have the following properties.
  \begin{enumerate}
  \item
    \begin{equation}
      \label{eq:unionSigma_n}
      \Sigma = \bigcup_{n=1}^{2E}\Sigma_{n}, \qquad\mbox{and}\qquad
      \msing = \bigcup_{1\le{m}<n\le{2E}}\Sigma_{n}\cap\Sigma_{m}.
    \end{equation}
  \item The map 
    \begin{equation}
      \label{map_kvn_def}
      \kv_{n}(\xv) := (\xv,0) - \theta_{n}(\xv),
    \end{equation}
    is a bijection $X \to \mg_{n}$ which is is differentiable except
    possibly for a set of positive co-dimension in $X$.
  \item For any measurable $h$ on $ \mg_{n} $,
    \begin{equation}
      \label{eq:integral1_sigman}
      \int_{\mg_{n}}h\dd\mu_{\textbf{1}}
      =\frac{1}{2E}\int_{\xv\in X}h(\kv_{n}(\xv))\frac{\dd \xv}{(2\pi)^{E-1}}.
    \end{equation}
    Namely, the push forward of the Lebesgue measure on
    $X=\opcl{0,2\pi}^{E-1}$ by $ \kv_{n} $ is the restriction of the
    BG-measure $ \mu_{\textbf{1}} $ to $ \mg_{n} $, up to normalization.    
  \end{enumerate}
\end{lem}

See Figure \ref{fig: layers structure} for an example of the layer
structure given by \eqref{eq:unionSigma_n}, and Figure \ref{fig: XxT} for the construction of the Layer structure by embedding $ X\times\T $ into $ \T^{E} $.

\begin{proof}
  The first part follows immediately from \eqref{eq:spec_Uk} and Remark \ref{rem:Sigma_sets_properties}.

  To show the second part we combine \eqref{eq:Sigman_def} and
  \eqref{eq:theta_def}: the latter can be rewrtten as
  \begin{equation*}
    \theta_n\big( (\xv,0) + \kappa_E \big) = \theta_n(\xv) + \kappa_E,
  \end{equation*}
  which is equal to 0 if and only if
  $\kappa_E = -\theta_n(\xv) \in \T$. The inverse of $\kv_{n}$ is
  the projection $\p$: $\p(\kv_{n}(\xv)) = \xv$.  The map $\kv_{n}$
  may not be differentiable at $\xv$ if $ \theta_{n}(\xv) $ is not differentiable there, which can only happen if $e^{i\theta_n(\xv)}$ is
  a multiple eigenvalue.  The set of such $\xv$ is obtained by
  projecting $\msing \cap \Sigma_n$ to $X$ by $\p$. Positive
  co-dimension of $ \p(\msing \cap \Sigma_n) $ in $ X $ follows from
  \[\dim(\p(\msing \cap \Sigma_n))\le\dim(\msing \cap \Sigma_n)\le\dim(\msing)\le E-2,\]
  where the last inequality is stated in Lemma~\ref{lem: secular equation and
    isomorphism}.
	
  To prove~\eqref{eq:integral1_sigman} we consider $ \kv_{n}(\xv) $ as
  a parameterization of $\mg_{n}$ by $  X $. We claim
  that wherever it is differentiable, the parameterization
  $\kv=\kv_{n}(\xv)$ satisfies the identity
  \begin{equation}
    \label{eq: alternating sum}
    \sum_{j=1}^{E}(-1)^{j}\left|\frac{
        \partial\left(\kappa_1,\kappa_2,\ldots ,\widehat{\kappa_j},\ldots ,\kappa_E
        \right)}{\partial\left(x_1,x_2,\ldots ,x_{E-1}\right)}\right|=1.
  \end{equation}
  We will now prove \eqref{eq:integral1_sigman} assuming \eqref{eq:
    alternating sum} and will provide the proof of \eqref{eq:
    alternating sum} later.

  The density of the BG measure $\mu_{\textbf{1}}$ (as described in  \eqref{eq: BG differetial forms}) can be written in terms of $\kv=\kv_{n}(\xv)$ as  
  \begin{align*}
    d\mu_{\textbf{1}}
    &= \frac{\pi}{E}\frac1{(2\pi)^{E}}\sum_{j=1}^{E}(-1)^{j-1}d\kappa_1\wedge
      d\kappa_2 \wedge \cdots \widehat{d\kappa_j} \cdots\wedge d\kappa_E\\
    &= \frac{1}{2E}\frac1{(2\pi)^{E-1}}\sum_{j=1}^{E}(-1)^{j-1}
      \left|\frac{
      \partial\left(\kappa_1,\kappa_2,\ldots ,\widehat{\kappa_j},\ldots ,\kappa_E
      \right)}{\partial\left(x_1,x_2,\ldots ,x_{E-1}\right)}\right|
      dx_1 \wedge dx_2\wedge \cdots \wedge dx_{E-1}\\
    &=\frac1{2E}\frac1{(2\pi)^{E-1}}dx_1\wedge dx_2\wedge\cdots\wedge dx_{E-1},
  \end{align*}
  where in the last equality we used the identity \eqref{eq:
    alternating sum}. This establishes~\eqref{eq:integral1_sigman}. 

	Let us prove identity \eqref{eq: alternating sum}. Let $ D\kv_{n}(\xv) $ be the $ E\times E-1 $ matrix of derivatives of $\kv_{n}(\xv)$, and construct $ M $, an $ E\times E $ matrix, by adding to $ D\kv_{n}(\xv) $ an auxiliary column of ones, so that
	\[\det(M)=\sum_{j=1}^{E}(-1)^{j}\left|\frac{
		\partial\left(\kappa_1,\kappa_2,\ldots ,\widehat{\kappa_j},\ldots ,\kappa_E
		\right)}{\partial\left(x_1,x_2,\ldots ,x_{E-1}\right)}\right|,\]
	by expansion of the determinant over the column of ones. To prove identity \eqref{eq: alternating sum} we need to show that $ \det(M)=1 $. The
	derivatives of $\kv_{n}(\xv)$ are
	\begin{equation}
		\frac{\d \kappa_{i}}{\d x_{j}}=\delta_{i,j}+\frac{\d
			\theta_{n}(\xv)}{\d x_{j}},
		\qquad i=1,\ldots ,E,\quad j=1,\ldots ,E-1,
	\end{equation}
so $ M $ is given by 
	\begin{equation}
		\label{eq:Mdef}
		M =
		\begin{pmatrix}
			1+ \frac{\partial \theta_n}{\partial x_1} & \frac{\partial \theta_n}{\partial x_2}
			& \cdots & 1 \\[4pt]
			\frac{\partial \theta_n}{\partial x_1} & 1+\frac{\partial \theta_n}{\partial x_2}
			& \cdots & 1 \\
			\vdots & \vdots & & \vdots \\
			\frac{\partial \theta_n}{\partial x_1} & \frac{\partial \theta_n}{\partial x_2}
			& \cdots & 1
		\end{pmatrix}.
	\end{equation}
	To see that the determinant of this matrix is 1 subtract the last row from all others.
\end{proof}

We may now use Lemma \ref{lem: layers are cubes} to prove Theorem \ref{thm: BGmeasure as spectral measures}.

\begin{proof}[Proof of Theorem \ref{thm: BGmeasure as spectral
    measures}]
  We start with the special case $\lv = (1,\ldots,1)$ when
  $\boldsymbol{L}=\diag(\lv,\lv)$ is the identity matrix, and
  so the factor 
  \[
    \frac{\ba_{n}(\kv)^{*}\boldsymbol{L}\ba_{n}(\kv)}{\tr(\boldsymbol{L})}\equiv\frac{1}{2E},
  \]  
  is independent of $\kv$, due to $\|\ba_{n}\|=1$. In this case,
  equation \eqref{eq: integral BG measure} reduces to
  \begin{equation}\label{eq:integralBGmeasure1again}
    \int_{\mg}h \dd \mu_{\textbf{1}}
    = \int_{\T^{E}}\frac{1}{2E}\sum_{n=1}^{2E} h\big(\kv-\theta_{n}\big)
    \frac{\dd \kv}{(2\pi)^E}.
  \end{equation}
  Focusing on one term in the sum on the right hand side, we have
  \begin{equation}
    \label{eq:torus_int_step}
    \frac{1}{2E}\int_{\T^{E}} h\big(\kv-\theta_{n}\big)
    \frac{\dd \kv}{(2\pi)^E}
    = \frac{1}{2E}\int_{t\in\T}\int_{\xv\in X}
    h\Big(\kv(\xv,t)-\theta_{n}\big(\kv(\xv,t)\big)\Big)
    \frac{\dd \xv\dd t}{(2\pi)^E},
  \end{equation}
  where we made the substitution $\kv(\xv,t):=(\xv,0)+t$ which has
  Jacobian equal to $1$ (the latter is an easy observation).  For
  every $ t\in\T $ and $ \xv\in X$ we have
  \begin{equation*}
    \kv(\xv,t)-\theta_{n}(\kv(\xv,t))
    = (\xv,0) + t - \big(\theta_{n}(\xv)+t\big)
    = (\xv,0) - \theta_{n}(\xv) = \kv_{n}(\xv),
  \end{equation*}
  where $\kv_{n}(\xv)$ was introduced in Lemma~\ref{lem: layers are
    cubes}.  In particular, the integrand of the right hand side of
  \eqref{eq:torus_int_step} is independent of $t$.  Integrating $t$
  out, we obtain
  \begin{equation*}
    \frac{1}{2E}\int_{\T^{E}} h\big(\kv-\theta_{n}\big)
    \frac{\dd \kv}{(2\pi)^E}
    = \frac{1}{2E} \int_{\xv\in X}
    h\big(\kv_n(\xv)\big)
    \frac{\dd \xv}{(2\pi)^{E-1}}
    = \int_{\mg_{n}}h \dd \mu_{\textbf{1}},
  \end{equation*}
  where we applied Lemma~\ref{lem: layers are cubes} in the last
  step.

  Finally, the summation of $\mg_{n}$ integrals is an integral over
  $\mg$ since the layers cover $\mg$ and intersect on $\msing$
  (Lemma~\ref{lem: layers are cubes}) which has measure zero
  (Lemma \ref{lem: secular equation and isomorphism}).
  To conclude, we have established \eqref{eq: integral BG measure} for $\mu_{\textbf{1}}$.

  To prove \eqref{eq: integral BG measure} for $ \mu_{\lv} $ with
  arbitrary $\lv$, of total length $ L=\sum_{j=1}^{E}\ell_{j} $, we
  use
  \begin{equation}
    \label{eq:RadonNikodym_sub}
    \int_{\mg}h \dd \mu_{\lv}
    = \int_{\mreg}h \dd \mu_{\lv}
    = \int_{\mreg}h
    \left|\frac{\hat{n}\cdot\lv}{\hat{n}\cdot\ones}\right|
    \frac{E}{L} d\mu_{\textbf{1}}.
  \end{equation}
  The first equality is due to $ \mreg $ being a full measure subset, by Lemma \ref{lem: secular equation and isomorphism}. In the second equality we used Definition \ref{def: BG} to compute the Radon--Nikodym
  derivative of $\mu_{\lv}$ with respect to $\mu_{\ones}$ on
  $\mreg$.  According to \cite[Thm.~1.1]{CdV_ahp15}, the normal vector
  $\hat{n}$ at a regular point $\kv\in \mreg$ is related to
  $\ba$, the normalized eigenvector of eigenvalue 1 of
  $U_\kv$, by
  \[
    \hat{n}_{e}
    \propto \left|\ba_{e}\right|^2
    + \left|\ba_{\hat{e}}\right|^2.
  \]
  It follows, using the normalization $\|\ba\|=1$, that
  \[
    \left|\frac{\hat{n}\cdot\lv}{\hat{n}\cdot\ones} \right|
    = \frac{\sum_{e} \ell_e \left(\left|\ba_{e}\right|^2
        + \left|\ba_{\hat{e}}\right|^2\right)}
    {\sum_e \left(\left|\ba_{e}\right|^2
        + \left|\ba_{\hat{e}}\right|^2\right)}
    =\frac{\ba^*\boldsymbol{L}\ba}{\|\ba\|^{2}}= \ba^*\boldsymbol{L}\ba.
  \]
  Recall that if $ \kv\in\mg_{n}\cap\mreg $ then the normalized eigenvector of the simple eigenvalue $ 1 $ is $ \ba_{n}(\kv) $ which is analytic for every $ \kv\notin\partial\Omega $ around that point by Lemma \ref{lem: structure}. Finally, we represent $2L = \tr(\boldsymbol{L})$, so that
  \[\int_{\mg}h \dd \mu_{\lv}
    =2E\int_{\mreg}h(\kv)\frac{\ba_{n}(\kv)^{*}\boldsymbol{L}\ba_{n}(\kv)}{\tr(\boldsymbol{L})}
    d\mu_{\textbf{1}}.\] The theorem now follows by
  applying~\eqref{eq:integralBGmeasure1again} to the measurable
  function
  \[\kv\mapsto
    h(\kv)\frac{\ba_{n}(\kv)^{*}\boldsymbol{L}\ba_{n}(\kv)}{\tr(\boldsymbol{L})},\]
  and replacing $ \ba_{n}(\kv-\theta_{n}) $ with $ \ba_{n}(\kv) $, see
  \eqref{eq:diagonal_proj_U} and the discussion following it.
\end{proof}

\section{The nodal surplus distribution}\label{sec: nodal surplus distribution}

The purpose of this section is to prove Theorem \ref{thm: sampling
  simplified} and its generalization to graphs with loops, Theorem
\ref{thm: generalized theorem}.  These theorems are established in
section~\ref{subsec: proof of main theorem} by applying Theorems~\ref{thm:
  equidistirbution} and \ref{thm: BGmeasure as spectral measures} to a
suitably defined oracle function
$\bsigma:\mg\rightarrow\{0,1,\ldots \beta\}$ which encodes the nodal
surplus on the secular manifold.  More precisely, the function
$\bsigma$ has the following
properties:
\begin{enumerate}
\item It is Riemann integrable.
\item For any $\lv\in\R_{+}^{E}$ and any $k_{n}$ eigenvalue of
  $\Gamma_{\lv}$ with $n\in\G$, the nodal surplus equals:
  \[
    \sigma(n)=\bsigma\big(\fr{k_{n}\lv}\big).
  \]
\end{enumerate}
Such a function was constructed in \cite{AloBanBer_cmp18,Ban_ptrsa14},
but in section~\ref{subsec: nodal magnetic} we give an alternative
definition in terms of the Hessian of an eigenphase $\theta$ of
$U_{\kv}$ and, in section \ref{subsec:hessian}, express the Hessian
in terms of the corresponding eigenvector $\ba$ and the inverse
Cayley transform of $e^{-i\theta}U_{\kv}$.

\subsection{The nodal surplus oracle function}
\label{subsec: nodal magnetic}

In this section we describe the nodal surplus function $\bsigma:\mg\rightarrow{\{0,1,\ldots \beta\}}$, originally introduced in \cite{AloBanBer_cmp18,Ban_ptrsa14}, and express it in terms of the layers and eigenphases, introduced in Lemma \ref{lem: layers are cubes}.  This will make $\bsigma$ compatible with Theorem \ref{thm: BGmeasure as spectral measures}.

To give a preview, the function $\bsigma$ is constructed as follows.
We introduce an extra set of parameters $\vec\alpha$ in the definition
of the unitary evolution matrix $U = U_{\kv;\av}$, see
Definition~\ref{def:The-unitary-evolution-alpha} below.  According to
a theorem known as the ``Nodal--Magnetic Connection'', established in
\cite{Ber_apde13,Col_apde13,BerWey_ptrsa14} and quoted below as
Theorem~\ref{thm: nodal magnetic}, the nodal surplus can be recovered
by counting the negative eigenvalues of the matrix of second
derivatives of the eigenphase $\theta$ of $U$ as a function of
$\vec\alpha$.  We remark that the word ``magnetic'' appears in the
name of the theorem because the parameters $\vec\alpha$ can be
understood as the fluxes in a magnetic Schr\"odinger operator.  It can
equivalently be viewed as twisting the graph Laplacian by the
characters of its homology group \cite{phiSar87}.  However, neither
interpretation will be particularly important in our calculations
below.

To put this discussion on a rigorous basis, recall that the nodal
surplus is defined only for generic eigenfunctions, see Definition
\ref{def: generic_index}.  We now describe the submanifold
$\mgen \subset \mg$ corresponding to such eigenfunctions.

\begin{lem}\cite{AloBanBer_cmp18}\label{lem: mreg}
  Given a graph $\Gamma$, there is a sub-manifold $\mgen\subset\mreg$
  such that for every $\lv\in\R_{+}^{E}$, the index set $\G$ of
  $\Gamma_{\lv}$ is characterized by
  \begin{equation}
    \label{eq: mreg}
    n\in\G \iff \fr{k_{n}\lv}\in\mgen.
  \end{equation}
  If $ \Gamma $ has no loops, then $ \mgen $ has full measure in
  $ \mg $. If $ \Gamma $ has loops, then
  \begin{equation}
    \label{eq:total_measure_generic}
    \mu_{\lv}\left(\mgen\right)=1-\frac{\LL}{2 L},
  \end{equation}
  where $L$ is the total length of $\Gamma_{\lv}$ and $\LL$ is the
  total length of its loops.
\end{lem}

\begin{defn}
\label{def:The-unitary-evolution-alpha} For a given $\av\in\T^{E}$ we define
\begin{equation}
e^{i\check{\alpha}}:=\diag\left(e^{i\alpha_{1}},e^{i\alpha_{2}}\ldots e^{i\alpha_{E}},e^{-i\alpha_{1}},e^{-i\alpha_{2}},\ldots e^{-i\alpha_{E}}\right)
\end{equation}
and further denote
\begin{equation}
  \label{eq: Ualpha}
  U_{\kv;\av} := e^{i\check{\alpha}}U_{\kv}
  = e^{i\check{\alpha}}e^{i\hat{\kappa}}S,
\end{equation}
such that the unitary evolution matrix defined previously in \eqref{eq:Udef} is given as $U_{\kv}:=U_{\kv;0}$.
\end{defn}

\begin{rem}
  Notice that the $ j $-th and $ (j+E) $-th diagonal elements of
  $ e^{i\check{\alpha}} $ are conjugate rather than equal, in contrast
  to $ e^{i\hat{\kappa}} $.
\end{rem}

Consider the eigenphases of $ U_{\kv} $
for every $ \kv\in\T^{E} $, as defined in \eqref{eq:theta_def}
\[\theta_{m}:\T^{E}\to\T, \qquad m=1,2,\ldots,2E.\]
Notice that $ \det(U_{\kv;\av})=\det(U_{\kv;0})=\det(U_{\kv}) $, since
$ \det(e^{i\check{\alpha}})=1 $. By fixing $ \kv\in\T^{E} $ and
applying Proposition \ref{prop: continuous eigenvalues}, we may deduce
that a labeling of the eigenphases of $U_{\kv}=U_{\kv;0}$ extends
continuously as a function of $\av \in \T^E$ to a labeling of
eigenphases of $U_{\kv;\av}$. This defines,
\[\theta_{m}:\T^{E}\times\T^{E}\to\T,\qquad m=1,2,\ldots,2E,\]
such that for any fixed $ \kv\in\T^{E} $, $ \theta_{m}(\kv;\av) $ is a continuous function of $ \av $ and $ \theta_{m}(\kv;0)=\theta_{m}(\kv) $. Moreover, for any $ \kv $ such that $ \theta_{m}(\kv)$ is simple, the extension $ \theta_{m}(\kv,\av) $ is analytic in $ \av $ around $ \av=0 $. This is a standard result of perturbation theory \cite{Rellich1969perturbation} and the fact that $ U_{\kv,\av} $ is analytic in $ \av $. We denote the Hessian of $\theta_{m}$ with respect to $\av$ at the point $(\kv;0)$ by $\hess_{\av}\theta_{m}(\kv)$. Recall the notation $\M(A)$ for the ``Morse index'', the number of strictly negative eigenvalues of $A$.  We can now state the main result of this subsection.

\begin{thm}\label{thm: nodal statistics}
Define the function $\bsigma:\mreg\rightarrow\Z$ by
\begin{equation}\label{eq:sigma_def_new}
  \bsigma(\kv):=\M\big(-\hess_{\av}\theta_{m}(\kv)\big),
\end{equation}
where $m=m(\kv)$ is determined by the condition $\kv\in \mg_{m}\cap\mreg$.
Then,
\begin{enumerate}
\item \label{enu: Sigma and sigma} On $\mgen$, $\bsigma$ gives the
  nodal surplus: for every $\lv\in\R_{+}^{E}$ and every $n\in\G$ of
  $\Gamma_{\lv}$,
  \[
    \sigma(n)= \bsigma\big(\fr{k_{n}\lv}\big).
  \]
\item \label{enu:Riemann_integrable}
  The indicator functions of $\mgen$ and $\mgen \cap
  \bsigma^{-1}(s)$ on $\Sigma$ are Riemann integrable.
\item \label{enu: Sigma and pj} If $\lv\in\R_{+}^{E}$ is rationally
  independent, then the nodal surplus distribution of $\Gamma_{\lv}$
  is given by
  \begin{equation}
    \label{eq: Pj with loops}
    P(\sigma=s) = \frac{\mu_{\lv}(\mgen\cap\bsigma^{-1}(s))}{\mu_{\lv}\left(\mgen\right)},\qquad s=0,1,2,\ldots,\beta.
  \end{equation}
  If we further assume that $\Gamma$ has no loops, this can be
  simplified,
  \begin{equation}
    \label{eq: Pj without loops}
    P(\sigma=s)=\mu_{\lv}(\bsigma^{-1}(s)).
  \end{equation}
\end{enumerate}
\end{thm}

The first part of Theorem \ref{thm: nodal statistics} is similar to
\cite[Thm.~3.4]{AloBanBer_cmp18}, up to the alternative definition of
$\bsigma$. Its proof, given below, follows from the Nodal--Magnetic
Connection, see \cite{Ber_apde13,Col_apde13} and, in the context of
quantum graphs, \cite[Thm.~2.1]{BerWey_ptrsa14}. To avoid introducing
magnetic Schr{\"o}dinger operators, we reformulate
\cite[Thm.~2.1]{BerWey_ptrsa14} as follows.

\begin{thm}\cite{BerWey_ptrsa14}\label{thm: nodal magnetic}
  Let $\Gamma_{\lv}$ be a standard graph with a simple eigenvalue $k_{n}$. Then
  \begin{enumerate}
  \item\label{item:secular_magnetic} $k_{n}$ has an analytic extension around $\av=0$, such that $k=k_{n}(\av)$ is a solution to
    \begin{equation} \label{eq:implicit relation}
      \det(1-U_{k \lv ; \av})=0.
    \end{equation}
    In a small neighborhood of $ \av=0 $, this extension satisfies
    \begin{equation}
      \label{eq:extension_derivative}
      \frac{\partial}{\partial k}\det(1-U_{k \lv ; \av})\ne0.
    \end{equation}
  \item \label{enu: critivcal point} $\av=0$ is a critical point of
    $k_{n}(\av)$, i.e.
    \begin{equation}
      \label{eq:critical_point_cond}
      \frac{\partial k_{n}}{\partial \alpha_{j}}(0) = 0,
    \end{equation}
    for all $j=1,\ldots,\beta$.
  \item \label{enu: Hessian and nodal surplus} Denote the Hessian of $k_{n}(\av)$ at $\av=0$ by $\hess k_{n}(0)$. If $n\in\G$, then
    \begin{equation}\label{eq: dim ker Hess}
      \dim{\ker{\left(\hess k_{n}(0)\right)}}=E-\beta,
    \end{equation}
    and the nodal surplus, $\sigma(n)$, is given by the Morse index of $k_{n}(0)$
    \begin{equation}
      \M(\hess k_{n}(\av)|_{\av=0})=\sigma(n).
    \end{equation}
  \end{enumerate}
\end{thm}

\begin{rem}
  Magnetic Schr{\"o}dinger operators are introduced in
  \cite{BerKuc_graphs,GnuSmi_ap06,KotSmi_ap99} among other sources.
  Part~(\ref{item:secular_magnetic}) of Theorem~\ref{thm: nodal
    magnetic} gives a direct description of their simple eigenvalues.
  In a small departure from \cite[Thm.~2.1]{BerWey_ptrsa14}, we have
  $E-\beta$ ``extra'' parameters $\alpha$ (in the spirit of
  \cite{Col_apde13}) which accounts for the appearance of the kernel
  in \eqref{eq: dim ker Hess}.
\end{rem}

\begin{proof}[Proof of Theorem \ref{thm: nodal statistics}]
Fix $\lv$ and an eigenvalue $ k_{n} $ of $ \Gamma_{\lv} $ with $ n\in\G $. Then $ \kv=\fr{k_{n}\lv}\in\mgen $ and in particular $ \kv\in\mgen\cap\mg_{m} $ for some fixed $ m=1,2,\ldots 2E $, so that
\begin{equation}
	e^{i\theta_{m}(\kv;0)}=1,
	\qquad\text{and}\qquad
	e^{i\theta_{j}(\kv;0)}\ne1, \text{  for any  } j\ne m.
\end{equation}
To be able to take derivatives of $ \theta_{m}(k,\av) $ in both $ k $ and $ \av $, including the case $ \fr{k\lv}\in\partial\Omega $, extend $ \theta_{m} $, locally around $ (\kv,\av)=(\fr{k_{n}\lv},0) $, as an eigenphase of $ U_{\kv,\av} $. Abusing notation for brevity, $ \theta_{m}(k,\av):=\theta_{m}(\fr{k\lv},\av) $ is therefore analytic in $ (k,\av) $ around $ (k_{n},0) $.\\

Using the relation 
\begin{equation}
	\det(1-U_{\kv ; \av})=\prod_{j=1}^{2E}(1-e^{i \theta_j(\kv;\av)}),
\end{equation}
we can determine that the extension of $ k_{n} $ to 
$k_{n}(\av)$, as described in Theorem \ref{thm: nodal magnetic}, is a solution of  
\begin{equation}
	\label{eq:tilde_theta_condition}
	\theta_{m}(k;\av)=0,
	\qquad\text{and satisfies}\qquad
	\frac{\partial}{\partial k}\theta_{m}(k ;\av)\ne0,
\end{equation}
in some $ (k,\av) $ neighborhood of $ (k_{n}(0),0) $. Differentiating the first condition in
\eqref{eq:tilde_theta_condition} gives, for all $j$,
\begin{equation}
  \label{eq: theta dot}
  \frac{\partial}{\partial \alpha_{j}}\theta_{m}(k_{n};\av)
  = -\left(\frac{\partial}{\partial
      k}\theta_{m}(k_{n};\av)\right)
  \frac{\partial k_{n}(\av)}{\partial \alpha_{j}}.
\end{equation}
Substituting $\av=0$ and using \eqref{eq:critical_point_cond}, we
conclude that
\begin{equation}
  \label{eq: theta dot=0}
  \frac{\partial}{\partial \alpha_{j}}\theta_{m}(k_{n};0)
  = 0,\qquad j=1,2,\ldots E.
\end{equation}
Differentiating \eqref{eq: theta dot} with respect to $\alpha_{j'}$,
substituting $\av = 0$, using \eqref{eq:critical_point_cond} and
rearranging gives
\begin{equation}\label{eq: Hess k Hess theta}
  \hess k_{n}(\av)|_{\av=0}
  =
  -\frac{1}{\frac{\partial}{\partial k}\theta_{m}(k_{n};0)}
  \hess_{\av}\theta_{m}(\kv),\qquad \kv=\fr{k_{n}\lv}.
\end{equation}
As $\frac{\partial}{\partial k}\theta_{m}(k
;0) >0$ whenever $ \theta_{m}(k
;0) $ is simple (see
\cite[eq. (83)]{GnuSmi_ap06} for example), we may use Theorem \ref{thm:
  nodal magnetic}, part (\ref{enu: Hessian and nodal surplus}), to conclude that for $  \kv=\fr{k_{n}\lv}$
\begin{equation}
  \M(-\hess_{\av}\theta_{m}(\kv))=\M(\hess k_{n}(\av)|_{\av=0})=\sigma(n).
\end{equation}
This proves Theorem \ref{thm: nodal statistics}, part (\ref{enu: Sigma and sigma}).

It was shown\footnote{The definitions of $ \bsigma $ in
\cite{AloBanBer_cmp18} and in the present paper are slightly different
but agree when restricted
to $\mgen$} in \cite{AloBanBer_cmp18} that $\bsigma$ is locally
constant on $\mgen$.  The argument relies on the fact that the Morse
index can change only when an eigenvalue passes through 0, i.e.\ the
dimension of the kernel changes.  The latter cannot happen on $\mgen$
because the dimension of the kernel is fixed by \eqref{eq: dim ker
  Hess} (and the two types of Hessians are non-zero multiples of each
other by \eqref{eq: Hess k Hess theta}).

Since every connected component of $\mgen$ has boundary of positive
co-dimension in $\mreg$
\cite{AloBanBer_cmp18,Alon,Alon_PhDThesis}, the
indicator functions of $\mgen\cap\bsigma^{-1}(j)$ and $\mgen$ are
Riemann integrable.

Now assume that $\lv$ is rationally independent
and evaluate $P(\sigma=s)$ by definition \eqref{eq:prob_defn},
\begin{align*}
  P(\sigma=s)
  &:= \lim_{N\to\infty}
    \frac{\big|\sigma^{-1}(s)\cap[N]\big|}{\big|\G\cap[N]\big|} \\
  &= \lim_{N\to\infty}
    \frac{\Big|\set{n\le N}{\fr{k_{n}\lv}\in \mgen\cap\bsigma^{-1}(s)}\Big|}
    {\Big|\set{n\le N}{\fr{k_{n}\lv}\in \mgen}\Big|} \\
  &= \frac{\mu_{\lv}\left(\mgen\cap\bsigma^{-1}(s)\right)}{\mu_{\lv}\left(\mgen\right)},
\end{align*}
where $ [N]:=\{1,2,\ldots,N\} $ and we used Lemma~\ref{lem: mreg} and part~(\ref{enu: Sigma and
  sigma})- of the present theorem to get to the second line. In the last equality we apply Theorem~\ref{thm: equidistirbution} to the (Riemann integrable)
characteristic functions of $\mgen\cap\bsigma^{-1}(s)$ and $\mgen$.

If $\Gamma$ has no loops, then $\mgen$ is a set of full measure and
\eqref{eq: Pj without loops} follows.
\end{proof}

\subsection{Hessian in terms of the unitary evolution matrices}\label{subsec:hessian}

Previously, we described the nodal surplus using the signature of
$\hess_{\av}\theta_{n}(\kv)$, the Hessian of
$\theta_{n}(\kv;\av)$ with respect to $\av$ at the point
$(\kv;0)$. We will now derive $\hess_{\av}\theta_{n}(\kv)$ explicitly
in terms of $U_{\kv}$ and its eigenpair $e^{i\theta_{n}}$ and
$\ba_{n}$. Note that $U_{\kv;\av}$ (see Definition \ref{def:The-unitary-evolution-alpha}) can be written as
\begin{equation}\label{eq: Zj alpha generators}
    U_{\kv;\av}=e^{i\sum_{j=1}^{E}{\alpha_{j}Z_{j}}}U_{\kv}.
\end{equation}
Recall that the diagonal matrices $\{Z_{j}\}_{j=1}^E$ were defined in \eqref{eq:Zmatrix_def} as follows,
\[(Z_{j})_{i,i'}=\begin{cases}
	1 & i=i'=j\\
	-1 & i=i'=\hat{j}\\
	0 & \mbox{elsewhere}
\end{cases}, \qquad i,i'=1,2,\ldots,2E.\]
We now define $ g(U) $:
\begin{defn}\label{def: Cayley}
  Let $ U $ be an $ n $-dimensional unitary matrix.  Consider the
  orthogonal decomposition $ \C^{n}=V_{1}\oplus V_{2} $ with
  $ V_{2}=\ker{(\id-U)} $ and $ V_{1}=V_{2}^{\perp} $. Let $ \id_{1} $
  and $ \id_{2} $ be the identity matrices on $ V_{1} $ and $ V_{2} $
  correspondingly, so that
  \[U =  
    \begin{pmatrix}
      U_{1,1} & U_{1,2}\\
      U_{2,1} & U_{2,2}
    \end{pmatrix} =\begin{pmatrix}
      U_{1,1} & 0\\
      0 & \id_{2}
    \end{pmatrix}
    \quad\mbox{and  }\det(\id_{1}-U_{1,1})\ne0.\\ \]  
  Then, the Moore--Penrose inverse of $ (\id-U) $ is defined by
  \begin{equation*}
    \left(\id-U\right)^{+} :=
    \begin{pmatrix}
      (\id_{1}-U_{1,1})^{-1} & 0\\
      0 & 0
    \end{pmatrix}.    
  \end{equation*}
  In particular, the range of $\left(\id-U\right)^{+}$ is orthogonal
  to $V_{2}=\ker{(\id-U)}$.  Finally, we define the self-adjoint
  matrix
  \begin{align*}	
    g(U) :=  i\left(\id+U\right)\left(\id-U\right)^{+} & =\begin{pmatrix}
      i\left(\id_{1}+U_{1,1}\right)(\id_{1}-U_{1,1})^{-1} & 0\\
      0 & 0
    \end{pmatrix}.
  \end{align*}
\end{defn}
\begin{rem}
  Given a self-adjoint operator $ A $, its Cayley
  transform (see, e.g., \cite[Sec.~13.1]{Schmudgen}) is the unitary operator
  $ C(A):=(A-i\id)(A+i\id)^{-1}$. The map $ g $ agrees with the inverse
  of the Cayley transform whenever $ (\id-U) $ is invertible,
  \cite[Thm.~13.5 and Cor.~13.7]{Schmudgen}. Namely,
  for any self-adjoint $ A $,
	\[g(C(A))=A.\] 
\end{rem}

\begin{prop}
  \label{prop: Hessian in terms of U}
  Let $e^{i\theta}$ be a simple eigenvalue of $U_{\kv}$ with
  normalized eigenvector $\ba$. Then the Hessian of $\theta$
  with respect to $\av$ at $\av=0$ is given by
  \begin{equation}\label{eq: Hessian in terms of...}
  	    \left(\hess_{\vec{\alpha}}\theta\right)_{j,j'}
  	= \ba^{*}Z_j g\left(e^{-i\theta}U_{\kv}\right)Z_{j'}\ba,\qquad j,j'=1,2,\ldots,E.  	
  \end{equation}
  In particular, if $\kv\in\mreg$ and $\ba$ is the normalized
  eigenvector of the eigenvalue 1,
  \begin{equation}
    \label{eq:sigma_Hessian_explicit}
    \bsigma(\kv) = \M(-\bH_\kv), \qquad
    \text{where }
    (\bH_\kv)_{j,j'} := \ba^{*}Z_j g\left(U_{\kv}\right)Z_{j'}\ba.
  \end{equation}
\end{prop}
The proof of the matrix equality \eqref{eq: Hessian in terms of...} has two steps:
\begin{enumerate}
	\item We show in Lemma \ref{lem: H is symmetric} that the matrix on the right-hand-side of \eqref{eq: Hessian in terms of...} is real symmetric.
	\item Having real symmetric matrices on both sides of \eqref{eq: Hessian in terms of...}, we prove the equality by showing that the associated real quadratic forms agree:
	\[\frac{d^{2}}{dt^{2}}\theta(t \vv)|_{t=0}:=\sum_{j,j'=1}^{E}\left(\hess_{\vec{\alpha}}\theta\right)_{j,j'} v_{j} v_{j'}=\sum_{j,j'=1}^{E}\left(ba^{*}Z_j g\left(e^{-i\theta}U_{\kv}\right)Z_{j'}\ba\right) v_{j} v_{j'},\]
	for every $ \vv\in\R^{E}=T_{0}\T^{E} $. The notation $ \theta(t\vv) $ stands for $ \theta $ at $ \kv $ with $ \av=t\vv $ for small enough $ t $. The notation $ T_{0}\T^{E} $ stands for the tangent space to $ \T^{E} $ at $ \av=0 $ which is the space on which the Hessian acts.  
\end{enumerate} 
\begin{lem}\label{lem: H is symmetric}
Let $e^{i\theta}$ be a simple eigenvalue of $U_{\kv}$ with normalized eigenvector $\ba$. Then the $ E\times E $ matrix $H$ defined by
\begin{equation}
H_{j,j'}:=\ba^{*}Z_{j}g\left(e^{-i\theta}U_{\kv}\right)Z_{j'}\ba,\qquad j,j'=1,2,\ldots,E,
\end{equation}
is real symmetric.
\end{lem}
\begin{proof} It will be convenient to write $H_{j,j'}$ as a trace:
\begin{equation}\label{eq: H as trace}
    H_{j,j'}=\tr\left(\ba\ba^{*} Z_{j} g(e^{-i\theta} U_{\kv}) Z_{j'}\right).
\end{equation}
The matrix $g\left(e^{-i\theta}U_{\kv}\right)$ is self-adjoint. Therefore, $H$ is self-adjoint by
\begin{equation}\label{eq: Hij=conj Hji}
    \overline{(H_{j,j'})}=\tr\left(\ba\ba^{*} Z_{j} g(e^{-i\theta} U_{\kv}) Z_{j'}\right)^{*}=\tr\left(Z_{j'} g(e^{-i\theta} U_{\kv}) Z_{j}\ba\ba^{*} \right)= H_{j',j}.
\end{equation}

It is less immediate to show that $H$ is real.  Intuitively, it is due to time-reversal symmetry of the problem. The operator implementing the time-reversal is $T_{\kv}:= Je^{-i\hat{\kappa}}$, with the ``edge-reversing" matrix $J$ given by
\begin{equation}\label{eq: J}
    \forall{i\le E}\qquad J_{i,i+E}=J_{i+E,i}=1, \text{  and zero elsewhere}.
\end{equation}
We now derive several commutation relations between $T_{\kv}$ and the matrices entering \eqref{eq: H as trace}.
From the definition, $J=J^{-1}$ and $J$ commutes with $e^{i\hat{\kappa}}$, so $T_{
\kv}^{-1}=Je^{i\hat{\kappa}}$.
The following relations hold,
\begin{equation}
    J S J = S^{T},
    \label{eq: time reversal}
    \qquad
    T_{\kv}{U}_{\kv} T_{\kv}^{-1} = J S J e^{i\hat{\kappa}} = S^{T} e^{i\hat{\kappa}}= \overline{U}_{\kv}^{-1}.
\end{equation}
The first is obtained from  the definitions of $S$ and $J$, and the second by using the first and the fact that $ U_{\kv}^{T}=\overline{U}_{\kv}^{-1} $. The second relation encapsulates the time-reversal symmetry of a quantum graph on the level of the unitary evolution matrix $U_{\kv}$ \cite{KotSmi_ap99}.
Consider the spectral decomposition of $U_{\kv}$,
\begin{equation}
  \label{eq:spec_decomp_U}
    U_{\kv}= e^{i\theta}\ba\ba^{*} + \sum_{\varphi_j\neq\theta} e^{i\varphi_{j}}P_{j},
\end{equation}
where, by our assumptions, $e^{i\theta}$ is a simple eigenvalue and $P_{j}$ is the orthogonal projection onto the eigenspace of $e^{i\varphi_{j}}$. Substituting \eqref{eq:spec_decomp_U} into the two sides of \eqref{eq: time reversal} we have
\begin{equation}
\overline{U}_{\kv}^{-1}= e^{i\theta}\overline{\ba\ba^{*}} + \sum_{\varphi_j\neq\theta} e^{i\varphi_{j}}\overline{P}_{j},
\end{equation}
and
\begin{equation}
    T_{\kv}{U}_{\kv} T_{\kv}^{-1}=e^{i\theta}T_{\kv}\ba\ba^{*}T_{\kv}^{-1} + \sum_{\varphi_j\neq\theta} e^{i\varphi_{j}}T_{\kv}P_{j}T_{\kv}^{-1}.
\end{equation}
Note that both expressions are spectral decompositions since $\overline{P}_j$ is an orthogonal projector and the conjugation of $P_j$ by a unitary matrix $T_{\kv}$ remains an orthogonal projection. By uniqueness of spectral decomposition we have
\begin{equation}\label{eq: aa and aa bar}
    T_{\kv}\ba\ba^{*} T_{\kv}^{-1}=\overline{\ba\ba^{*}},
\end{equation}
and, for every $j$,
\begin{equation}\label{eq: Pj conj}
    T_{\kv}P_{j}T_{\kv}^{-1}=\overline{P_{j}}.
\end{equation}
The decomposition of $g\left(e^{-i\theta}U_{\kv}\right)$ is given by:
\begin{equation}\label{eq: g(U) explicit}
    g\left(e^{-i\theta}U_{\kv}\right)=\sum_{\varphi_j\neq\theta}i\frac{1+e^{i(\varphi_{j}-\theta)}}{1-e^{i(\varphi_{j}-\theta)}}
    P_{j}=\sum_{\varphi_j\neq\theta}
    \cot\left(\frac{\varphi_{j}-\theta}{2}\right)P_{j},
\end{equation}
where $\ba\ba^{*}$ disappears because $\ba$ is in the kernel of $1-e^{-i\theta}U_{\kv}$ and thus in the kernel of $(1-e^{-i\theta}U_{\kv})^+$, see the definition of $g$ and the Moore-Penrose inverse. Conjugating $g\left(e^{-i\theta}U_{\kv}\right)$ and applying \eqref{eq: Pj conj} gives:
\begin{equation}
    T_{\kv} g(e^{-i\theta} U_{\kv}) T_{\kv}^{-1}
    =\sum_{\varphi_j\neq\theta} \cot\left(\frac{\varphi_{j}-\theta}{2}\right) \overline{P_{j}}
    \label{eq: g(U) and g(U) bar}
    =\overline{g(e^{-i\theta} U_{\kv})}.
\end{equation}

Finally, notice that $Z_{j}$ and $e^{-i\hat{\kappa}}$ commute (both being diagonal) and that $J Z_{j} J = -Z_{j}$ by the definitions of $ J $ and $ Z_{j} $, therefore 
\begin{equation}\label{eq: Zj conjugated is -Zj}
    T_{\kv}^{-1} Z_{j} T_{\kv} = - Z_{j}.
\end{equation}
Consider the expression for $H_{j,j'}$ in \eqref{eq: H as trace} and substitute $Z_{j}$ and $Z_{j'}$ with \eqref{eq: Zj conjugated is -Zj}:
\begin{align*}
    H_{j,j'}
    &=\tr\left(\ba\ba^{*}\ T_{\kv}^{-1} Z_{j} T_{\kv}\ g(e^{-i\theta} U_{\kv})\ T_{\kv}^{-1} Z_{j'} T_{\kv}\right)\\
    &=\tr\left(T_{\kv}\ba\ba^{*}T_{\kv}^{-1} \ Z_{j}\ T_{\kv}g(e^{-i\theta} U_{\kv}) T_{\kv}^{-1}\ Z_{j'} \right)\\
    &=\tr\left(\overline{\ba\ba^{*}}\ Z_{j}\
    \overline{g(e^{-i\theta}U_{\kv})}\ Z_{j'}\right)
    =\overline{H_{j,j'}}.
\end{align*}
This completes the proof.
\end{proof}

\begin{proof}[Proof of Proposition \ref{prop: Hessian in terms of U}]
  We fix $\kv\in\T^{E}$ and consider $U=U_{\kv}$ with simple
  eigenvalue $ \theta $ and normalized eigenvector $ \ba $, where we
  omit the $\kv$-dependence for brevity. The fact that $ U_{\kv,\av} $ is analytic in $ \av $ with simple eigenvalue $ \theta $ at $ \av=0 $, means that both $ \theta $ and its normalized eigenvector $ \ba $ can be analytically extended to $ \theta(\av) $ and $ \ba(\av) $ around $ \av=0 $ (see \cite{Rellich1969perturbation} and \cite{Kato_book} for example). We denote the Hessian of $ \theta(\av) $ at $ \av=0 $ by
  $\hess_{\av}\theta=\hess_{\av}\theta(\kv)$.

We want to show that $\hess_{\av}\theta=H$, where the real symmetric $H$ is defined in
Lemma~\ref{lem: H is symmetric}.  Since $\hess_{\av}\theta$ is also real symmetric (by definition), the two matrices are equal if and only if their real quadratic forms agree.
\begin{equation}\label{ref: quadratic forms}
  \frac{d^{2}}{dt^{2}}\theta(t \vv)|_{t=0}
  :=\sum_{j,j'=1}^{E}\left(\hess_{\vec{\alpha}}\theta\right)_{j,j'}
  v_{j} v_{j'}
  =\sum_{j,j'=1}^{E}\left(\ba^{*}Z_j g\left(e^{-i\theta}U_{\kv}\right)Z_{j'}\ba\right) v_{j} v_{j'},
\end{equation}
for every $ \vv\in\R^{E} $. The notation $ \theta(t\vv) $ uses $ \av=t\vv $ for small enough $ t $. To do so, fix $\vv\in\R^{E}$ and denote $  $
\begin{equation}
    Z_{\vv}:=\sum_{j=1}^{E}{v_{j}Z_{j}},
\end{equation}
so that $ \theta(t\vv) $ is the eigenphase of $U_{\kv;t\vv} = e^{i t Z_{\vv}}U$, and \eqref{ref: quadratic forms} which we want to show can now be written as
	\begin{equation}\label{eq: Hess=H reformulated}
		\frac{d^{2}}{dt^{2}}\theta(t \vv)|_{t=0}=\ba^{*}Z_{\vv} g\left(e^{-i\theta}U_{\kv}\right)Z_{\vv}\ba.
	\end{equation}

To prove \eqref{eq: Hess=H reformulated}, denote the real symmetric $ 2E\times 2E $ matrix 
\[ M(t):=t Z_\vv-\theta(t\vv)\id, \]
where $ \id $ is the identity matrix. We abbreviate $\ba(t):=\ba(t\vv)$. This is the
normalized continuation of $\ba$ defined (uniquely up to a phase) by
\begin{equation}\label{eq: a_t equation M}
    \forall t\in\R
    \qquad
    \left(\id-e^{i M(t)} U \right)\ba(t)=0,
\end{equation}
and the normalization condition $\|\ba(t)\|\equiv 1$. We only consider $t$ around $t=0$, so that $\theta(t\vv)$ remains simple and both $\theta(t\vv)$ and $\ba(t)$ are smooth in $ t $. We will denote $ t $ derivatives by tags. We can fix the phase of $ \ba(t) $ such that the derivative $ \ba'(0) $ is orthogonal to $ \ba=\ba(0) $ at $ t=0 $. To see that, choose an arbitrary smooth real function $ \varphi(t) $ with $ \varphi(0)=0 $. Taking derivative at $ t=0 $ of the normalization condition $ \|e^{i\varphi(t)}\ba(t)\|^{2}\equiv1 $ gives
\[\Re[\ba^{*}\cdot(i\varphi'(0)\ba+\ba'(0))]=\Re[\ba^{*}\cdot\ba'(0)]= 0.\] 
We can choose $ \varphi $ to cancel the imaginary part, $ \varphi'(0)=-\Im[\ba^{*}\cdot\ba'(0)] $ and redefine $ \ba(t) $ as $ e^{i\varphi(t)}\ba(t) $. It will be a solution to \eqref{eq: a_t equation M}  with $ \ba(0)=\ba $ which also satisfies  
\begin{equation}\label{eq: phase normalization}
	\ba^{*}\cdot\ba'(0)=0.
\end{equation} 
Differentiating \eqref{eq: a_t equation M} and then substituting $e^{i M(t)}U \ba(t)= \ba(t)$, gives
\begin{equation}\label{eq: first derivative}
    \left(\id-e^{i M(t)} U \right) \ba'(t)-iM'(t)\ba(t)=0.
\end{equation}
Differentiating \eqref{eq: first derivative} and rearranging gives,
\begin{equation}\label{eq: second derivative}
    \left(\id-e^{i M(t)} U \right)\ba''(t) - iM'(t)\left(\id+e^{i M(t)} U\right) \ba'(t)
    -iM''(t)\ba(t)=0.
\end{equation}
Multiplying by $i \ba(t)^{*}$ cancels the first term since $\ba(t)^{*}\left(1-e^{i M(t)} U \right)=0$, so
\begin{equation}\label{eq: second derivative clean}
  \ba(t)^{*}M'(t)\left(1+e^{i M(t)} U\right) \ba'(t)
    +\ba(t)^{*}M''(t)\ba(t)=0.
\end{equation}
To get an expression at $ t=0 $ that depends only on $ \ba $ and not
$\ba'(0)$, we need to solve \eqref{eq: first derivative} for
$\ba'(0)$.  We know a priori that a $\ba(t)$ exists, therefore
$iM'(0)\ba(0) \in \Ran\left(\id-e^{i M(0)} U \right)$.  Applying the
Moore--Penrose inverse to \eqref{eq: first derivative} at $t=0$ we
obtain the solution
\begin{equation}\label{eq: a_t explicit}
  \ba'(0)=\left(\id-e^{i M(0)} U \right)^{+}iM'(0)\ba(0),
\end{equation}
that is orthogonal to the kernel of $\id-e^{i M(t)} U$ by
Definition~\ref{def: Cayley}.  Since the latter kernel is spanned by
$\ba(0)$, this is precisely the solution we seek.

Substituting \eqref{eq: a_t explicit} in \eqref{eq: second derivative
  clean} and rearranging, we get
\begin{align}
  -\ba(0)^{*}M''(0)\ba(0)
  &= \ba(0)^{*}M'(0)i\left(1+e^{i M(0)} U\right)\left(1-e^{i M(0)} U\right)^{+}M'(0)\ba(0)\\
  &= \ba(0)^{*}M'(0) g\left(e^{i M(0)} U\right)M'(0) \ba(0).
    \label{eq: g exp M}
\end{align}
 The value and derivatives of $M(t)$ at $t=0$, using $	\frac{d}{dt}\theta(t \vv)|_{t=0}=0$ (by \eqref{eq: theta dot=0}), are
\begin{equation}\label{eq: initial conditions}
    M(0)=-\theta\id,\qquad M'(0)=Z_\vv, \quad\mbox{and}\quad M''(0)=-\id	\frac{d^{2}}{dt^{2}}\theta(t \vv)|_{t=0}.
\end{equation}
Substituting the values of $ M(0),M'(0) $ and $ M''(0) $ in the relation \eqref{eq: g exp M}, gives
\begin{equation}
\frac{d^{2}}{dt^{2}}\theta(t \vv)|_{t=0}=  \ba^{*}Z_{\vv}g(e^{-i\theta}U)Z_{\vv}\ba,
\end{equation}
establishing \eqref{eq: Hess=H reformulated} and thus completing the proof.
\end{proof}

\subsection{Proof of Theorem \ref{thm: sampling simplified} and its
  generalization to graphs with loops.}
\label{subsec: proof of main theorem}

\begin{proof}[Proof of Theorem \ref{thm: sampling simplified}]
  The graph $\Gamma$ has no loops, by the statement's assumption. Fix
  rationally independent $\lv\in\R^{E}_{+}$. The nodal surplus
  distribution of $\Gamma_{\lv}$ is described in Theorem \ref{thm:
    nodal statistics}, part \eqref{enu: Sigma and pj}, as
  $P(\sigma=s) = \mu_{\lv}(\bsigma^{-1}(s)) =
  \mu_{\lv}(\bsigma^{-1}(s) \cap \mgen)$.  According to
  Proposition~\ref{prop: Hessian in terms of U}, a point
  $\kv \in \mgen$ belongs to the set $\bsigma^{-1}(s)$ if and only if
  the matrix $\bH_\kv$ has $s$ positive eigenvalues. Recall the indicator $ I^+_s(A) $ (introduced in Theorem~\ref{thm: sampling simplified}) which equals $ 1 $ if the matrix $ A $ has exactly $ s $ positive eigenvalues and $ 0 $ otherwise. The function $f(\kv) := I^+_s(\bH_\kv)$ coincides with the indicator
  function for the set $\bsigma^{-1}(s) \cap \mgen$ on $ \mgen $ which is a set of full
  measure in $ \Sigma $ when the graph has no loops. In particular, $I^+_s(\bH_\kv)$ is measurable and
  \begin{equation*}
      \int_{\mg}I^+_s(\bH_\kv)\dd \mu_{\lv}=\mu_{\lv}\left(\bsigma^{-1}(s) \cap \mgen\right)=P(\sigma=s).
  \end{equation*}
  By substituting the function $ h(\kv) $ in Theorem~\ref{thm: BGmeasure as spectral measures} with $I^+_s(\bH_\kv)$ gives
  \begin{equation}
    \label{eq: Pj as integral}
    P(\sigma=s)
    = \int_{\T^{E}}\sum_{n=1}^{2E}
    I^+_s(\bH_{\kv-\theta_n})
    \frac{\ba_{n}^*\boldsymbol{L}\ba_{n}}{\tr(\boldsymbol{L})}
    \frac{\dd \kv}{(2\pi)^E}.
  \end{equation}
  We conclude, due to $U_{\kv-\theta_n} =
  e^{-i\theta_n} U_\kv$, that the matrix $\bH_{\kv-\theta_n}$ is precisely
  the matrix $\bH_n(\kv)$ we introduced in \eqref{eq:Hmatrix_def}.
\end{proof}

In the case of a graph with loops, unavoidably, there are
eigenfunctions supported on loops that appear with non-zero
frequency.  Such eigenfunctions are non-generic and should be excluded
from the statistics.  This is done by excluding a certain subspace,
$\Vas\subset\C^{2E}$, which is a common invariant subspace of all
$\{U_{\kv}\}_{\kv\in\T^{E}}$.

\begin{defn}
Let $\Gamma$ be a graph and let $\EL$ be its set of loops. For every
loop $e\in\EL$ define its anti-symmetric vector $\as_{e}\in\C^{2E}$ to
be equal to 1 on $e$, $-1$ on $\hat{e}$ and zero elsewhere. Consider the orthogonal decomposition $ \C^{2E}=\Vas\oplus\Vs $, with
\begin{align}
  \Vas:= & \Span{\{\as_{e}\}_{e\in\EL}}, \label{eq: Vas}\\
  \Vs:= & \Vas^{\perp}.\label{eq: Vs}
\end{align}
\end{defn}

It is a simple check (using \eqref{eq:Udef} and \eqref{eq: S matrix}) to see that for any $e\in\EL$ and any $\kv\in\T^{E}$,
\begin{equation}\label{eq: Uk as=as}
    U_{\kv}\as_{e}=e^{i\kappa_{e}}\as_{e},
\end{equation}
and therefore we can choose a basis of eigenvectors of $U_\kv$, all of which belong to either $\Vs$ or
$\Vas$.
The generalization of Theorem \ref{thm: sampling simplified} can now be stated:

\begin{thm}
  \label{thm: generalized theorem}
  Let $\Gamma$ be a graph, possibly with loops.  Then, the nodal surplus
  distribution of $\Gamma_{\lv}$, for rationally independent $\lv$ is
  given by
  \begin{equation}
    \label{eq: generalization nodal stat}
    P\left(\sigma=s\right)
    = \frac{1}{\left(2\pi\right)^{E}}\int_{\T^{E}}\sum_{n\colon \ba_n
      \in \Vs} \idxH \frac{\ba_{n}^{*}\boldsymbol{L}\ba_{n}}{2L-\LL}\dd \kv,
\end{equation}
where $\boldsymbol{L}:=\diag\{\lv,\lv\}$, $L$ is the total length of
$\Gamma_{\lv}$, $\LL$ is the total length of its loops and $\idxH$ is
defined in Theorem \ref{thm: sampling simplified}.
\end{thm}

\begin{rem}
Observe that the $2L-\LL$ factor is equal to the partial trace,
\[
  \tr_{\Vs}\left(\boldsymbol{L}\right)
  :=
  \sum_{n\colon \ba_n
    \in \Vs} \ba_{n}^{*}\boldsymbol{L}\ba_{n}
  =\tr\left(\boldsymbol{L}\right)-\sum_{e\in\EL}\frac{\as_{e}^{*}\boldsymbol{L}\as_{e}}{2}
  =2L-\LL.
\]
The normalization in the second equality is due to $\|\as_{e}\|^{2}=2$. If we define the matrices $f(U_{\kv}):=\sum_{n=1}^{2E}\idxH\ba_{n}\ba_{n}^{*}$, then
\eqref{eq: nodal stat in the simplified thm} states that
\[
P\left(\sigma=s\right)=\int_{\T^{E}}\frac{\tr(f(U_{\kv})\boldsymbol{L})}{\tr(\boldsymbol{L})}\frac{\dd \kv}{\left(2\pi\right)^{E}},
\]
when $\Gamma$ has no loops. If $\Gamma$ has loops, \eqref{eq: generalization nodal stat} has a similar form, only restricted to $\Vs$,
\[
P\left(\sigma=s\right)=\int_{\T^{E}}\frac{\tr_{\Vs}(f(U_{\kv})\boldsymbol{L})}{\tr_{\Vs}(\boldsymbol{L})}\frac{\dd \kv}{\left(2\pi\right)^{E}}.
\]

\end{rem}

The following lemma characterizes the part of the secular manifold
which corresponds to eigenfunctions supported exclusively on a loop.

\begin{lem}\cite{Alon,Alon_PhDThesis,AloBanBer_cmp18,BerLiu_jmaa17}
  \label{lem: mreg and Sigma decomposition}
  Let $\Gamma$ be a graph and let $\EL$ be its set of loops. Let
  \begin{align*}
    \mL
    &:=\set{\kv\in\T^{E}}{\exists e\in\EL
      \text{ s.t. }\kappa_{e}=0} \\
    &=\set{\kv\in\mg}{\exists \as \in \Vas
      \text{ s.t. }U_\kv \as = \as}.
  \end{align*}
  Then $\mgen$ is a subset of $\mreg\setminus\mL$ of full measure.
\end{lem}
The fact that $\mgen$ is a subset of $\mreg\setminus\mL$ of full measure follows from the results of \cite{BerLiu_jmaa17} and appears in the proof of \cite[Proposition A.1]{AloBanBer_cmp18} for example. The interpretation of $\mL$ as $\set{\kv\in\mg}{\exists \as \in \Vas   \text{ s.t. }U_\kv \as = \as}$ is straightforward from \eqref{eq: Uk as=as}.

\begin{proof}[Proof of Theorem \ref{thm: generalized theorem}.]
  Fix a rationally independent $\lv\in\R^{E}_{+}$. By Theorem \ref{thm:
    nodal statistics}, the nodal surplus distribution of
  $\Gamma_{\lv}$ is equal to
  \begin{equation}
    \label{eq: Pj in terms of Sigma in the proof}
    P(\sigma=s)
    = \frac{\mu_{\lv}(\mgen\cap\bsigma^{-1}(s))}{\mu_{\lv}\left(\mgen\right)},\qquad s=0,1,2,\ldots,\beta.
  \end{equation}
  Since $\tr\left(\boldsymbol{L}\right)=2L$, Lemma \ref{lem: mreg}
  gives
  $\frac{1}{\mu_{\lv}\left(\mgen\right)}=\frac{\tr\left(\boldsymbol{L}\right)}{2L-\LL}$. The
  set $\mgen\cap\bsigma^{-1}(s)$ is a full measure subset of
  $(\mreg\setminus\mL) \cap \bsigma^{-1}(s)$.  Let $\chi_{\mathrm{sym}}$
  be the indicator function of the subspace $\Vs$ and let $\ba(\kv)$ be the eigenvector of eigenvalue $1$ of $U_\kv$ for $\kv\in\mreg$. Then $\chi_{\mathrm{sym}}(\ba(\kv))$ is the indicator function of $\mreg\setminus\mL$ for $\kv\in\mreg$. Consider
  $f(\kv) := \chi_{\mathrm{sym}}(\ba(\kv)) I^+_s(\bH_\kv)$. This
  function differs from the indicator function of the set
  $\mgen\cap\bsigma^{-1}(s)$ on a set of measure zero.  The rest of the
  proof is identical to the proof of Theorem~\ref{thm: sampling simplified}.
\end{proof}

\section{The polytope of nodal surplus distributions and discrete
  graphs symmetries.}
\label{Edge lengths dependence and discrete graphs symmetries}

In general, the nodal surplus distribution depends on the choice of
lengths $\lv$ (Theorem \ref{thm: TOC} being a notable exception).
Recall the notation $\L(\Gamma)$ for the set of rationally independent
edge lengths of $\Gamma$ and the description of the distribution in
terms of the vector $\vec{P}_{\lv}\in\R_{\ge0}^{\beta+1}$ whose $j$-th
entry is $P(\sigma=j)$.  In Subsection \ref{subsec: We and convex
  hull}, we associated to every discrete graph $\Gamma$ without loops,
a set of vectors $ \{\vec{W}_e\}_{e\in\E} $ in
$ \R_{\ge0}^{\beta+1} $, such that
\begin{equation}\label{eq: Pl in terms of We no loops}
  \vec{P}_{\lv}=\frac{1}{L}\sum_{e\in\E}\ell_{j}\vec{W}_{e}, \qquad
  \text{for all } \lv\in\L(\Gamma),
\end{equation}
and therefore
\begin{equation}\label{eq: convex 1}
  \overline{\set{\vec{P_\lv}}{\lv\in\L(\Gamma)}}
  = \conv\set{\vec{W}_e}{e \in \E} =: \mathcal{P}(\Gamma),
\end{equation} 
where $\conv$ is the convex hull and the closure is taken with respect
to the standard topology on $ \R^{\beta+1} $.  We refer to the set
$\mathcal{P}(\Gamma)$ as the \emph{polytope of nodal surplus
  distributions}. In this section, using the tools constructed in
previous sections, we will provide an expression for the vectors
$\{\vec{W}_e\}_{e\in\E}$ that extends naturally to graphs with loops.
We will also show the effect of discrete graph symmetries on the
polytope $\mathcal{P}(\Gamma)$.  These results will serve as analytic
tools in the proof of Conjecture \ref{conj: Universality} to stowers
and mandarins in the next section.

\begin{prop}
  \label{prop:W_clean}
  Let $\Gamma$ be a graph, possibly with loops.  Define
  \begin{equation}
    \label{eq:ce_def}
    c_e :=
    \begin{cases}
      1 & e \text{ is a loop}\\
      2 & e \text{ is not a loop}
    \end{cases}, 
  \end{equation}
  and recall the measure $\mu_{e}$ from Definition \ref{defn:specialBGmeasures}.\\
  Then the vectors $\vec{W}_e$, $e\in\E$, defined by
  \begin{equation}\label{eq: wej generalized}
    W_{e,j} :=
    \frac{\mu_{e}\left(\mgen\cap\bsigma^{-1}(j)\right)}{\mu_{e}\left(\mgen\right)},
    \qquad j=0,1,\ldots,\beta,
  \end{equation} 

  have the following
  properties.
  \begin{enumerate}
  \item For all $\lv\in \L(\Gamma)$,
    \begin{equation}\label{eq: Pvec convex in We}
      \vec{P_\lv}=\frac{1}{\sum_{e\in\E}\lve c_{e}}\sum_{e\in\E}\lve c_{e} \vec{W}_{e},
    \end{equation}
    and, therefore, \eqref{eq: convex 1} holds.
  \item If $\Gamma$ has no loops, $\vec{W}_e$ defined by \eqref{eq:
      wej generalized} satisfy \eqref{eq: Wej def}.
  \end{enumerate}
\end{prop}

\begin{proof}
  Combining \eqref{eq: Pj with loops} with
  \eqref{eq:total_measure_generic}, and subsequently using
  \eqref{eq: BG mue}, we get
  \begin{align}
    P(\sigma=s) 
    &= \frac{2L}{2L-\LL}\mu_{\lv}\left(\mgen\cap\bsigma^{-1}(s)\right)
      \nonumber\\
    \label{eq:P_expansion}
    &= \frac{2}{2L-\LL}
      \sum_{e\in\E}\ell_{e}\mu_{e}\left(\mgen\cap\bsigma^{-1}(s)\right).
  \end{align}
  Observe that $2L-\LL = \sum_{e\in\E}\ell_{e}c_e$ and also
  $\mu_{e}\left(\mgen\right)=\frac{c_{e}}{2}$ (for instance, using
  \eqref{eq:total_measure_generic}).  Equation \eqref{eq:P_expansion}
  becomes
  \begin{align*}
    P(\sigma=s) &= \frac{1}{2L-\LL}
    \sum_{e\in\E}\ell_{e}c_e
    \frac{\mu_{e}\left(\mgen\cap\bsigma^{-1}(s)\right)}{\mu_{e}\left(\mgen\right)}
    \\
    &= \frac{1}{\sum_{e\in\E}\ell_{e}c_e}
    \sum_{e\in\E}\ell_{e}c_e
    W_{e,s},
  \end{align*}
  establishing \eqref{eq: Pvec convex in We}.
  To show that $\vec{W}_e$ defined by \eqref{eq: wej generalized} satisfy \eqref{eq: Wej def} we mimic the proof in
  section~\ref{subsec: proof of main theorem} noting that
  $I^+_s(\bH_\kv)$ is the indicator function for the set
  $\bsigma^{-1}(s) \cap \mgen$ and using Theorem~\ref{thm: BGmeasure
    as spectral measures} with $\lv = (0,\ldots,1,\ldots, 0)$.
\end{proof}

\begin{rem}
  One can also understand $\vec{W}_e$ as the limit of $\vec{P}_\lv$ as
  $\lv \to (0,\ldots, 1, \ldots, 0)$ while remaining rationally
  independent.
\end{rem}

A useful implication of \eqref{eq: Pvec convex in We} is that every
$\vec{W}_{e}$ satisfies the following symmetry,
\begin{equation}\label{eq: Wej symmetry}
  W_{e,s}=W_{e,\beta-s},\qquad s=0,1,2,\ldots,\beta,
\end{equation}
since every $\vec{P}_{\lv}$ with $\lv\in\L(\Gamma)$ satisfies this
symmetry by \cite[Thm. 2.1]{AloBanBer_cmp18}.

While \eqref{eq: Wej symmetry} holds for every graph, discrete
symmetries of the graph $\Gamma$ further reduce the set of distinct
$\vec{W}_e$ vectors, which bounds the number of vertices of the polytope
$\mathcal{P}(\Gamma)$.  The definition below is adjusted to include
graphs with loops and multiple edges.

\begin{defn}
  \label{def: symetry group}
  Given a (discrete) graph $\Gamma=\Gamma(\E,\V)$ its symmetry group
  $\sym$ is the group of all graph automorphisms, i.e.\ invertible mappings
  $g$ that map vertices to vertices, edges to edges and preserve
  incidence: $e\in\E_{v}\iff g(e)\in\E_{g(v)}$. The orbit of an edge $e$ is
  denoted by
  \[
    [e]:=\set{e'\in\E}{\exists g\in\sym \text{ s.t. } g(e)=e'},
  \]
  and the set of such orbits is denoted by $\E/\sym$.
\end{defn}

\begin{thm}\label{thm: symmetry theorem}
Let $\Gamma$ be a graph and let $\sym$ be its symmetry group. Then, for any edge $e\in\E$ and graph automorphism $g\in\sym$:
\begin{equation}\label{eq: We equal Wge}
\vec{W}_{e}=\vec{W}_{g(e)}.
\end{equation}
In particular, if $ \Gamma $ is edge transitive\footnote{That is, any
  pair of edges $e,~e' $ has a graph automorphism $ g\in\sym $ such
  that $ g(e)=e' $.}, then $ \vec{P_\lv} $ is independent of
$ \lv\in\L(\Gamma) $ since the polytope
$\mathcal{P}(\Gamma)$ is a single point.
\end{thm}

\begin{rem}
  The theorem implies there are at most $|\E/\sym|$ distinct
  $\vec{W}_{e}$'s.  The actual number of distinct $\vec{W}_{e}$'s can
  be much lower. For example, a graph with disjoint cycles and no
  particular symmetries has the same $\vec{P}_{\lv}$ for all
  $ \lv\in\L(\Gamma) $.  In this case the polytope
  $\mathcal{P}(\Gamma)$ is a single point.
\end{rem}

\begin{proof}
  Fix arbitrary $g\in \sym$ and $\lv\in\R_{+}^{E}$ and let $g.\lv$
  denote the permuted length vector, i.e.
  \[
    (g.\lv)_{e}:=\l_{g^{-1}(e)}.
  \]
Consider the metric graph $\Gamma_{g.\lv}$, namely $\Gamma$ with lengths $g.\lv$. Given a function $f$
on $\Gamma_{\lv}$, define the function $g.f$ on $\Gamma_{g.\lv}$ by its restrictions:
\begin{equation}
\left(g.f\right)|_{e}=f|_{{g^{-1}(e)}}.
\end{equation}

Clearly, $f$ is a $\Gamma_{\lv}$ eigenfunction of eigenvalue $k$ if and only if $g.f$ is a $\Gamma_{g.\lv}$ eigenfunction of the same eigenvalue $k$. In particular $\Gamma_{\lv}$ and $\Gamma_{g.\lv}$ share the same spectrum including multiplicity. Moreover, it is easy to see that $f$ is generic if and only if $g.f$ is, and if they are generic then they share the same nodal count. Therefore, $\Gamma_{\lv}$ and $\Gamma_{g.\lv}$ share the same nodal surplus sequence, and in particular the same nodal surplus distribution, $\vec{P}_{\lv}=\vec{P}_{g.\lv}$. That is,
\begin{align}\label{eq: Wvec and g}
0=\vec{P}_{\lv}-\vec{P}_{g.\lv} & =\frac{1}{L}\left(\sum_{e\in\E}\lve \vec{W}_{e}-\sum_{e'\in\E}\left(g.\lv\right)_{e'} \vec{W}_{e'}\right)\\
& =\frac{1}{L}\sum_{e\in\E}\lve\left( \vec{W}_{e}-\vec{W}_{g(e)}\right),\label{eq: We-Wge}
\end{align}
where moving to the second line we use $\left(g.\lv\right)_{e'} \vec{W}_{e'}=\l_{g^{-1}(e')} \vec{W}_{e'}=\l_{e''} \vec{W}_{g(e'')}$ and summing over all $e''$.
\end{proof}

\begin{rem}
  We showed in the proof that given a graph $\Gamma$, for any $\lv\in\R_{+}^{E}$ and $g\in\sym$, $\Gamma_{\lv}$ and $\Gamma_{g.\lv}$ share the same
  spectrum (including multiplicity) and the same index set $\G$. It follows that $\mg,~\mreg$ and $\mgen$ are each invariant
  under the action of $\sym$ on $\T^E$ by permutations.
\end{rem}
The entries of $\vec{W}_{e}$ are defined using the restriction of
$\mu_{e}$ to $\mgen$ (see \eqref{eq: wej generalized}). The next lemma
shows that both $\mu_{e}$ and its restriction to $\mgen$ are pull-back
measures\footnote{By pull-back of the uniform measure we mean that
  their push-forward is the uniform (Lebesgue) measure.} of the
uniform (Lebesgue) measure on $\T^{E-1}$ by some projection
$\boldsymbol{\pi}_{e}:\mg\rightarrow\T^{E-1}$. This lemma plays a role
in the proof of Theorem \ref{thm: stowers and mandarins} in
section~\ref{graph families}. 

\begin{lem}\label{lem: mj and pj} Given $e\in\E$, consider $\mu_{e}$ (see Definition \ref{def: mue}) and its restriction to $\mgen$. Define the canonical projection $\boldsymbol{\pi}_{e}:\mg\rightarrow\T^{E-1}$ by omitting the $e$-th coordinate. Then, for any Borel set $A\subset\T^{E-1}$,
\begin{equation}\label{eq: mj and pj}
\frac{\mu_{e}(\mgen\cap\boldsymbol{\pi}_{e}^{-1}(A))}{\mu_{e}(\mgen)}=\mu_{e}(\boldsymbol{\pi}_{e}^{-1}(A))=\frac{\vol\left(A\right)}{(2\pi)^{E-1}}.
\end{equation}
\end{lem}
Lemma \ref{lem: mj and pj} is visually demonstrated in Figure \ref{fig: projection measures}.

\begin{figure}[h!]
  \centering

  \includegraphics[width=1\textwidth]{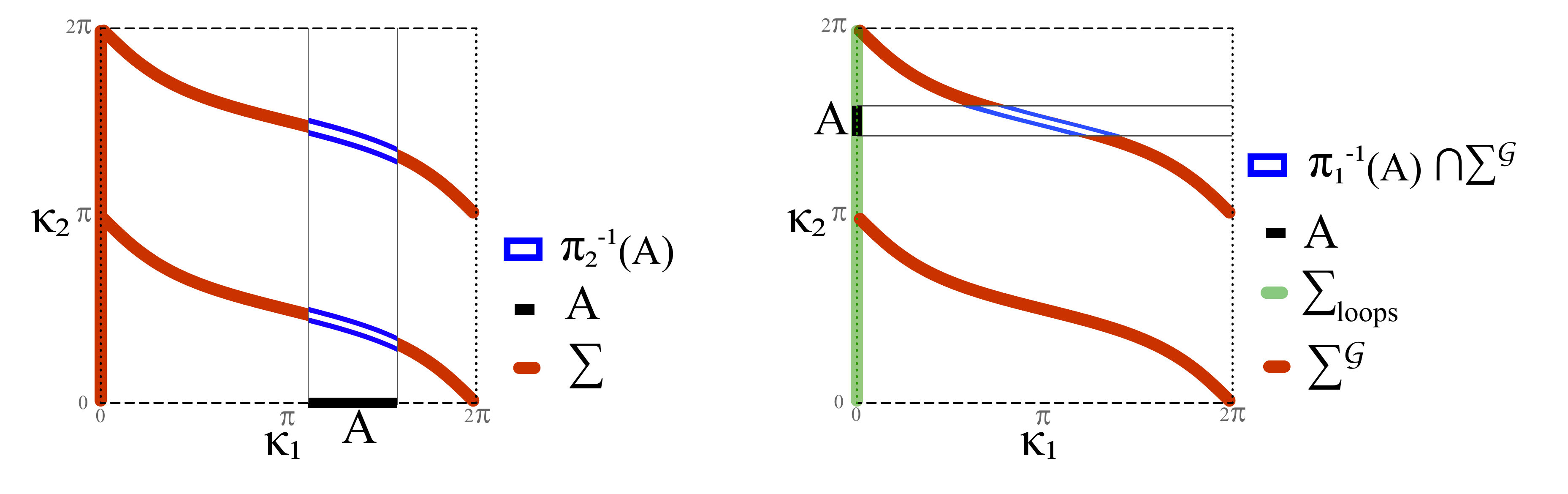}

\caption{(left) A set $A\subset\T$ and its preimage $\boldsymbol{\pi}_{2}^{-1}(A)$. It is visible that $\left|\boldsymbol{\pi}_{2}^{-1}(x)\right|=2$ for every $x\in\T$, except $x=0$. (right) A set $A\subset\T$ and its $\mgen$ preimage $\boldsymbol{\pi}_{1}^{-1}(A)\cap\mgen$. It is visible that $\left|\boldsymbol{\pi}_{1}^{-1}(x)\cap\mgen\right|=1$ for every $x\in\T$, except $x=0,\pi$.}

\label{fig: projection measures}
\end{figure}

\begin{proof}[Proof of Lemma \ref{lem: mj and pj}.]
Order the edges and compare \eqref{eq: BG mue} with \eqref{eq: BG differetial forms}. The density of $\mu_{e}$, when restricted to $\mreg$, is
\begin{equation}\label{eq: dmue}
    \dd\mu_{e}=\frac{\pi}{(2\pi)^{E}}(-1)^{e-1}\,d\kappa_{1}\wedge d\kappa_{2}\ldots \wedge\widehat{d\kappa_e}\ldots \wedge d\kappa_{E}.
\end{equation}
Given a Borel set $A\subset\T^{E-1}$, integrating \eqref{eq: dmue} over $\boldsymbol{\pi}_{e}^{-1}\left(A\right)$ results in
\begin{equation}\label{eq: mue(A) first eq}
    \mu_{e}\left(\boldsymbol{\pi}_{e}^{-1}\left(A\right)\right)=\frac{\pi}{(2\pi)^{E}}\int_{A}|\boldsymbol{\pi}_{e}^{-1}(\xv)|\dd\xv.
\end{equation}
The integral on the right-hand-side is well defined since $|\boldsymbol{\pi}_{e}^{-1}(\xv)|=2$ for (Lebesgue) almost every $\xv\in\T^{E-1}$ (see the proof of Proposition 3.1 in \cite{CdV_ahp15}). Substituting $|\boldsymbol{\pi}_{e}^{-1}(\xv)|=2$ in \eqref{eq: mue(A) first eq} proves the needed equality:
\begin{equation*}
\mu_{e}\left(\boldsymbol{\pi}_{e}^{-1}\left(A\right)\right)=\frac{\vol\left(A\right)}{(2\pi)^{E-1}}.
\end{equation*}

To show that the above holds also for the restriction of $\mu_{e}$ to $\mgen$, let us recall the definition of the set $\mL$ (see Lemma \ref{lem: mreg and Sigma decomposition}),
\begin{equation*}
    \mL:=\set{\kv\in\T^{E}}{\exists e\in\EL \text{ s.t. }\kappa_{e}=0}.
\end{equation*}
Notice that for almost every $\xv\in\T^{E-1}$,
\begin{equation*}
|\mL\cap\boldsymbol{\pi}_{e}^{-1}(\xv)|=\begin{cases} 1 & e\in\EL\\
0 & e\notin\EL\end{cases}.
\end{equation*}
We therefore have a constant $c_{e}$ such that $|\boldsymbol{\pi}_{e}^{-1}(\xv)\setminus\mL|=c_{e}$ for almost every $\xv\in\T^{E-1}$. Integrating \eqref{eq: dmue} over $\boldsymbol{\pi}_{e}^{-1}\left(A\right)\setminus\mL$ results in
\begin{equation}\label{eq: mue(A) second eq}
    \mu_{e}\left(\boldsymbol{\pi}_{e}^{-1}\left(A\right)\setminus\mL\right)=\frac{\pi c_{e}}{(2\pi)^{E}}\vol\left(A\right).
\end{equation}
Lemma \ref{lem: mreg and Sigma decomposition} states that $\mgen$ is equal to $\mg\setminus\mL$ up to a set of measure zero, so
\begin{equation}
\frac{\mu_{e}\left(\boldsymbol{\pi}_{e}^{-1}(A)\cap\mgen\right)}{\mu_{e}(\mgen)}=\frac{\mu_{e}\left(\boldsymbol{\pi}_{e}^{-1}(A)\setminus\mL\right)}{\mu_{e}(\boldsymbol{\pi}_{e}^{-1}(\T^{E-1})\setminus\mL)}=\frac{\vol\left(A\right)}{(2\pi)^{E-1}}.
\end{equation}
\end{proof}

\section{Mandarins and stowers - proof of Theorem \ref{thm: stowers and mandarins}}\label{graph families}\label{subsec: proof for stowers and mandarins}

In this section we prove Theorem \ref{thm: stowers and mandarins} by
approximating the nodal surplus distributions of stowers and mandarins
in terms of suitable binomial distributions, uniformly in $\beta$.
This approximation result is as follows.

\begin{prop}
  \label{prop: stowers and mandarins} Let $\Gamma$ be either a mandarin or a stower graph with $\beta>1$. Define $N_{\beta}:=\beta-1$ if $\Gamma$ is a stower, and $N_{\beta}:=\beta$ if $\Gamma$ is a mandarin. Let $X_{\beta}$ be a binomial random variable  $X_{\beta}\sim\Bin\left(N_{\beta},p=\frac{1}{2}\right)$. Then, for any rationally independent $\lv$, the nodal surplus distribution of $\Gamma_{\lv}$ satisfies
\begin{equation}\label{eq: P and X binom}
    \forall{t\in\R}\qquad P(X_{\beta}\le t-3)\le P(\sigma \le t)\le P(X_{\beta}\le t+3).
\end{equation}
\end{prop}

Assuming Proposition~\ref{prop: stowers and mandarins} to be true, we
prove Theorem \ref{thm: stowers and mandarins}.

\begin{proof}[Proof of Theorem \ref{thm: stowers and mandarins}]
 Let $\Gamma_{\lv}$ be either a stower or a mandarin with rationally independent $\lv$ and first Betti number $\beta>1$, and let $\sigma$ be its nodal surplus random variable. We will consider  the $\beta\to\infty$ limit without adding $\beta$ superscript to $\Gamma_{\lv}$ and $\sigma$, to ease the reading.

 Let $X_{\beta}$ and $N_{\beta}$ be as in Proposition \ref{prop: stowers and mandarins} and observe that $\var\left(X_{\beta}\right)=N_{\beta}/4$ which grows like $\beta/4$. Denote the normalized random variables,
 \begin{equation*}
     \tilde{X_{\beta}}:=\frac{X_{\beta}-\frac{N_{\beta}}{2}}{\sqrt{\var\left(X_{\beta}\right)}},\quad\mbox{and}\qquad \tilde{\sigma}:=\frac{\sigma-\frac{\beta}{2}}{\sqrt{\var\left(X_{\beta}\right)}}.
 \end{equation*}
Notice that $\tilde{\sigma}$ is normalized with the variance of $X_{\beta}$. Let $\varepsilon_{\beta}:=\frac{4}{\sqrt{\var\left(X_{\beta}\right)}}$ and note that $|N_{\beta}-\beta|\le 1$. We can manipulate \eqref{eq: P and X binom} to show that for any $t\in\R$,
\begin{equation}\label{eq: Xbeta and sigma}
  P(\tilde{X_{\beta}}\le t-\varepsilon_{\beta})\le P(\tilde{\sigma}\le t)\le P(\tilde{X_{\beta}}\le t+\varepsilon_{\beta}).
\end{equation}
It is known that $\tilde{X_{\beta}}$ converges in distribution to
$N(0,1)$ (the standard normal random variable). Now consider
\eqref{eq: Xbeta and sigma} for a fixed $t\in\R$ and let
$\beta\rightarrow\infty$. Since  $\varepsilon_{\beta}\rightarrow 0$
and $N(0,1)$ has continuous distribution, we establish that
$\tilde{\sigma}$ also converges in distribution to $N(0,1)$. In particular,
\[
1=\lim_{\beta\rightarrow\infty}\var\left(\tilde{\sigma}\right)=\lim_{\beta\rightarrow\infty}\frac{\var\left(\sigma\right)}{\var\left(X_{\beta}\right)},
\]
which proves that $\var\left(\sigma\right)$ grows like $\beta/4$ and that the properly normalized random variable, $\frac{\sigma-\frac{\beta}{2}}{\sqrt{\var\left(\sigma\right)}}=\sqrt{\frac{\var\left(\sigma\right)}{\var\left(X_{\beta}\right)}}\tilde{\sigma}$, converges in distribution to $N(0,1)$ as needed.
\end{proof}

We now return to the proof of Proposition \ref{prop: stowers and mandarins}.
A crucial part of the proof is that for mandarins and stowers, the function $\bsigma$ can be expressed explicitly, as seen in Lemmas A.5 and A.9 in \cite{AloBan19}. We restate the needed results  from Lemmas A.5 and A.9 as follows:
\begin{lem}\cite{AloBan19}\label{lem: explicit nodal stat}
Let $\Gamma$ be a mandarin or a stower. There is a negligible set $B$ with $\dim{(B)}\le E-2$, such that
\begin{enumerate}
    \item If $\Gamma$ is a stower, then for any $\kv\in\mgen\setminus{B}$
    \begin{equation}\label{eq: stower sigma}
    \bsigma(\kv)=\Big|\set{e\in\EL }{\pi<\kappa_{e}<2\pi}\Big|.
\end{equation}
    \item If $\Gamma$ is a mandarin graph, then for any $\kv\in\mgen\setminus{B}$, either
\begin{equation}\label{eq: mandain sigma plus}
    \bsigma(\kv)=\Big|\set{e\in\E}{\pi<\kappa_{e}<2\pi}\Big|-C(\kv),
\end{equation}
or,
\begin{equation}\label{eq: mandain sigma minus}
    E-\bsigma(\kv)=\Big|\set{e\in\E}{\pi<\kappa_{e}<2\pi}\Big|-C(\kv),
\end{equation}
with some bounded error function $ C:\mgen\to\{-1,0,1\} $.
\end{enumerate}
\end{lem}
\begin{rem}
The set $\mg^{\text{gen}}$ defined in \cite{AloBan19} is a subset of
$\mgen$, reflecting a more restrictive definition of generic
eigenfunctions.  The set $B$ in the above lemma is the difference $\mgen\setminus\mg^{\text{gen}}$, which has $\dim\left(\mgen\setminus\mg^{\text{gen}}\right)\le E-2$ as seen in \cite{Alon,Alon_PhDThesis}.
\end{rem}
\begin{figure}
  \centering
    \subfloat{
  \includegraphics[width=0.55\textwidth]{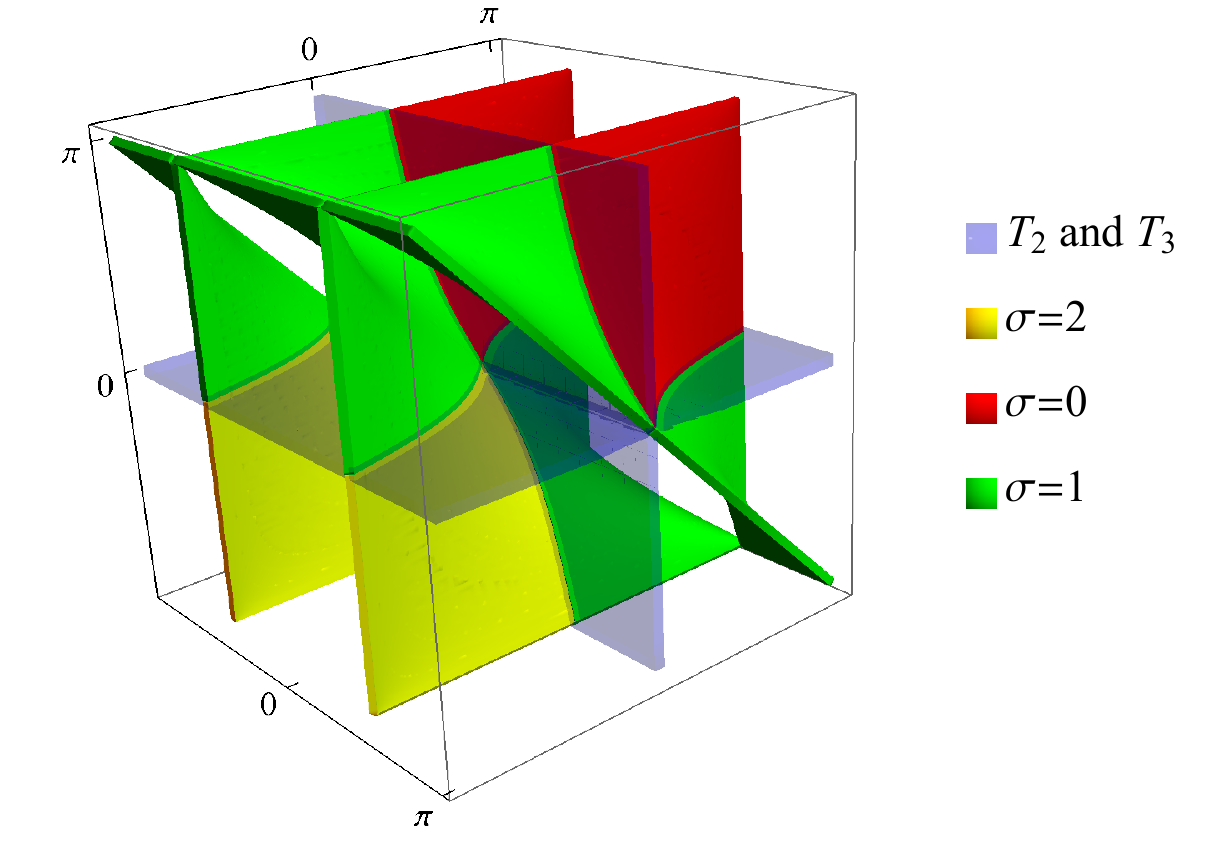}
}
\subfloat{
  \includegraphics[width=0.35\textwidth]{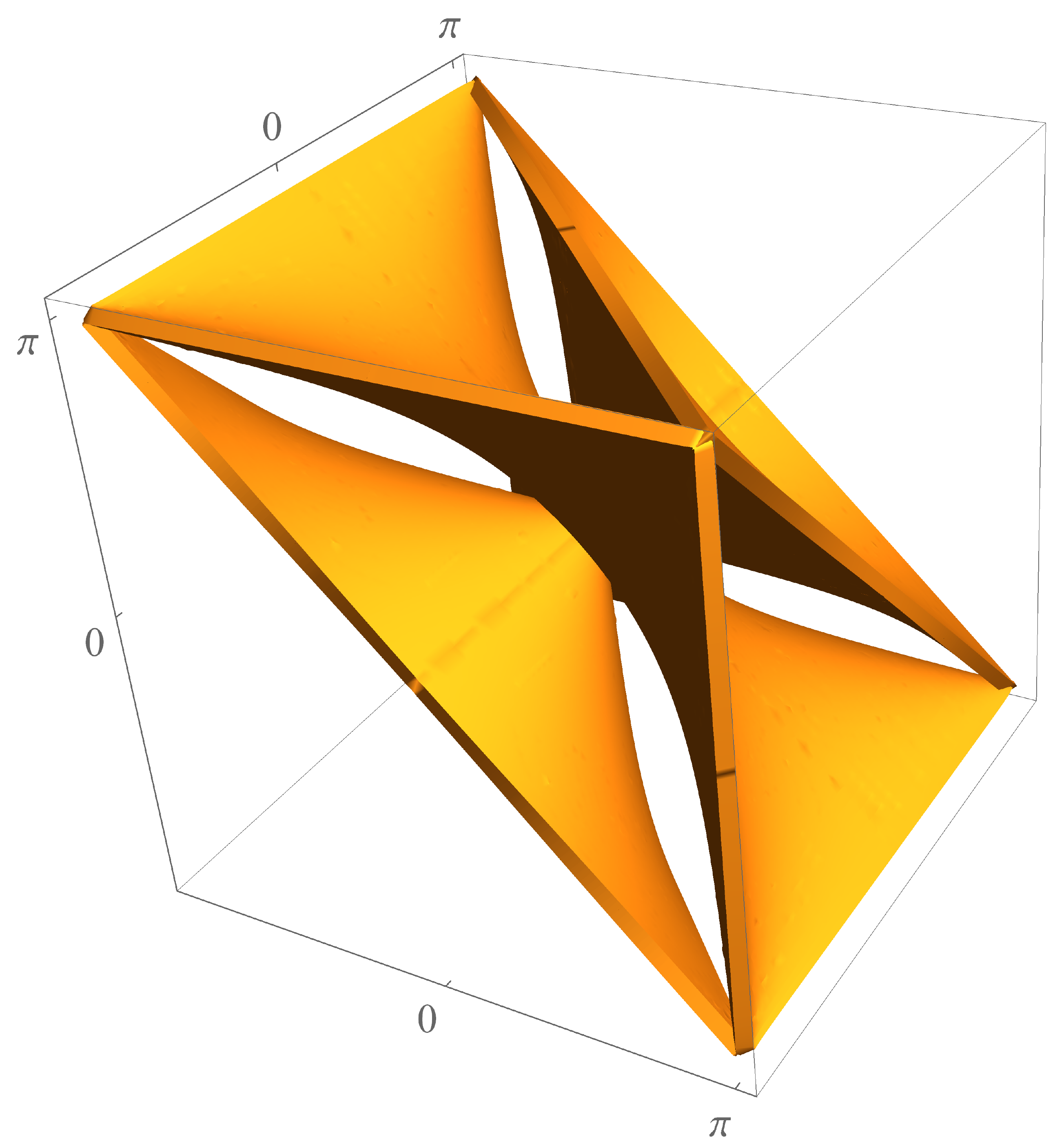}
}
\caption{On the left, the secular manifold $\mg$ of a stower with one tail $ e_{1}$ and two
  loops $ e_{2} $ and $ e_{3} $, presented in $(-\pi,\pi)^{3}$. The loops related part is the union of two sub-tori, $ \mL=T_{2}\cup T_{3} $, with $ T_{j}:=\set{\kv\in\T^{3}}{\kappa_{j}=0} $, presented in low opacity. The rest is colored according to the values of $\bsigma$. On the right, the secular manifold $\mg$ of a mandarin with three edges. In this case $\bsigma\equiv 1$. }

\label{fig:secular manifolds}
\end{figure}
Examples of $\mg$ and the function $\bsigma$ are shown in figure \ref{fig:secular manifolds}. The next step towards proving Proposition \ref{prop: stowers and mandarins} is introducing an auxiliary function $\bx$ which approximates $\bsigma$.
\begin{defn}\label{def: xi} Given a graph $ \Gamma $, we choose $ e_{1}\in\E $ and $ \widetilde{\E}\subset\E $ as follows. \begin{enumerate}
	\item If $ \Gamma $ is a stower (which is not a star) let $ e_{1} $ be a loop, and let $\widetilde{\E}:=\EL\setminus e_{1}$.
	\item If $ \Gamma $ is a mandarin, let $ e_{1} $ be any edge and let $\widetilde{\E}:=\E\setminus e_{1}$.
\end{enumerate} 	
	In both cases, $|\widetilde{\E}|= N_{\beta} $. Define the function $\bx:\mg\to\{0,1\ldots,N_{\beta}\}$ by
	\begin{equation}\label{eq: xi}
		\bx(\kv):=\Big|\set{e\in\widetilde{\E}}{\pi<\kappa_{e}<2\pi}\Big|.
	\end{equation}
For every edge $e$ we define a random variable $\xi_{e}$, taking values in $\{0,1,\ldots,N_{\beta}\}$, with probability
\begin{equation}\label{eq: xie def}
	P(\xi_{e}=j)=\frac{\mu_{e}\left(\mgen\cap\bx^{-1}\left(j\right)\right)}{\mu_{e}\left(\mgen\right)}.
\end{equation}
\end{defn}
\begin{lem}\label{lem: xie binomial}
If $e\notin\widetilde{\E}$ then $\xi_{e}\sim \mathop{Bin}\left(N_{\beta},p=\frac{1}{2}\right)$, namely
\begin{equation}
    P(\xi_{e}=j)=\binom{N_{\beta}}{j}2^{-N_{\beta}}, \qquad j=0,1,2,\ldots,N_{\beta}.
\end{equation}
\end{lem}
\begin{proof}
Consider the canonical projection $\boldsymbol{\pi}_{e}$ (see Lemma \ref{lem: mj and pj}) such that the entries of $\xv=\boldsymbol{\pi}_{e}^{-1}(\kv)$ are $x_{e'}=\kappa_{e'}$ for $e'\in\E\setminus{e}$. Since $e\notin\tilde{\E}$, the level sets of $\bx$ are given by
\begin{equation*}
\bx^{-1}(j)= \set{\kv\in\mg}{\left|\set{e'\in\tilde{\E}}{\pi<\kappa_{e'}<2\pi}\right|=j}  =\boldsymbol{\pi}_{e}^{-1}\left(A_{j}\right),
\end{equation*}
with
\begin{equation*}
A_{j}:= \set{\xv\in\T^{E-1}}{\left|\set{e'\in\tilde{\E}}{\pi<\xv_{e'}<2\pi}\right|=j}.
\end{equation*}
Applying Lemma \ref{lem: mj and pj} gives
\begin{align*}
    P(\xi_{e}=j)= & \frac{\mu_{e}\left(\mgen\cap\bx^{-1}\left(j\right)\right)}{\mu_{e}\left(\mgen\right)}\\
    = & \frac{\mu_{e}\left(\mgen\cap\boldsymbol{\pi}_{e}^{-1}\left(A_{j}\right)\right)}{\mu_{e}\left(\mgen\right)}\\
    = & \frac{\vol\left(A_{j}\right)}{(2\pi)^{E-1}}.
\end{align*}
It is a simple observation that
\begin{equation*}
    \frac{\vol\left(A_{j}\right)}{(2\pi)^{E-1}}=\binom{|\tilde{\E}|}{j}2^{-|\tilde{\E}|}=\binom{N_{\beta}}{j}2^{-N_{\beta}},
\end{equation*}
which proves the lemma.
\end{proof}
The proof of the proposition follows.

\begin{proof}[Proof of Proposition \ref{prop: stowers and mandarins}] Recall the random variables $\omega_{e}$, taking values in $\{0,1,\ldots,\beta\}$, with probability (see \eqref{eq: wej generalized}),
\begin{equation}\label{eq: P omega}
    P(\omega_{e}=j):=W_{e,j}=\frac{\mu_{e}\left(\mgen\cap\bsigma^{-1}(j)\right)}{\mu_{e}\left(\mgen\right)}.
\end{equation}
By \eqref{eq: Pvec convex in We}, the nodal surplus distribution of $\Gamma_{\lv}$, for $\lv$ rationally independent, satisfies
\begin{equation}
\forall t\in\R, \qquad   P(\sigma\le t)=\frac{1}{\sum_{e\in\E}\lve c_{e}}\sum_{e\in\E}\lve c_{e} P(\omega_{e}\le t),
\end{equation}
and so to prove the proposition, we will show that for every edge $e\in\E$,
\begin{equation}\label{eq: We and Xb ineq}
    \forall t\in\R,\qquad P(X_{\beta}<t-3)\le P(\omega_{e}<t)\le P(X_{\beta}<t+3).
\end{equation}
In fact, if $\sym$ is the symmetry group of $\Gamma$, then according to Theorem \ref{thm: symmetry theorem} it is enough to prove \eqref{eq: We and Xb ineq} for one representative edge per equivalence class in $\E/\sym$.
Recalling the choice of $e_1$ in Definition \ref{def: xi}, we have
the following.
\begin{enumerate}
    \item\label{enu: ETS mandarin or flower} If $\Gamma$ is a mandarin or a flower, then $\E/\sym=\{[e_{1}]\}$ and so it is enough to show that \eqref{eq: We and Xb ineq} holds for $\omega_{e_1}$.
    \item\label{enu: ETS stower not flower}If $\Gamma$ is a stower which is not a flower or a star, then $\E/\sym=\{[e_{1}],[e_{2}]\}$ where $e_{2}\notin\EL$, and so it is enough to prove \eqref{eq: We and Xb ineq} for $\omega_{e_1}$ and $\omega_{e_2}$.
\end{enumerate}

We proceed by proving the proposition for stowers. Assume that $\Gamma$ is a stower (possibly a flower). Recall that $\widetilde{\E}:=\EL\setminus e_{1}$ and compare \eqref{eq: xi} with \eqref{eq: stower sigma}, to get
\begin{equation}\label{eq: xi and sigma ineq stower}
    \left|\bsigma(\kv) - \bx(\kv)\right|\leq 1,
\end{equation} for every $\kv\in\mgen\setminus{B}$, for some $B$ of measure zero. If we neglect $B$, this gives

\begin{equation*}
    \mgen\cap\bx^{-1}\left((-\infty,t-1]\right) ~\subset~\mgen\cap\bsigma^{-1}\left((-\infty,t]\right) ~ \subset~\mgen\cap\bx^{-1}\left((-\infty,t+1]\right).
\end{equation*}

Given any $e\in\E$, compare $\omega_{e}$ and $\xi_{e}$ (see \eqref{eq: P omega} and \eqref{eq: xie def}) to conclude that
\begin{equation}
  \forall{t\in\R}\qquad P(\xi_{e}\le t-1)\le P(\omega_{e}\le t)\le P(\xi_{e}\le t+1).
\end{equation}
In particular, if $e\notin\widetilde{\E}$, then $\xi_{e}$ has the same probability distribution as $X_{\beta}$ by Lemma \ref{lem: xie binomial},
\begin{equation}
    \forall t\in\R,\qquad P(X_{\beta}\le t-1)\le P(\omega_{e}\le t)\le P(X_{\beta}\le t+1).
\end{equation}
This proves Proposition \ref{prop: stowers and mandarins} for stowers, as $e_{1}\notin\widetilde{\E}$ and if $e_{2}\notin\EL$ then also $e_{2}\notin\widetilde{\E}$.

We proceed with proving Proposition \ref{prop: stowers and mandarins} for mandarins. Let $\Gamma$ be a mandarin graph, choose $e_{1}$ to be any edge and define $ \widetilde{\E} $ and $ \xi $ correspondingly, as in Definition \ref{def: xi}. As before, Lemma \ref{lem: xie binomial} ensures that $\xi_{e_{1}}$ is binomial like $X_{\beta}$, and so it is left to prove
\begin{equation}\label{eq: We1 xie1 mandarin}
\forall{t\in\R}\qquad  P(\xi_{e_{1}}\le t-3)\le P(\omega_{e_{1}}\le t)\le P(\xi_{e_{1}}\le t+3).
\end{equation}
Notice that both $\xi_{e_{1}}$ and $\omega_{e_{1}}$ are symmetric around $\frac{\beta}{2}$: $\xi_{e_{1}}$ since it is binomial with $p=\frac{1}{2}$ and $N_{\beta}=\beta$, and $\omega_{e_{1}}$ due to \eqref{eq: Wej symmetry}. Under such symmetry, \eqref{eq: We1 xie1 mandarin} is equivalent to
\begin{equation}\label{eq: We1 xie1 abs}
  P\left(\left|\xi_{e_{1}}-\frac{\beta}{2}\right|\le t-3\right)\le P\left(\left|\omega_{e_{1}}-\frac{\beta}{2}\right|\le t\right)\le P\left(\left|\xi_{e_{1}}-\frac{\beta}{2}\right|\le t+3\right),
\end{equation}
for all $t\ge 0$. Let us now prove \eqref{eq: We1 xie1 abs}.

According to Lemma \ref{lem: explicit nodal stat}, given any $\kv\in\mgen\setminus B$, one of the following is true:
\begin{enumerate}
    \item Equation \eqref{eq: mandain sigma plus} holds. Comparing \eqref{eq: mandain sigma plus} and \eqref{eq: xi}, using $|C|\le 1$ and $\widetilde{\E}:=\E\setminus{e_{1}}$ gives $\left|\bx(\kv)-\bsigma(\kv)\right|\le 2$. We may write it as
    \begin{equation}\label{eq: xi and sigma ineq1 mandarin}
    \left|\left(\bx(\kv)-\frac{\beta}{2}\right)-\left(\bsigma(\kv)-\frac{\beta}{2}\right)\right|\le 2,
\end{equation}
    \item Equation \eqref{eq: mandain sigma minus} holds, and a similar argument gives $\left|\bx(\kv)-(E-\bsigma(\kv))\right|\le 2$. Recall that for mandarins $E=\beta+1$, so $\left|\bx(\kv)+(\bsigma(\kv)-\beta)\right|\le 3$. We may write it as
    \begin{equation}\label{eq: xi and sigma ineq2 mandarin}
\left|\left(\bx(\kv)-\frac{\beta}{2}\right)+\left(\bsigma(\kv)-\frac{\beta}{2}\right)\right|\le 3.
\end{equation}
\end{enumerate}
In both cases\footnote{Using $|a-b|<c~\Rightarrow~ -c<|a|-|b|<c$ (the ``dark side" of the  triangle inequality).}
\begin{equation}
    -3\le\left|\bx(\kv)-\frac{\beta}{2}\right|-\left|\bsigma(\kv)-\frac{\beta}{2}\right|\le 3.
\end{equation}
As this holds for all $\kv\in\mgen\setminus{B}$, we conclude that \eqref{eq: We1 xie1 abs} is true using the same arguments as before. This finishes the proof of the proposition.
\end{proof}

\subsection*{Acknowledgments}

We thank the referees for valuable comments. We thank Ronen Eldan for interesting
discussions.
The authors were supported by the Binational Science Foundation Grant (Grant No. 2016281). RB and LA were also supported by ISF (Grant No. 844/19). LA was also supported
by the Ambrose Monell Foundation and the Institute for Advanced Study. GB was also supported by the National Science Foundation (DMS-1815075).

\appendix
\section{Continuous eigenvalues of Unitary matrices}\label{sec: appendix numbering eigenvalues}
In sections \ref{sec:secular} and \ref{sec: nodal surplus distribution}, we consider an analytic family of unitary matrices and their eigenvalues. In this appendix we deal with the question of whether the eigenvalues of a continuous family of unitary matrices can be written as continuous functions. The general setting is as follows: Let $ \textbf{U}(N) $ be the unitary group of $ N\times N $ matrices and a consider a continuous family of matrices $ U_{x}\in\textbf{U}(N),~~x\in M $, for some (finite dimensional) manifold $ M $ with or without boundary. We say that the family is continuous if the map $ x\mapsto U_{x} $ is a continuous map from $ M $ to $ \textbf{U}(n) $. 

\begin{defn}\label{def: numbering}
	We say that the family $ U_{x} $ for $ x\in M $ has a \emph{continuous counterclockwise ordering} (CC ordering) of eigenvalues if there exist $ N $ continuous functions (eigenvalues)
	\[\lambda_{n}:M\to S^{1}, \qquad n=1,2,\ldots,N\]
	and $ N $ continuous functions (the phase gaps),
	\[a_{n}:M\to [0,2\pi], \qquad n=1,\ldots N,\]
	such that 
	\begin{enumerate}
		\item The eigenvalues of $ U_{x} $ at every $ x\in M $ are $ \{\lambda_{n}(x)\}_{n=1}^{N} $ with multiplicity. 
		\item At every $ x\in M $,
		\begin{align*}
			a_{1}+a_{2}+\ldots+a_{N}& = 2\pi\\
			\lambda_{2}=\lambda_{1}e^{ia_{1}},\quad\lambda_{3}=\lambda_{2}e^{ia_{2}}, &\ \ldots,\  \lambda_{1}=\lambda_{N}e^{ia_{N}}.
		\end{align*}
	\end{enumerate}
\end{defn}
\begin{rem}\label{rem: uniqunes of ordering}
	Notice that if $ \{(\lambda_{1}, \lambda_{2},\ldots, \lambda_{N}), (a_{1}, a_{2},\ldots, a_{N})\} $ is a CC ordering, then a cyclic permutation such as $ \{(\lambda_{2}, \ldots, \lambda_{N}, \lambda_{1}), (a_{2},\ldots, a_{N}, a_{1})\} $ is also a CC ordering. The inclusion of both eigenvalues and the phase gaps ensures that a CC ordering (if exists) is unique up to such cyclic permutation. Moreover, these cyclic permutations are distinct (as ordered tuples) at every point $ x\in M $, including the two cases where either all eigenvalues are equal, or all phase gaps are equal.  
\end{rem}
	\begin{rem}\label{rem: local ordering}
	Locally, the eigenvalues can always be ordered in such a fashion, but the question is whether the ordering extends globally. The local argument is as follows. Given $ x\in M $ choose $ t\in\R $ such that $ e^{it} $ is not an eigenvalue of $ U_{x} $. The set of eigenvalues depends continuously on the matrix entries\footnote{Given a matrix $ A $ with eigenvalue $ \lambda $ of algebraic multiplicity $ m $. For small enough $\epsilon>0$ there is a $ \delta $ such that any $ A' $ in a $ \delta $ neighborhood of $ A $ has exactly $ m $ eigenvalues (counting with algebraic multiplicity) in an $\epsilon  $ ball around $ \lambda $. To see that, apply Rouché's theorem to characteristic polynomials.}, so there exists a neighborhood $ \Omega_{x}  $ of $ x $ such that $ e^{it} $ is not an eigenvalue of any $ U_{x'} $ for any $ x'\in \Omega_{x}  $. Then, the eigenvalues of $ U_{x'} $ with $ x'\in\Omega_{x} $ can be written as $ \lambda_{n}(x')=e^{i\theta_{n}(x')} $, where $ \theta_{n}(x')\in(t,t+2\pi) $ are ordered increasingly and are continuous in $ x'\in \Omega_{x} $. Now the claim follows by setting $ a_{N}=\theta_{1}+2\pi-\theta_{N} $ and $ a_{n}=\theta_{n+1}-\theta_{n+1} $ for $ n<N $.
\end{rem}
Given $ \gamma $, a closed path in $ M $, let $ [\gamma] $ denote its homotopy equivalence class. For a continuous function $ f:M\to S^{1} $, denote its winding number along $ \gamma $ by $ \w([\gamma],f)\in\Z $. The induced homomorphism $ f_{*}:\pi_{1}(M)\to\Z $ is defined by 
\[ f_{*}([\gamma]):=\w([\gamma],f).\]  
\begin{prop}\label{prop: continuous eigenvalues}
	Given a continuous family of unitary matrices, $U_{x}\in \textbf{U}(N), ~ x\in M $, the function $ x\mapsto \det(U_{x}) $ is continuous from $ M $ to $ S^{1} $. There is a CC ordering of the eigenvalues if and only if
	\begin{equation}\label{eq: det U mod N}
		\det(U_{x})_{*}\equiv0~~\mod N.
	\end{equation}  
	In particular,
	\begin{enumerate}
		\item If $ M $ is simply connected, then there exist a continuous counterclockwise ordering of the eigenvalues.
		\item Let $ M=\T^{E},~~N=2E, ~~E>1 $ and $ S\in\textbf{U}(2E) $, as in Definition \ref{def:The-unitary-evolution}. If the family of unitary matrices has the form 
		\[U_{\kv}:=\diag(e^{i\kappa_{1}},e^{i\kappa_{2}},\ldots,e^{i\kappa_{N}},e^{i\kappa_{1}},e^{i\kappa_{2}},\ldots,e^{i\kappa_{N}})S\in \textbf{U}(2E),\quad\mbox{with}\quad \kv\in\T^{E},\]
		then there is no CC ordering of the eigenvalues.  
	\end{enumerate}
\end{prop}
\begin{proof}[Proof of Proposotion \ref{prop: continuous eigenvalues}]
	The proof consists of three parts. Part I, where we show that having a CC ordering leads to \eqref{eq: det U mod N}. Part II, where we assume that \eqref{eq: det U mod N} holds, and construct the needed CC ordering. Part III, where we prove (1) and (2). For the ease of reading, we first prove Part III, assuming Part I and Part II. That is, we need to show that in case (1) \eqref{eq: det U mod N} holds while in case (2) \eqref{eq: det U mod N} fails.
	
	Part III (assuming Part I and Part II) - 
	\begin{enumerate}
		\item If $ M $ is simply connected then $ \pi_{1}(M) $ is trivial and therefore any homomorphism from it is trivial, namely $ \det(U_{x})_{*}\equiv 0 $. 
		\item Consider the family discussed above. Namely, $ N=2E $ and $ M=\T^{E} $ and
		\[U_{\kv}:=\diag(e^{i\kappa_{1}},e^{i\kappa_{2}},\ldots,e^{i\kappa_{N}},e^{i\kappa_{1}},e^{i\kappa_{2}},\ldots,e^{i\kappa_{N}})S.\]
		Consider the closed path $ \gamma $ in $ \T^{E} $, defined by $ \gamma(t)=(t,0,0,\ldots,0) $ for $ t\in \R/2\pi\Z $. Then $ \det(U_{\gamma(t)})=e^{2it}\det(S) $ and so the winding number of $ \det(U_{\kv}) $ along $ \gamma $ is $ 2 $. If $ E>1 $ then $ 2 $ is not a multiple of $ 2E $ and therefore \eqref{eq: det U mod N} fails.
	\end{enumerate}

	Part I - Let $ U_{x} $ for $ x\in M $ be a continuous family of unitary matrices and	assume that there exist a continuous counterclockwise ordering of their eigenvalues. Consider the $ \lambda_{n} $'s and $ a_{n} $'s described in Definition \ref{def: numbering}. Notice that $ \w([\gamma],e^{ia_{n}})=0 $ for any continuous $ a_{n}:M\to[0,2\pi] $ and any $ [\gamma] $. In other words, the induced homomorhpism is trivial, $ (e^{ia_{n}})_{*}\equiv 0 $. Since the winding of a product of functions is the sum of the winding of the functions, then for any $ n\ge 2 $
	\[(\lambda_{n})_{*}=(e^{ia_{n-1}}\lambda_{n-1})_{*}=(\lambda_{n-1})_{*}+(e^{ia_{n-1}})_{*}=(\lambda_{n-1})_{*}=\ldots=(\lambda_{1})_{*}.\]
	Therefore, using $ \det(U_{x})=\prod_{n=1}^{N}\lambda_{n} $, we get
	\[\det(U_{x})_{*}=\sum_{n=1}^{N}(\lambda_{n})_{*}=N(\lambda_{1})_{*}\]
	which proves Part II.

	Part II - Let $ U_{x} $ for $ x\in M $ be a continuous family of unitary matrices and	assume that 
	\[\det(U_{x})_{*}\equiv0~~\mod N .\]
	To construct the CC ordering of the eigenvalues we use three steps. The first step is showing that any path $ \gamma:[0,1]\to M $ admits a unique CC ordering that depends only on the initial ordering of the eigenvalues at $ \gamma(0) $. The second step is to show that the ordering at the final point $ \gamma(1) $ actually depends only on the initial ordering at $ \gamma(0) $ and not on the path $ \gamma $. The third step is then to fix some initial point $ x_{0}\in M $ with initial ordering, and number the eigenvalues at any other point $ x\in M $ by a path\footnote{we may assume that $ M $ is a connected finite dimensional manifold and hence path connected.} going from $ x_{0} $ to $ x $. Then, the previous steps ensure that this is indeed a CC ordering.\\

	Constructing CC ordering along a path -  According to Remark \ref{rem: local ordering}, any $ x\in M $ has a neighborhood $ \Omega_{x} $ for which there is a CC ordering of the eigenvalues. Assume that two such neighborhoods $ \Omega_{x} $ and $ \Omega_{x'} $ intersect and that each neighborhood has its CC ordering. If both CC orderings agree on $ \Omega_{x}\cap\Omega_{x'} $ then by definition we have a CC ordering for the union $ \Omega_{x}\cup\Omega_{x'} $. Otherwise, they can only differ by a cyclic permutation on $ \Omega_{x}\cap\Omega_{x'} $, according to Remark \ref{rem: uniqunes of ordering}, in which case we may cyclically permute the CC ordering at $ \Omega_{x'} $ and get a CC ordering on $ \Omega_{x}\cup\Omega_{x'} $. Given a path $ \gamma:[0,1]\to M $, we can cover it by finitely many such neighborhoods, each with its CC ordering. Denote the neighborhoods along $ \gamma $ by $ \Omega_{j} $ for $ j=1,2,\ldots, J $, ordered increasingly\footnote{For this to make sense we should take the neighborhoods small enough such that each $ \set{t\in[0,1]}{\gamma(t)\in\Omega_{j}} $ is a connected open interval, and non of these intervals is completely contained in another.}. Applying the above procedure to every pair of $ \Omega_{j} $ and $ \Omega_{j+1} $, permuting the CC ordering of $ \Omega_{j+1} $ if needed, we get a CC ordering along $ \gamma $. It is now clear that the ordering of the eigenvalues at the final point $ \gamma(1) $ is uniquely determined by the path $ \gamma $ and the ordering at the initial point $ \gamma(0) $. \\
	
	CC ordering along a closed path  - Consider the case that $ \gamma $ is closed, $ \gamma(1)=\gamma(0) $. A priori, the initial and final ordering may differ by a cyclic permutation. We now show that the assumption of \eqref{eq: det U mod N}, namely $ \det(U_{x})_{*}\equiv0~~\mod N $, implies that the initial ordering is equal to the final ordering. Let 
	\[\lambda_{n}:[0,1]\to S^{1},\qquad a_{n}:[0,1]\to[0,2\pi],\qquad n=1,2,\ldots,N,\]
	be the CC ordering of the eigenvalues of $ U_{\gamma(t)} $ for $ t\in[0,1] $. Assume by contradiction that the final ordering differ from the initial ordering. Then, for some $ 1\le d\le N-1 $, 
\begin{align*}
	&\lambda_{1}(1)  = \lambda_{1+d}(0)\\
	&\qquad\vdots\\
	&\lambda_{N-d}(1)  = \lambda_{N}(0)\\
	&\lambda_{N-d+1}(1)  = \lambda_{1}(0)\\
	&\qquad\vdots\\
	&\lambda_{N}(1)  = \lambda_{d}(0).
\end{align*} 
By the Lifting Theorem, there is a continuous function $ \theta_{1}:[0,1]\to\R $ such that $ \lambda_{1}=e^{i\theta_{1}} $. Define inductively $ \theta_{n+1}:=\theta_{n}+a_{n} $, so that each $ \theta_{n}:[0,1]\to\R $ is continuous and satisfies $ \lambda_{n}=e^{i\theta_{n}} $. Due to the properties of the $ a_{n} $ functions, the $ \theta_{n} $ functions are ordered in a $ 2\pi $ interval for every $ t\in[0,1] $, 
	\[\theta_{1}\le\theta_{2}\le\ldots\le\theta_{N}=\theta_{1}+2\pi-a_{N}\le \theta_{1}+2\pi.\]
	We may deduce that 
	\begin{align*}
		&\theta_{1}(1)  = \theta_{1+d}(0)+2\pi m\\
		&\qquad\vdots\\
		&\theta_{N-d}(1)  = \theta_{N}(0)+2\pi m\\
		&\theta_{N-d+1}(1)  = \theta_{1}(0)+2\pi (m+1)\\
		&\qquad\vdots\\
		&\theta_{N}(1)  = \theta_{d}(0)+2\pi (m+1)
	\end{align*} 
	for some integer $ m\in\Z $. Since $ \det(U_{\gamma(t)})=e^{i\sum\theta_{n}(t)} $, then the winding number of $ \det(U_{x}) $ along $ \gamma $ is equal to
	\begin{align*}
		\w(\det(U_{x}),\gamma)&=  \frac{1}{2\pi}\left(\sum_{n=1}^{N}\theta_{n}(1)-\sum_{n=1}^{N}\theta_{n}(0)\right)\\
		&= mN + d.
	\end{align*}
	Since $ 1\le d\le N-1 $, this is a contradiction to \eqref{eq: det U mod N}. \\
	
	Path independence of the ordering - Consider an initial point $ x_{0}\in M $ with a fixed ordering and let $ x\in M $ be another point. Assume by contradiction that there are two paths $ \gamma $ and $ \gamma' $ from $ x_{0} $ to $ x $ whose CC orderings do not agree at the final point $ x $. Let $ \phi $ be the closed path obtained by concatenation of $ \gamma' $ in reverse with $ \gamma $. Namely, starting from $ \phi(0)=x $, going along $ \gamma' $ backwards to $ x_{0} $ (where both the CC ordering agree) and then along $ \gamma $ back to $ \phi(1)=x $. The CC ordering along $ \varphi $ is then defined by that of $ \gamma' $ (reversing is obtained by re-parameterization $ t\mapsto 1-t $) and that of $ \gamma $, which agree at the concatenation point $ x_{0} $, and leads to the disagreement between the initial ordering at $ x=\phi(0) $ and the final at $ x=\varphi(1) $. Contradiction. Hence, the ordering that $ x $ inherit from $ x_{0} $ is independent of the path between them.\\
	
	Constructing CC ordering on $ M $ - Fix an arbitrary initial point $ x_{0} $ and an initial ordering. Then any $ x\in M $ inherits a unique ordering, $ \{(\lambda_{1}(x),\ldots,\lambda_{N}(x)),(a_{1}(x),\ldots,a_{N}(x))\} $ which satisfies (1) and (2) of Definition \ref{def: numbering}. As this ordering agree with the (unique) CC ordering along any path going from $ x_{0} $ to $ x $, then $ \{(\lambda_{1}(x),\ldots,\lambda_{N}(x)),(a_{1}(x),\ldots,a_{N}(x))\} $ must be continuous along any path in $ M $, and hence continuous in $ M $.
\end{proof}

\section{Proof of Lemma \ref{lem: upper bound We}.}\label{sec: appendix proof of upper bound}
Let us restate the lemma. Recall that given a random variable $X$ we defined $N(X)$ as a normal (Gaussian) random variable with mean $\mathbb{E}(X)$ and variance $\var(X)$ (assuming these exist). Consider the distance between two random variables $X,Y$ to be
\begin{equation}
    d_{KS}(X,Y):=\sup_{t\in\R}\left|P(X\le t)-P(Y\le t)\right|,
\end{equation}
as defined in \eqref{eq: DS dist defined}. 
\begin{rem}
	Notice that $ d_{KS}(X,Y) $ is a distance between the cumulative probability functions and is independent of the probability spaces on which $ X $ and $ Y $ are defined. The same is true for $ \var(X) $ and $ \var(Y) $. For this reason we do not specify at any point what are the probability spaces of these random variables.
\end{rem}
Lemma \ref{lem: upper bound We} can be written as follows. 
\begin{lem}
Fix positive integers $\beta,E\in\N$, and consider a set of random variables $\{\omega_{e}\}_{e=1}^{E}  $, all taking values in $\{0,1,\ldots,\beta\}$ symmetrically, namely
\[P(\omega_{e}=s)=P(\omega_{e}=\beta-s),\qquad s=0,1,\ldots,\beta,\]
 for every $ e=1,2,\ldots,E $. Let $ \sigma $ be another random variable that takes values in $\{0,1,\ldots,\beta\}$ symmetrically, and satisfies
\begin{equation}\label{eq: app convex eq}
    P(\sigma=s)=\frac{1}{L}\sum_{e\in\E}\lve P(\omega_{e}=s),\qquad s=0,1,\ldots,\beta,
\end{equation}
for some $\lv\in\R_{+}^{E}$ with $L=\sum_{e\in\E}\lve$. Then,
\begin{enumerate}
    \item $d_{KS}(\sigma,N(\sigma))$ is bounded by
    \begin{equation}\label{eq: appendix ineq 1}
        d_{KS}(\sigma,N(\sigma))\le \max_{e\in\E} d_{KS}(\omega_{e},N(\omega_{e}))+\varepsilon,
    \end{equation}
where $\varepsilon:=\sqrt{\frac{\max_{e\in\E}\var(\omega_{e})}{\min_{e\in\E}\var(\omega_{e})}}-1$.
    \item The variances satisfy
    \begin{equation}\label{eq: appendix ineq 2}
        \min_{e\in\E}\var\left(\omega_{e}\right)\le\var(\sigma)\le\max_{e\in\E}\var\left(\omega_{e}\right).
    \end{equation}
\end{enumerate}
\end{lem}
\begin{proof}
First, due to the symmetry, every $\omega_{e}$ has mean $\frac{\beta}{2}$ and so does $\sigma$. Together with \eqref{eq: app convex eq}, this leads to $\var{(\sigma)}=\frac{1}{L}\sum_{e\in\E}\lve\var{(\omega_{e})}$ from which \eqref{eq: appendix ineq 2} follows. We proceed with proving \eqref{eq: appendix ineq 1}. Using \eqref{eq: app convex eq}, we may write $d_{KS}(\sigma,N(\sigma))$ as
\begin{align*}
    d_{KS}(\sigma,N(\sigma))= & \sup_{t\in\R}|\frac{1}{L}\sum_{e\in\E}\lve P(\omega_{e}\le t)- P(N(\sigma)\le t)|\\
    \le & \sup_{t\in\R}\frac{1}{L}\sum_{e\in\E}\lve |P(\omega_{e}\le t)- P(N(\sigma)\le t)|\\
    \le & \frac{1}{L}\sum_{e\in\E}\lve d_{KS}(\omega_{e},N(\sigma))\\
    \le & \max_{e\in\E}d_{KS}(\omega_{e},N(\sigma))\\
    \le & \max_{e\in\E}d_{KS}(\omega_{e},N(\omega_{e}))+\max_{e\in\E}d_{KS}(N(\omega_{e}),N(\sigma)).
\end{align*}
We are left with showing that $\max_{e\in\E}d_{KS}(N(\omega_{e}),N(\sigma))\le\varepsilon:=\sqrt{\frac{\max_{e\in\E}\var{\omega_{e}}}{\min_{e\in\E}\var{\omega_{e}}}}-1$. To do so, recall that both $N(\omega_{e})$ and $N(\sigma)$ are normal with the same mean, $\frac{\beta}{2}$, and possibly different variances. Let $s_{1}^{2}$ be the smaller variance and $s_{2}^{2}$ be the larger variance, so that
\begin{align*}
    d_{KS}(N(\omega_{e}),N(\sigma)) = & \sup_{t\in\R}|\int_{-\infty}^{t}\frac{e^{-\frac{1}{2}(\frac{x-\frac{\beta}{2}}{s_{1}})^{2}}}{\sqrt{2\pi}}\frac{\dd x}{s_{1}}- \int_{-\infty}^{t}\frac{e^{-\frac{1}{2}(\frac{x-\frac{\beta}{2}}{s_{2}})^{2}}}{\sqrt{2\pi}}\frac{\dd x}{s_{2}}|\\
    = & \frac{1}{\sqrt{2\pi}}\sup_{t\in\R}|\int_{-\infty}^{\frac{t}{s_{1}}}e^{-\frac{x^{2}}{2}}\dd x-\int_{-\infty}^{\frac{t}{s_{2}}}e^{-\frac{x^{2}}{2}}\dd x|\\
    = & \frac{1}{\sqrt{2\pi}}\sup_{t\ge 0}\int_{\frac{t}{s_{2}}}^{\frac{t}{s_{1}}}e^{-\frac{x^{2}}{2}}\dd x\\
    \le & \frac{1}{\sqrt{2\pi}}\sup_{t\ge 0}e^{-\frac{(\frac{t}{s_{2}})^{2}}{2}}(\frac{t}{s_{1}}-\frac{t}{s_{2}})\\
    = & \frac{e^{-\frac{1}{2}}}{\sqrt{2\pi}}(\frac{s_{2}}{s_{1}}-1)\\
    \le & \frac{s_{2}}{s_{1}}-1.
\end{align*}
As the variance inequality \eqref{eq: appendix ineq 2} gives
\[\min_{e\in\E}\var{(\omega_{e})}\le s_{1}^{2}\le s_{2}^{2}\le\max_{e\in\E}\var{(\omega_{e})},\]
we may deduce that $\frac{s_{2}}{s_{1}}-1\le\sqrt{\frac{\max_{e\in\E}\var{\omega_{e}}}{\min_{e\in\E}\var{\omega_{e}}}}-1=\varepsilon$. This finishes the proof.

\end{proof}
\section{Images from the numerical experiments}\label{sec: appendix experiment}
The experiment described in Subsection \ref{subsec: exp setting} considered the following graphs:
\begin{enumerate}
\item Complete graphs of $n$ vertices for $n\in\{5,\ldots 12\}$.
\item Periodic ladder graphs of $n$ steps for $n\in\{6,10,14,18,22\}$.
\item Periodic square lattices of $n^{2}$ vertices for $n\in\{4,5,6,7\}$.
\item Random $5$-regular graphs of $n$ vertices for $n\in\{12,14,16,18,20\}$.
\item Random Erd\H{o}s-Rényi graphs with $n$ vertices and $p=0.75$, for $n\in\{9,10,11,12\}$.
\end{enumerate}
We emphasize again that the randomly chosen graphs are not an average over random samples of graphs but rather one sample per type (as we discuss in Remark \ref{rem: random graphs}). The figures in Subsection \ref{subsec: exp result} provide the needed evidence for convergence to Gaussian with variance that grows linear in $\beta$. We provide complement figures in which the convergence of the distributions can be seen visually. Each figure shows the convergence for one family of graphs. For every graph in such a family we present the distribution of the $\omega_{e}$ which maximize $d_{KS}\left(\omega_{e},N(\omega_{e})\right)$ (namely the ``worst" candidate). The different graphs in each family are labeled by their first Betti number $\beta$ (for the Erd\H{o}s-Rényi graphs we first sample a graph and then compute its $ \beta $).

Recall that $\omega_{e}$ is supported on $\{0,1,\ldots,\beta\}$ with probability given by the weights $W_{e,s}=P(\omega_{e}=s)$, satisfying the symmetry
\[
W_{e,s}=W_{e,\beta-s},\qquad s=0,1,2,\ldots,\beta,
\]
for all $e\in\E$. To compare the distributions of graphs with different $\beta$ we consider the normalized distributions $\frac{\omega_{e}-\frac{\beta}{2}}{\sqrt{\beta}}$. This is done by plotting $W_{e,s}$ against $\frac{s-\frac{\beta}{2}}{\sqrt{\beta}}$. The normalized distribution is symmetric around $0$ with variance $\frac{\var{\left(\omega_{e}\right)}}{\beta}$. The experiment showed that for all graphs and all edges, $W_{e,s}<10^{-4}$ when $|\frac{s-\frac{\beta}{2}}{\sqrt{\beta}}|> 1.5$. For better visualization we restrict the figures to $|\frac{s-\frac{\beta}{2}}{\sqrt{\beta}}|\le 1.5$. The figures clearly support the conjecture, as one can see how the normalized distributions converge to a Gaussian with variance of order one.
\begin{rem}
Although these are discrete distributions, we present them with a line graph for visualization. This is not to be confused with a probability density. In particular, the area under the curve is not 1 but roughly $\frac{1}{\sqrt{\beta}}$.
To emphasize that these are discrete distributions, while maintaining a visually clear picture, we add markers on the data points of the smallest and largest graphs of every family.
\end{rem}
\begin{rem}
	The complementary figures in this appendix are partitioned into  different families of graphs, for visualization purpose. However, the evidence for the validity of Conjecture \ref{conj: Universality} are the figures in Subsection \ref{subsec: exp result}, in which the convergence to Gaussian is measured quantitatively and is shown to decay uniformly with $ \beta $, regardless of how we partition the sampled graphs into families of growing $ \beta $. 
\end{rem}

\begin{figure}[h!]
  \centering

  \includegraphics[width=0.7\textwidth]{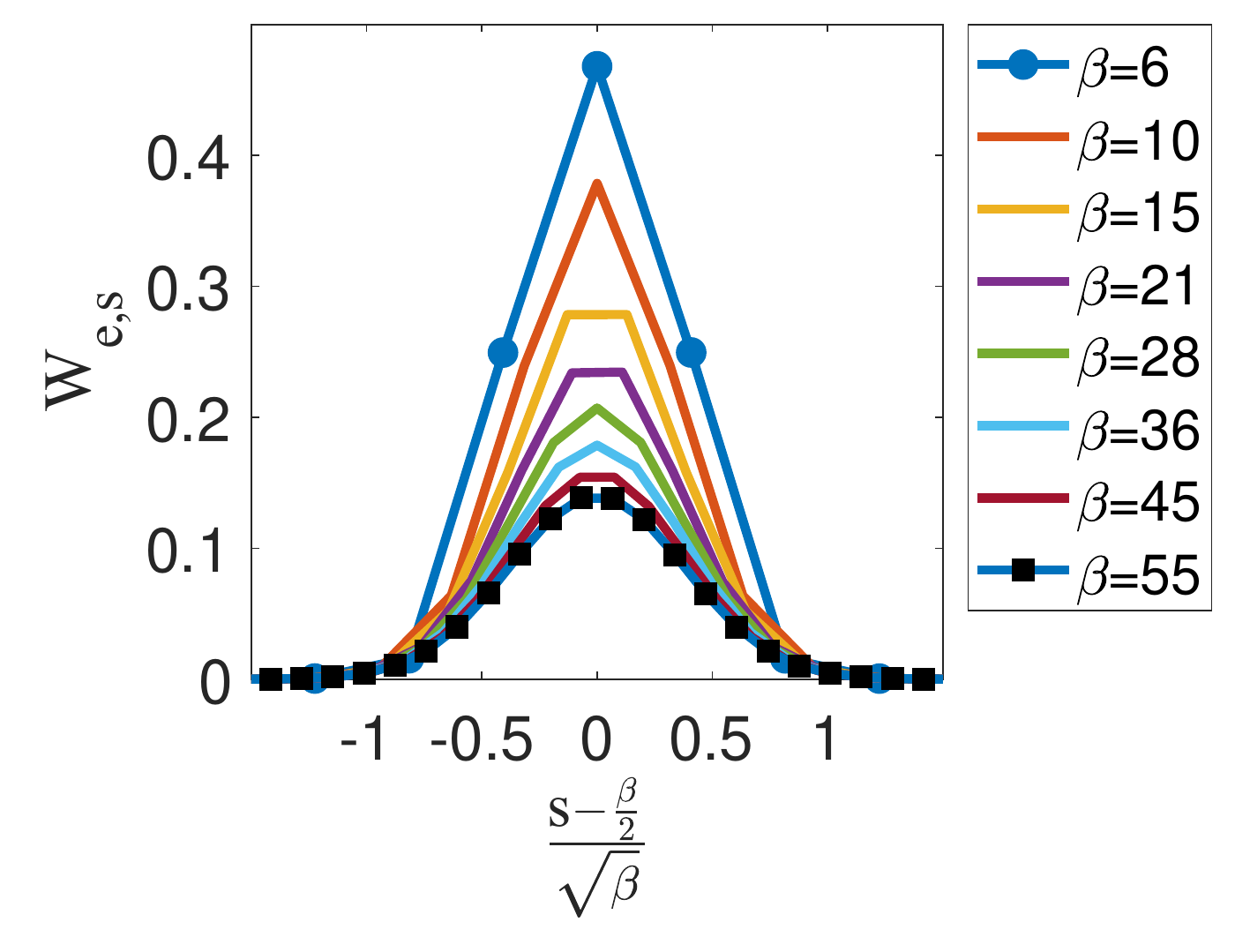}

\caption{Complete graphs of $n$ vertices, with $n$ ranging from $5$ to $12$.}

\label{fig:complete}
\end{figure}

\begin{figure}[h!]
  \centering

  \includegraphics[width=0.7\textwidth]{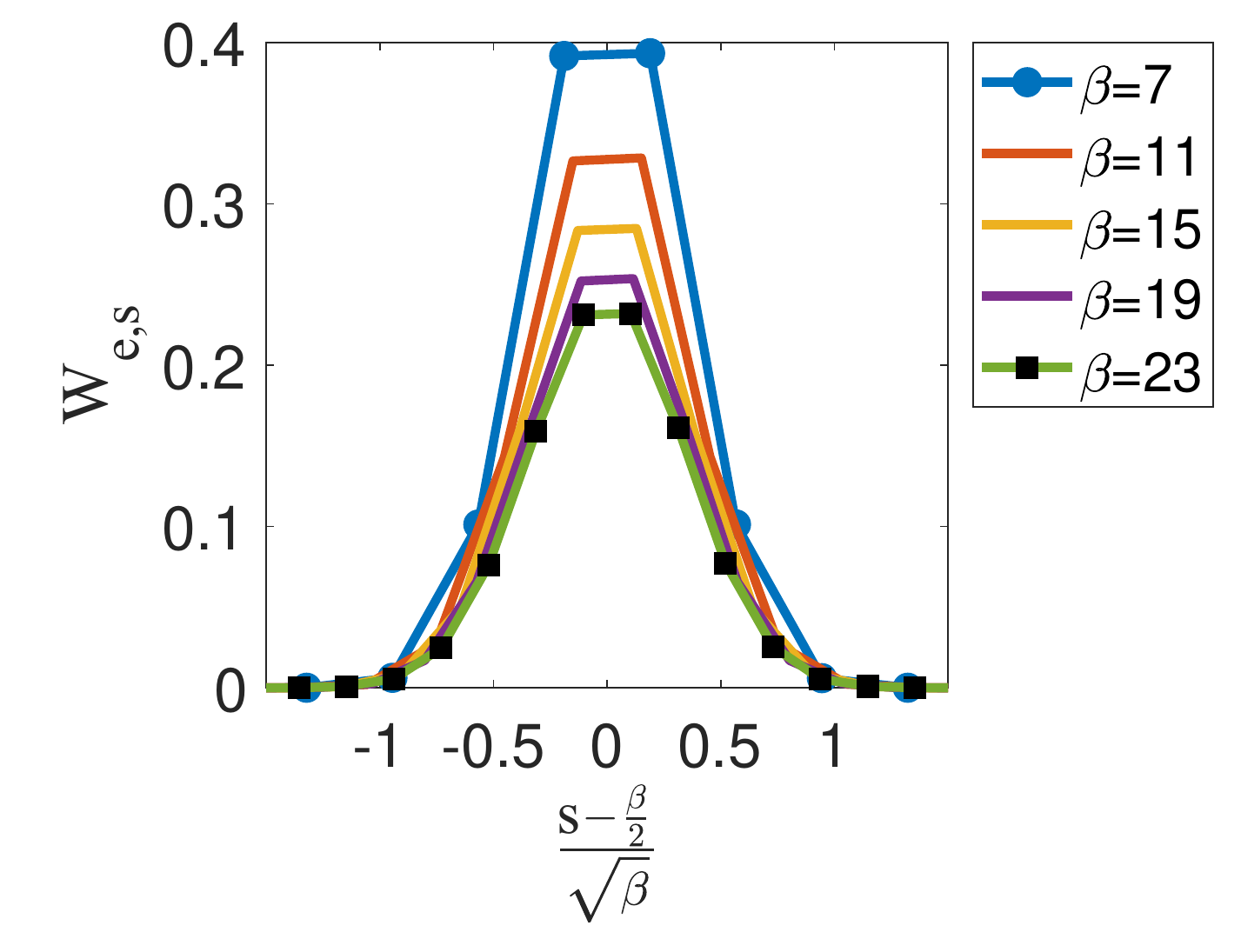}

\caption{Periodic ladder graphs of $n$ steps, with $n\in\{6,10,14,18,22\}$. }

\label{fig:ladder}
\end{figure}


\begin{figure}[h!]
  \centering

  \includegraphics[width=0.7\textwidth]{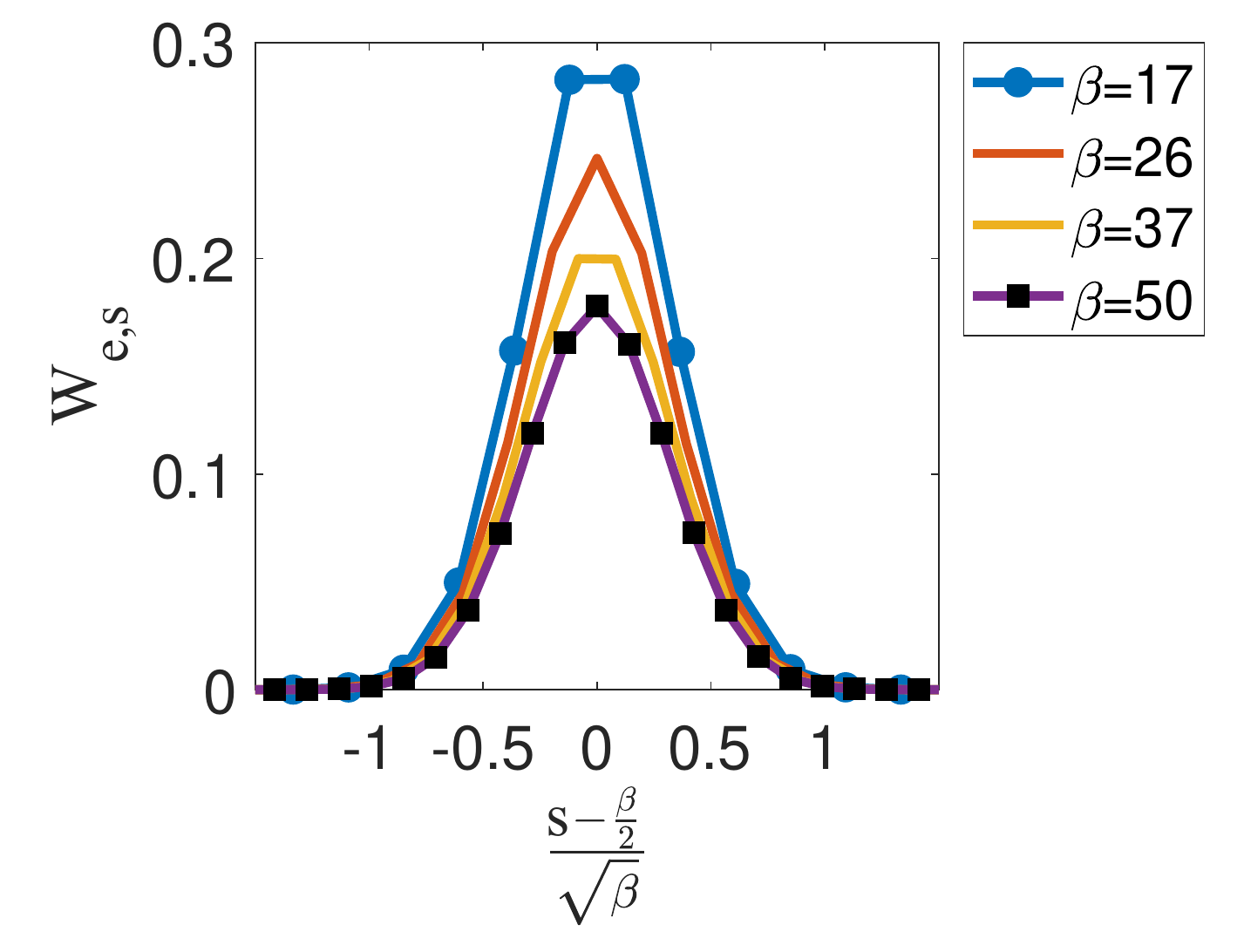}

\caption{Periodic square lattices of $n^{2}$ vertices for $n$ ranging between $4$ to $7$.}

\label{fig:Z2}
\end{figure}


\begin{figure}[h!]
  \centering

  \includegraphics[width=0.7\textwidth]{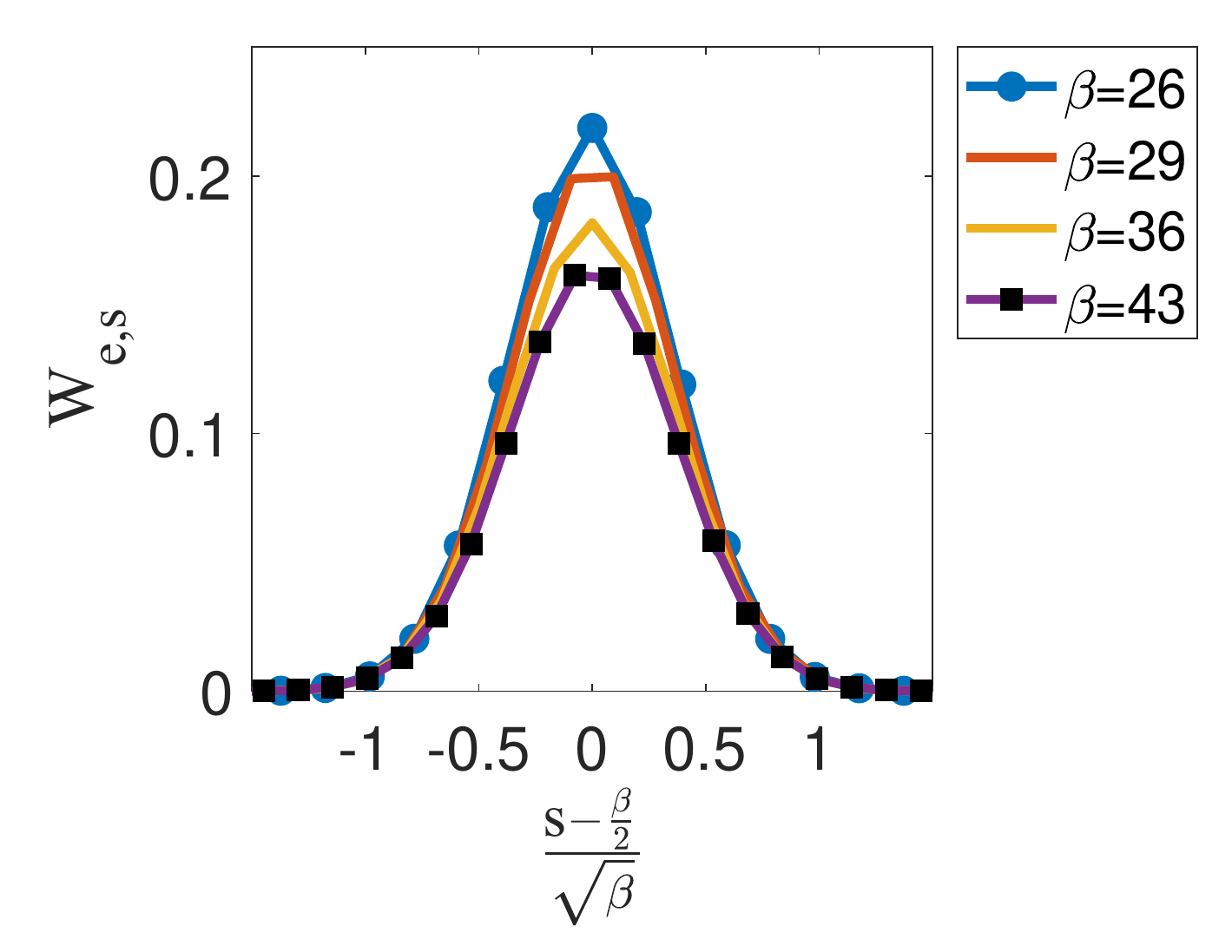}

\caption{Random Erd\H{o}s-Rényi graphs with $n$ vertices and $p=0.75$, for $n$ ranging between $9$ to $12$.}

\label{fig:ER}
\end{figure}


\begin{figure}[h!]
  \centering

  \includegraphics[width=0.7\textwidth]{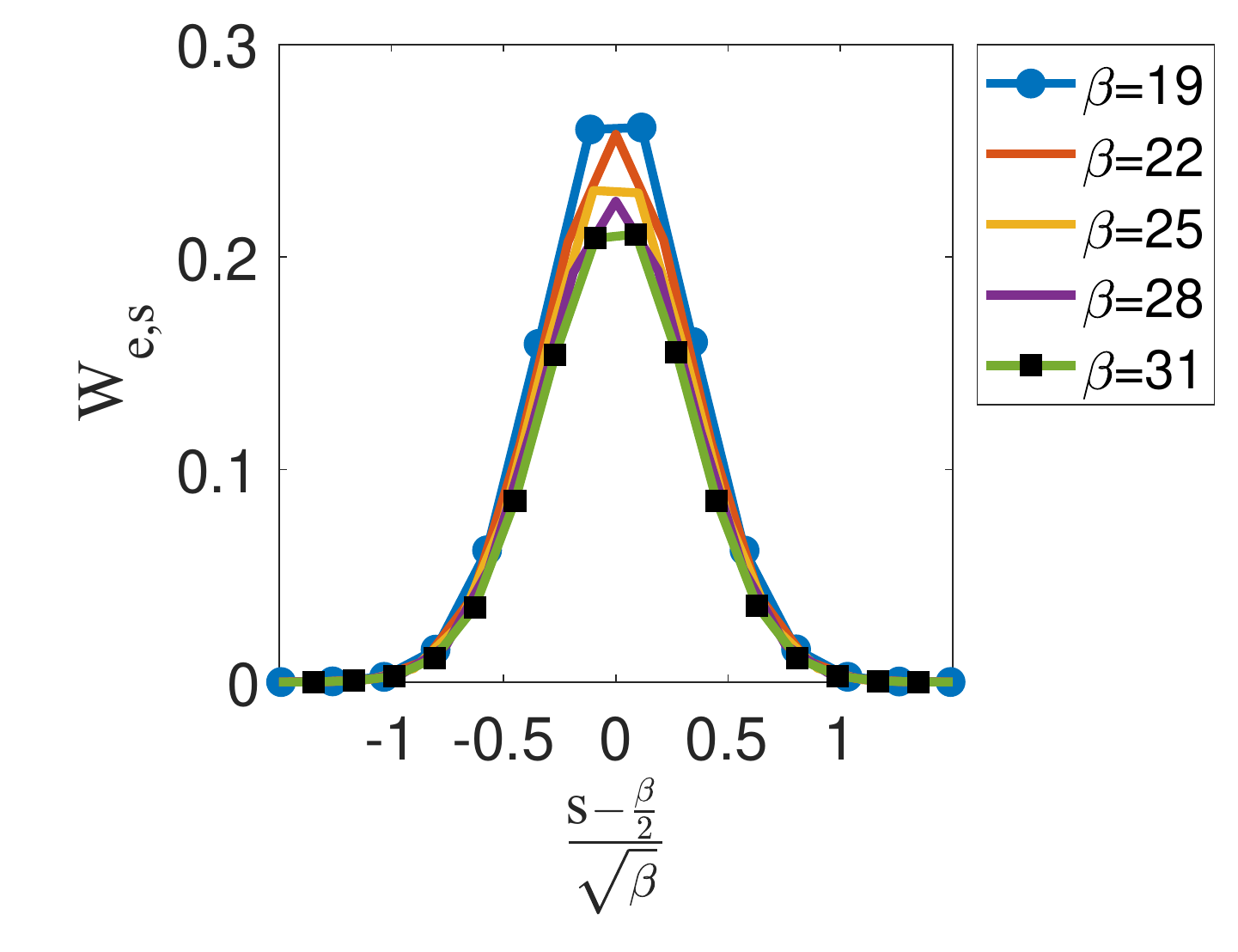}

\caption{Random $5$-regular graphs of $n$ vertices,\\
for $n\in\{12,14,16,18,20\}$. The $ \beta $ value was computed for each graph as discussed in Remark \ref{rem: random graphs}. }

\label{fig:Random regular}
\end{figure}
\clearpage

\bibliographystyle{siam}
\bibliography{additional,bib}

\end{document}